\documentclass[conference]{IEEEtran}
\IEEEoverridecommandlockouts

\usepackage{cite}
\usepackage{amsmath,amssymb,amsfonts}
\usepackage{graphicx}
\usepackage{textcomp}
\usepackage{xcolor}

\usepackage[cjk]{kotex}
\usepackage{graphicx}
\usepackage{amsmath}
\usepackage{color}
\usepackage{epsfig}
\usepackage{url}	
\usepackage{array,ragged2e}
\usepackage{mathtools}  
\usepackage{ifpdf}
\usepackage{endnotes}
\usepackage{xspace}	
\usepackage{multirow}	
\usepackage{tabularx}	
\usepackage{balance}	
\usepackage{marginnote}
\usepackage{enumitem}
\usepackage[utf8]{inputenc} 
\usepackage{flushend}
\usepackage{xcolor}
\usepackage{lipsum}
\usepackage{stfloats}
\usepackage{soul}

\usepackage{epstopdf}
\usepackage[cjk]{kotex}
\usepackage{graphicx}
\usepackage{amsmath}
\usepackage{color}
\usepackage{epsfig}
\usepackage{url}	
\usepackage{array,ragged2e}
\usepackage{mathtools}  
\usepackage[caption=false,font=footnotesize]{subfig}
\usepackage{ifpdf}
\usepackage{endnotes}
\usepackage{xspace}	
\usepackage{multirow}	
\usepackage{tabularx}	
\usepackage{balance}	
\usepackage{marginnote}
\usepackage{enumitem}
\usepackage{breakurl}	
\usepackage[utf8]{inputenc} 
\usepackage{flushend}
\usepackage{xcolor}
\usepackage{lipsum}
\usepackage{stfloats}
\usepackage{microtype}
\usepackage[htt]{hyphenat}
\usepackage{threeparttable}

\usepackage{hyphenat}
\usepackage{seqsplit}
\usepackage[T1]{fontenc}                
\usepackage[utf8]{inputenc}                 
\usepackage{threeparttable}
\usepackage{listings}
\usepackage{xcolor}
\lstdefinestyle{transaction}{
  basicstyle=\ttfamily\footnotesize,
  keywordstyle=\color{ForestGreen}\bfseries,
  keywordstyle=[2]\color{blue}\bfseries,
  commentstyle=\color{red},
  numbers=left,
  numbersep=5pt,
  frame=single,
  breaklines=true,
  morekeywords={Transaction},
  morekeywords=[2]{to, Read, Update},
  morecomment=[l]{//}
}

\usepackage{algorithm}
\usepackage{algpseudocode}

\usepackage{booktabs}
\usepackage{hhline}
\usepackage{comment}

\newcommand{\ie}{{\em i.e.}}
\newcommand{\eg}{{\em e.g.}}

\def\stx{{RCC}\xspace} 
\def\vanillamy{Vanilla MySQL\xspace}
\def\vanillapg{Vanilla PostgreSQL\xspace}

\def\tipl{{\em TIPL}\xspace}

\def\rccc{{RCC-C}\xspace}
\def\rcc{{RCC}\xspace}
\def\rccw{{RCC-W}\xspace}
\def\rccd{{RCC-D}\xspace}

\long\def\comment#1{}

\definecolor{blue}{rgb}{0.0,0.0,1.0}
\definecolor{purple}{rgb}{0.6,0.0,1.0}
\long\def\comment#1{}
\newcommand{\reminder}[1]{[\marginpar[\mbox{$\Longrightarrow$}]{\mbox{$\Longleftarrow$}}{\textcolor{red}{#1}}]}
\newcommand{\module}[1]{\textbf{#1\xspace}}

\usepackage{tikz}
\newcommand*\circled[1]{\tikz[baseline=(char.base)]{
            \node[shape=circle,fill=.,inner sep=0pt] (char) {\color{-.}\textsf\footnotesize #1};}}

\newtheorem{thm}{Theorem}

\newtheorem{lemma}[thm]{Lemma}

\usepackage{graphicx}
\usepackage[dvipsnames]{xcolor}

\usepackage{tikz}

\usepackage{tikz}
\definecolor{alizarin}{rgb}{0.82, 0.1, 0.26}
\definecolor{blue(ncs)}{rgb}{0.0, 0.53, 0.74}
\definecolor{black}{rgb}{0, 0, 0}

\newcommand{\jhpark}[1]{\textcolor{OliveGreen}{({\bf jh:} #1)}}
\newcommand{\swlee}[1]{\textcolor{red}{({\bf sw:} #1)}}
\newcommand{\hjin}[1]{\textcolor{orange}{({\bf hjin:} #1)}}
\newcommand{\msj}[1]{\textcolor{brown}{({\bf msj:} #1)}}
\newcommand{\fix}[2]{ \noindent\textcolor{blue}{{\bf #1:} #2}}
\long\def\revised#1{{\color{black}#1}}
\long\def\revisedmsj#1{{\color{black}#1}}
\long\def\revisedhj#1{{\color{black}#1}}
\usepackage{soul}
\usepackage{hyperref}

\definecolor{trimpink}{HTML}{C2185B}
\newif\ifshowtrim\showtrimtrue

\def\BibTeX{{\rm B\kern-.05em{\sc i\kern-.025em b}\kern-.08em
    T\kern-.1667em\lower.7ex\hbox{E}\kern-.125emX}}
\begin{document}


\title{Optimistic Deadlock Handling with Redo Log-Based Speculative Versioning under High Contention\\
}

\author{
\IEEEauthorblockN{
Hyejin Yoo\textsuperscript{\textdagger*},
Seongjae Moon\textsuperscript{\textdagger*},
Jonghyeok Park\textsuperscript{\textdaggerdbl},
Viktor Leis\textsuperscript{\S},
Sang-Won Lee\textsuperscript{\textdagger}
}

\IEEEauthorblockA{
\textsuperscript{\textdagger}Seoul National University
\textsuperscript{\textdaggerdbl}Korea University
\textsuperscript{\S}Technische Universit\"{a}t M\"{u}nchen\\
\{hjin.u, seongjae.moon\}@snu.ac.kr,
jonghyeok\_park@korea.ac.kr,
leis@in.tum.de,
swlee69@snu.ac.kr\\
}
}

\maketitle
\begingroup
\renewcommand{\thefootnote}{*}
\footnotetext{These authors contributed equally.}
\renewcommand{\thefootnote}{}
\footnotetext{https://github.com/FlashSQL/RCC.git}
\endgroup

\begin{abstract}
Modern OLTP engines rely on multi-versioning to eliminate read-write conflicts, yet their concurrency is severely limited by write-write conflicts.
With \emph{immediate, in-place updates}, only one transaction updates a record at a time, and the others wait for it to commit or abort.
Conflicting transactions are thus \emph{serialized}.
We propose \rcc, a \textbf{R}edo log-based \textbf{C}oncurrency \textbf{C}ontrol scheme for high write contention.
A transaction updates a record \emph{out of place} by creating a speculative write version from its redo log entry.
With multiple speculative versions per record, conflicting transactions are \emph{pipelined} without waiting for predecessors to commit.
\rcc guarantees serializability using lock-based access-time conflict ordering and dependency graph tracking.
Updates are installed \emph{lazily} at commit, so a transaction aborts by discarding its speculative versions.
Based on this, \rcc introduces two novel techniques.
\revised{
\emph{Optimistic deadlock handling} decouples deadlock detection from lock acquisition: it grants locks without checking for deadlocks and detects cycles asynchronously in the dependency graph.}
\emph{Columnar \rcc} (\rccc) removes record-level false conflicts with column-level redo logs.

We implemented \rcc on MySQL, PostgreSQL, and DBx1000.
In TPC-C at 128 threads on a 128-core server, \rcc improves throughput by up to 3.3$\times$ over the state-of-the-art schemes, and \rccc improves it further by avoiding false deadlocks and aborts.
On high-contention YCSB, \rcc remains robust up to 128 threads, whereas other schemes degrade beyond 32.

\comment{
In this paper, we propose to utilize redo logs as version store for higher concurrency, to support multiversion from redo recovery. The existing multiversion concurrency (MVCC) mechanism such as SI has exploited the {\em past record version} (\ie, snapshot) to resolve the RW and WR conflict, increasing the concurrency significantly. However, WW conflict 미해결. 모든 MVCC system들이 {\em current record version}에 대해 W-W blocking 발생 (first committer/writer winner서도?). Thus, when run against update-intensive OLTP workloads such as TPC-C, WW conflicts severely hinder the concurrency. To address the remaining ww-conflict-induced concurrency bottleneck, we suggest a novel xx scheme. It uses the redo logs used for the purpose of recovery as {\em future record versions} and 을 해결.

\begin{itemize}
    \item  기존 MVCC mechanisms (O2N, N20 모두 MySQL, Postgres) 확장해서 (recovery extension 포함) shadow tx 지원 방법  
    \item  관련한 precommit 및 group commit 제시
    \item light-weight cascade rollback by delayed update + hybrid deadlock prevention to further reduce the peril of exponential deadlock with larger connection 
    \item  redo log 특성 활용해서 columnar CC를 통해 false deadlock 해결 및 false/불필요 abort 회피/최소화
\end{itemize}
}
\comment{
=== ChatGPT version

Traditional database systems rely on redo logs primarily for crash recovery and durability, while concurrency control is handled separately through mechanisms such as locking or multi-version concurrency control (MVCC) with undo-based version chains. This separation often introduces bottlenecks under high contention, as write–write conflicts still force blocking and version management incurs significant overhead.

In this paper, we propose Redo Logs as Versions (RLV), a novel design that reinterprets redo logs as the primary version store. Instead of maintaining undo records or memory-resident version chains, each committed modification is directly preserved in an append-only redo log, which naturally provides immutable historical versions. Transactions read from consistent snapshots defined over log positions, while writers append new versions without interfering with concurrent readers. This log-as-version model unifies recovery and concurrency control, reduces the need for page-level locking, and allows higher concurrency by decoupling logical versioning from physical data organization.

We describe the system architecture of RLV, including log indexing for efficient snapshot access, garbage collection for obsolete versions, and integration with page caching and checkpointing. Our evaluation on TPC-C and TPC-H benchmarks shows that RLV achieves up to X\% higher throughput under high-contention workloads compared to traditional MVCC engines, while maintaining recovery guarantees with no additional overhead.

Our results demonstrate that rethinking redo logs as versioned data, rather than auxiliary recovery metadata, opens new directions for building simpler yet more scalable database systems.

}

\end{abstract}


\section{Introduction}
\label{sec:intro}

Concurrency control is a fundamental component of OLTP systems, as it determines how concurrent transactions interleave.
In traditional \emph{single-version} schemes such as two-phase locking (2PL), the database maintains only the most recent version of each object.
To prevent anomalies such as \emph{dirty reads}, \emph{non-repeatable reads}, and \emph{lost updates}, it pessimistically locks the current version~\cite{COWBOOK}.
Thus, \emph{read-write (RW)}, \emph{write-read (WR)}, and \emph{write-write (WW)} conflicts block transactions.
In particular, an updating transaction holds an exclusive lock on the current version until commit, so WW-conflicting transactions are strictly serialized, severely limiting concurrency.

Multi-version concurrency control (MVCC) alleviates this limitation.
By maintaining committed versions of an object as well as its current version, MVCC
resolves RW and WR conflicts so that reads and updates do not block each other ~\cite{Reed83}.
Many MVCC systems, including MySQL and Oracle, reconstruct past committed versions on demand using \emph{undo logs}, which were originally designed for recovery~\cite{BHG}.

Existing MVCC systems eliminate RW and WR conflicts using committed past versions and the current version, but they still handle WW conflicts \emph{pessimistically}.
Only one transaction \emph{at a time} can update the current version of an object, holding the write lock until commit or abort, so conflicting updates are \emph{serialized}.
To quantify this limitation, we run TPC-C on MySQL and PostgreSQL with up to 128 clients. As detailed in Section~\ref{sec:background}, throughput quickly plateaus \revised{(\autoref{fig:tpcc-performance}(a))} and CPU utilization remains low on a 128-core server.
These results suggest that WW conflicts are the primary bottleneck preventing production DBMSs from scaling on \revised{multi}-core servers. 
\revised{High WW contention in financial transaction processing, where many transactions update the same account balances, motivates specialized systems such as TigerBeetle~\cite{tigerbeetle}. \rcc targets such workloads within general-purpose DBMSs.
As the parallelism of multi-core processors increases, designing concurrency control for highly concurrent execution becomes important~\cite{dbx1000}.}

Meanwhile, many OLTP workloads share characteristics that make WW blocking avoidable: most transactions commit, and they update conflict-causing records early, typically at most once (Section~\ref{sec:background}). Under these conditions, exposing uncommitted writes to subsequent transactions and pipelining their execution need not cause anomalies~\cite{PWV}.
\comment{
}
\comment{
    \revised{
    Prior speculative write visibility schemes, such as Bamboo~\cite{BAMBOO} and Rebirth-Retire~\cite{rebirth_retire}, reduce WW blocking by allowing transactions to access uncommitted updates. However, such early write visibility introduces more dependencies among conflicting transactions, incurring more cascade aborts and complicating deadlock handling. This problem makes existing schemes to severely degrade their concurrency beyond 32 threads on high write contention workload, which is mainly due to two reasons. 
    First, both schemes update the current version in the data page before commit, so aborting a transaction requires rolling back its updates in reverse order to preserve consistency~\cite{Bamboo-impl,txsql}.
    \revisedhj{Second and more critically, both employ pessimistic deadlock handling to prevent dependency cycles from forming before allowing speculative reads. This interrupts transaction execution and offset the pipelining benefit of speculative write visibility.}
    }
    \swlee{Second, both employ \emph{pessimistic} deadlock handling: assuming that granting a lock is likely to cause a deadlock, they synchronously prevent dependency cycles from forming on the critical path of lock acquisition, using variants of timestamp-based wound-wait.\hjin{표현이 조금 강한 것 같은데, is likely to cause 를 may introduce로 바꾸면 어떨까요?} The timestamp management overhead grows with the number of threads conflicting on a record, so beyond 32 threads (Figure~\ref{fig:xxx}) concurrency degrades sharply, severely limiting the potential of speculative write visibility.} \hjin{timestamp management overhead라고 처음부터 하면, pessimistic vs optimistic의 차이보단 timestamp의 차이로 받아들일 수 있지 않을까요? 차라리 deadlock handling overhead 또는 pessimistic appoach 등이 어떨까요?}
    These two factors limits the speculative write visibility to realize its full potential. 
    
    \swlee{\textcolor{blue}{\bf 현재본:} Prior speculative write visibility schemes, such as Bamboo~\cite{BAMBOO} and Rebirth-Retire~\cite{rebirth_retire}, reduce WW blocking by allowing transactions to access uncommitted updates. However, such early write visibility introduces more dependencies among conflicting transactions, incurring more cascading aborts and complicating deadlock handling. As a result, existing schemes severely degrade in concurrency beyond 32 threads under high write contention (Figure~\ref{fig:xxx}), mainly for two reasons. 
    First, existing speculative write visibility systems employ \emph{pessimistic} deadlock handling: assuming that granting a lock may introduce a deadlock, they synchronously prevent dependency cycles from forming on the critical path of lock acquisition, using variants of timestamp-based wound-wait~\cite{BAMBOO,rebirth_retire}.
    Second, cascading aborts are costly: both schemes update the current version in the data page before commit, so aborting a transaction requires rolling back its updates in reverse order to preserve consistency~\cite{Bamboo-impl,txsql}. As the number of threads grows, more transactions join a dependency chain, and their cascading aborts are serialized one after another.
    The overhead of this deadlock handling grows with the number of threads conflicting on a record, interrupting transaction execution and offsetting the pipelining benefit of speculative write visibility.
    These two factors prevent speculative write visibility from realizing its full potential.}
} 
\revised{
Speculative write visibility removes WW blocking by letting transactions access uncommitted updates, but it shifts the bottleneck to deadlock handling, which becomes severe under high contention.
Prior schemes such as Bamboo~\cite{BAMBOO} and Rebirth-Retire~\cite{rebirth_retire} must prevent dependency cycles synchronously on the critical path of lock acquisition.
As a result, these schemes severely degrade beyond 32 threads under high write contention (\autoref{fig:tpcc_varying_clients_threads}) for two reasons.
First, \revisedhj{aborts} are costly: both schemes update the current version in the data page before commit, so aborting a transaction requires rolling back its updates to preserve consistency~\cite{bamboo-inplace,bamboo-rollback,rebirth-inplace}.
With more threads, more transactions join a dependency chain, \revisedhj{increasing the cost of cascading aborts}.
Second, both systems employ \emph{pessimistic} deadlock handling: assuming that granting a lock may form a dependency cycle, they prevent it synchronously on the critical path of lock acquisition using variants of timestamp-based wound-wait~\cite{BAMBOO,rebirth_retire}.
This overhead grows with the number of threads conflicting on a record (\autoref{fig:deadlock_handling}), 
offsetting the pipelining benefit of speculative write visibility.
}
\comment{
To address this, we propose \rcc, a redo log-based concurrency control scheme for high write contention. 
First, \rcc resolves deadlocks after cycles have formed, rather than preventing their formation during lock acquisition, thereby minimizing interruptions to transaction execution and avoiding unnecessary aborts. By detecting these cycles asynchronously, \rcc identifies deadlocks soon after they form, limiting the number of cascading aborts.
Second, \rcc leverages redo logs as speculative write versions to enable lightweight rollback. 
A transaction creates a speculative write version using its redo log entry, leaving the current version intact until commit. 
By maintaining a private speculative version of each accessed record per transaction, \rcc enables transactions conflicting on the same record to be \emph{pipelined}, each proceeding on the latest speculative version without waiting for predecessors to commit.

Beyond deadlock handling, \rcc's redo log-based design naturally enables a second contribution: \emph{Columnar \rcc} (\rcc-C). Record-level conflict tracking causes \emph{false} WW conflicts when transactions update disjoint columns of the same record, forcing unnecessary waits and aborts. \rccc eliminates these by tracking modifications at column granularity as redo log deltas, further improving throughput and latency.
}

\revised{
To fully realize the potential of speculative write visibility, we propose \rcc, a redo log-based concurrency control scheme for high write contention.
\rcc leverages redo logs as speculative write versions to enable lightweight rollback. A transaction creates a speculative version from its redo log entry, leaving the current version in the data page intact until commit, so an abort simply discards its speculative versions.
Unlike existing schemes, \rcc thus cascades aborts in parallel, without undoing \revisedhj{in-place} page updates.
By keeping a private speculative version of each accessed record per transaction, \rcc \emph{pipelines} transactions conflicting on the same record, each proceeding on the latest speculative version without waiting for its predecessors to commit.
}

\revised{
To further improve the performance of its speculative write versioning, \rcc introduces two novel techniques. First, assuming that deadlocks are rare, \rcc adopts \emph{optimistic deadlock handling}: it skips deadlock checks at lock acquisition
and allows conflicting updates to proceed on speculative versions. Its \emph{asynchronous deadlock detector (ADD)} then detects cycles in the transaction dependency graph off the execution path, aborting only true deadlocks shortly after they form. 
By taking deadlock handling off the critical path, \rcc advances speculative write visibility one step further, as lightweight rollback keeps the occasional abort cheap.
Second, \emph{columnar \rcc} (\rccc) eliminates record-level false conflicts by leveraging column-level redo logs, thus avoiding unnecessary lock waits, false aborts, and false deadlocks.
}
\comment{
\revised{\rcc creates speculative versions in lock acquisition order and maintains a dependency graph to enforce commit ordering and cascade aborts among conflicting transactions. 
\rcc supports all standard isolation levels by controlling the scope of dependency tracking. Under \texttt{Read Committed} and \texttt{Repeatable Read}, \rcc tracks only WW dependencies from speculative writes to enforce commit ordering, detect deadlocks, and cascade aborts to dependent transactions. 
Under \texttt{Serializable}, \rcc additionally tracks RW anti-dependencies~\cite{SI:TODS09, SSIpostgres}, allowing writer commits without waiting for those readers. \rcc alswo retains RW edge after writer commits to detect cycles to prevent anomalies such as write skew~\cite{ansi-isolation}.}
} 

\revised{
To support these techniques, conflicting transactions create speculative versions in lock acquisition order, and \rcc maintains a dependency graph among them to enforce commit ordering, detect deadlocks, and cascade aborts.
\rcc supports all standard isolation levels by controlling the scope of dependency tracking.
Under \texttt{Read Committed} and \texttt{Repeatable Read}, \rcc tracks only WW dependencies from speculative writes.
\revised{Under \texttt{Serializable}, \rcc additionally tracks WR dependencies and RW anti-dependencies~\cite{SI:TODS09, SSIpostgres} to prevent anomalies such as write skew~\cite{ansi-isolation}. Because RW anti-dependencies are used only for cycle detection and impose no commit order, a writer does not wait for earlier readers of the record it updates to commit.}
}

\comment{
Yet existing MVCC systems leave this opportunity unexploited because the trade-off between early write visibility and correctness has long been resolved in favor of pessimism. The prevailing practice of \revised{\emph{immediate, in-place} updates} causes only one transaction to update a record at a time, forcing conflicting transactions to be \emph{serialized}. We argue that this in-place update convention is the root cause of pessimistic WW handling~\cite{Gray81}. Anomalies such as dirty reads and non-repeatable reads do not stem from exposing uncommitted updates, but from overwriting a shared current version. If active transactions' updates on the same record are maintained as separate versions and ordered by lock acquisition time, they can proceed concurrently without compromising correctness.
\comment{their ordering is controlled, conflicting transactions} 

Based on this observation, we propose \rcc, which resolves WW conflicts by leveraging \emph{redo logs} as speculative write versions \revised{under contented workloads}.\hjin{under high contention?}
When updating a record, rather than immediately overwriting the current version in the data page, a transaction creates a speculative write version using its redo log entry, leaving the current version intact until commit. After the update, the transaction immediately releases the lock, allowing other transactions to access the record without waiting. Upon commit, the speculative version is installed into the data page. This \emph{lazy update policy} of \rcc affords \emph{lightweight} rollback: since no changes have yet been installed to the data page, a transaction aborts simply by discarding its speculative versions. In contrast, existing speculative write versioning systems such as Bamboo and Rebirth-Retire~\cite{BAMBOO, Bamboo-impl,rebirth_retire, rebirth-retire-code} eagerly install updates to the data page and thus require undo on abort.

By maintaining a private speculative version of every accessed record per transaction, \rcc enables transactions conflicting on the same record to be \emph{pipelined}, each proceeding on the latest speculative version without waiting for predecessors to commit.
\rcc creates speculative versions in lock acquisition order and maintains a dependency graph to enforce commit ordering and cascade aborts among conflicting transactions. 
\rcc supports all standard isolation levels by controlling the scope of dependency tracking. Under \texttt{Read Committed} and \texttt{Repeatable Read}, \rcc tracks only WW dependencies from speculative writes to enforce commit ordering, detect deadlocks, and cascade aborts to dependent transactions. Under \texttt{Serializable}, \rcc additionally tracks RW anti-dependencies~\cite{SI:TODS09, SSIpostgres} and also detects cycles to prevent anomalies such as write skew~\cite{ansi-isolation}.

Although \rcc eliminates lock-induced blocking, deadlock handling remains critical for correctness and performance. Existing speculative versioning systems rely on \emph{pessimistic} deadlock prevention: wound-wait in Bamboo~\cite{BAMBOO}, which may cause false aborts, and wound-retire in Rebirth-Retire~\cite{rebirth_retire}, which incurs per-lock coordination overhead. These prevention-based schemes interrupt transaction execution on every lock acquisition, preventing speculative versioning from fully realizing its pipelining potential. To address this, \rcc takes an \emph{optimistic} approach to deadlock handling. Rather than preventing deadlocks pessimistically at every lock acquisition, \rcc proposes \emph{commit-time deadlock detection}: cycle detection is deferred entirely to commit time and invoked exactly once per transaction, keeping the execution path free of any deadlock handling logic. \rcc further introduces an \emph{asynchronous} optimization that detects cycles between the conflict point and commit time, minimizing unnecessary pre-commit waits. Transactions are aborted only upon true deadlock cycles, eliminating false aborts entirely. This optimistic design enables \rcc to outperform the prevention-based deadlock handling in existing speculative versioning systems at scale (Section~\ref{sec:eval}): \rcc scales to 128 cores while competing schemes
--- Bamboo and Rebirth-Retire ~\cite{BAMBOO,rebirth_retire} --- degrade past 32 concurrent threads, dropping below even vanilla 2PL.
\comment{
While speculative versioning enables transactions to be pipelined, realizing its full potential requires two novel mechanisms. First, we propose \emph{optimistic commit-time deadlock detection}, a novel deadlock handling scheme designed for speculative versioning, assuming that true deadlocks are not frequent in practice~\cite{TPBOOK}. \hjin{deferring cycle detection으로? speculative read cycle은 기존의 deadlock과 다르므로 not frequent라고 하면 안 될 것 같습니다(디자인장에서는 해당 내용 삭제했습니다)} Although \rcc eliminates lock-induced blocking, deadlock handling remains critical for both correctness and performance. 
Existing speculative versioning systems rely on prevention-based deadlock handling: {\em wound-retire}~\cite{BAMBOO}, which may cause false aborts, or {\em rebirth-retire}~\cite{rebirth_retire}, which reassigns timestamps upon conflict and incurs per-lock overhead.
\textcolor{blue}{In contrast, \rcc allows each transaction to execute to completion without any interruption from deadlock-handling logic, deferring cycle detection to commit time and invoking it \emph{exactly once} per transaction. Transactions are aborted only when a true deadlock cycle is detected, eliminating both false aborts~\cite{BAMBOO} and per-lock detection overhead from the execution path~\cite{rebirth_retire}.
This design is what enables \rcc to outperform existing deadlock prevention schemes at scale: as shown in Section~\ref{sec:eval}, on a 128-core machine, \rcc throughput plateaus gracefully beyond 32 concurrent threads, whereas competing schemes such as Bamboo, Rebirth-Retire, and Silo~\cite{BAMBOO,rebirth_retire,silo} suffer continuous throughput degradation past that point --- eventually falling below even vanilla two-phase locking at 128 cores.} \hjin{차라리 앞 문단이랑 합쳐서 architecture장처럼 전개하는 게 좋을 것 같습니다}
}
}

\comment{
\fix{R2D1}{Novelty (vs. Bamboo)} 
\revised{1) speculative write versioning의 full potential을 달성하기 위해 deadlock prevention mechanism이 critical하고 이를 위해 \emph{delayed deadlock detection at commit} propose. 기존 immediate Bamboo wound-wait나 deadlock detection은 transaction latency를 severely hampering함. 2) Redo logs carrying the column-level update information enable  RCC to resolve false/unnecessary blockings, boosting the concurrency further}

\comment{
\revised{(우선 단순하게 구현측면) Speculative write versioning을 MVCC까지 지원하는 오픈소스 RDBMS(MySQL, PostgreSQL)에 seamless하게 구현할 수 있음으로 보였고, 특히 redo log와 결합해서(?) undo-less recoverability 달성.} 

\revised{((Commit-time deadlock detection to realize the full potential of speculative write version)) The ideal goal of speculative write versioning is to pipeline conflicting transactions through early write visibility, eliminating blocking during execution. However, removing lock-induced blocking merely shifts the bottleneck to deadlock handling. Bamboo~\cite{BAMBOO} adopts wound-wait for deadlock prevention, which must synchronously wait for a lower-priority transaction to abort on every lock conflict — reintroducing blocking into the critical execution path and preventing speculative write versioning from realizing its full parallelism potential. \stx proposes \emph{commit-time deadlock detection}: since speculative unlock structurally confines all deadlock cycles to the commit stage, \stx defers detection entirely to that point. In-flight transactions proceed to completion without any execution-time interruption, and only actual deadlock cycles trigger an abort at minimal cost via lightweight TLA discard — eliminating both the false aborts of wound-wait and its synchronous blocking overhead. As shown in Section~\ref{sec:eval}, this yields up to 4× higher throughput over Bamboo under SR isolation, demonstrating that deadlock handling is the critical factor in realizing the full potential of speculative write versioning.}

\revised{(( Redo log enabled columnar CC)) However, transactions such as \texttt{New-Order} and \texttt{Payment} in TPC-C modify disjoint columns of the same row and do not truly conflict, yet they are treated as WW conflicts because they update the same record, forming unnecessary dependency edges. We observe that such false conflicts force transactions to wait for the preceding transaction to commit and cause cascading aborts, limiting concurrency. Since redo logs capture column-level deltas, \stx can compare the modified columns of transactions and avoid forming dependency edges when their modifications are disjoint, eliminating false conflicts and further improving concurrency.}
}
}

This paper makes four contributions.
\comment{
\circled{1} We identify low abort rates, rare re-updates, and early conflict points as the characteristics of OLTP workloads that make speculative write visibility effective, and show that the in-place update convention leaves them unexploited in existing concurrency control schemes.
\circled{2} We pinpoint two factors that limit existing speculative write visibility schemes.
In-place updates make cascading rollbacks costly, and timestamp-based deadlock handling on the lock acquisition path incurs overhead that grows with contention.
\circled{3} We propose \rcc, a redo log-based concurrency control scheme with lightweight abort through lazy update installation, optimistic deadlock handling with asynchronous detection, and a columnar extension, \rccc, that eliminates record-level false WW conflicts.
\circled{4} We implement \rcc on MySQL, PostgreSQL, and DBx1000 and evaluate it against Bmaboo~\cite{BAMBOO}, and Rebirth-Retire~\cite{rebirth_retire} on TPC-C and YCSB under high write contention.
\rcc is open source at \url{our-public-github-link}.
}
\begin{itemize}[leftmargin=8pt, topsep=0pt, partopsep=0pt]
\item We identify low abort rates, rare re-updates, and early conflict points as characteristics of OLTP workloads that make speculative write visibility effective, and show that the in-place update convention leaves this opportunity unexploited in existing concurrency control schemes.
\item  We pinpoint two factors that limit the scalability of existing speculative write visibility schemes. In-place updates make cascading rollbacks costly, and the overhead of timestamp-based deadlock handling on the lock acquisition path grows with the number of conflicting threads.
\item We propose \rcc, a redo log-based concurrency control scheme with lightweight rollback and parallel cascading aborts through lazy update installation, optimistic deadlock handling with asynchronous detection, and a columnar extension, \rccc, that eliminates record-level false conflicts.
\item We implement \rcc on MySQL, PostgreSQL, and DBx1000 and evaluate it against Bamboo~\cite{BAMBOO} and Rebirth-Retire~\cite{rebirth_retire} on TPC-C and YCSB under high write contention.
\end{itemize}
\section{Background and Motivation}
\label{sec:background}

In this section, we review how MVCC resolves RW and WR conflicts yet leaves WW ones unaddressed. We next identify opportunities for speculative write visibility and trace pessimistic locking to in-place updates, motivating \rcc. 
We then position \rcc among prior speculative versioning schemes.

\comment{
\subsection{Multi-Version Concurrency Control}
\label{sec:cc} 
Multi-version concurrency control (MVCC) improves concurrency by maintaining multiple physical versions of a single logical data item, allowing read and write operations to proceed without blocking each other~\cite{mysql,postgres,oracle,mvcc-eval}.
The key idea of MVCC lies in allowing read operations to access an appropriate version, while write operations create new versions without blocking concurrent readers.
In the past decade, a variety of strategies have been proposed for integrating MV with concurrency control, including timestamp ordering (MVTO), optimistic concurrency control (MVOCC), and two-phase locking (MV2PL)~\cite{mvcc-eval}.

\module{Version Construction}
To access an appropriate version, MVCC systems must construct a version that is visible to the reader.
Existing MVCC Systems~\cite{mysql,postgres} organize version chains for physical data in an append-only manner, but they differ in the ordering of versions.
MySQL adopts Newest-to-Oldest ordering, where the current committed version resides in the primary data page, and past versions are chained backward through undo log entries~\cite{mysql-mv,Oracle95,DIVA}.
PostgreSQL adopts Oldest-to-Newest ordering, where each version is stored as a separate physical tuple and newer versions are appended to the tail of the chain~\cite{postgres-mv,DIVA}. 

\module{Eliminating RW and WR Blocking.}
 \hjin{이 것도 모듈 따로 쓰지 말고 WW 내용으로 넣자}
By maintaining multiple versions, MVCC allows readers and writers to operate independently without blocking each other.
In N2O systems, when a reader requires a version that is currently being modified, it does not block on the writer's lock. Instead, the system reconstructs a past committed version by traversing the undo log chain, allowing the readers to proceed immediately.
Writers modify the current version in place without waiting for readers, since any concurrent reader can always reconstruct an appropriate past version from the undo log.
In O2N systems, readers and writers operate on physically separate tuples.
A reader traverses the version chain and selects a tuple whose {\tt xmin} and {\tt xmax} indicate visibility to its snapshot, regardless of whether a concurrent writer is appending a new version~\cite{postgres-mv}.
Writers simply append new tuples to the chain without modifying existing ones, leaving past versions intact for concurrent readers. 

\module{WW Blocking Limits Concurrency.}
While MVCC eliminates RW and WR blocking, WW conflicts remain unresolved. When concurrent transactions attempt to update the same record, they must be serialized regardless of the version storage organization.
In N2O systems, a writer must acquire an exclusive lock on the current version and hold it until commit.
Any subsequent writer attempting to update the same record is blocked until the first writer commits and releases the lock.
This blocking is unavoidable because N2O systems perform in-place updates on the current version, and allowing concurrent modifications would result in lost updates~\cite{ansi-isolation}.
WW conflicts also persist in O2N systems. 
\comment{Before appending a new version, R4.Q3==>}
Before creating a new tuple version, a writer must check if another transaction is updating the same record.
In that case, the second writer must wait for the first to complete.
While MVCC decouples readers and writers, no such mechanism exists for concurrent writers, so WW conflicts limit concurrency and scalability.
}

\subsection{WW Conflicts Limit Concurrency}
\label{sec:cc} 
\comment{ \textcolor{purple}{
Multi-version concurrency control (MVCC) improves concurrency by maintaining multiple physical versions of a single logical data item, allowing readers and writers to operate independently without blocking each other~\cite{mysql,postgres,oracle,mvcc-eval}.
Read operations can access an appropriate version, while write operations create new versions without blocking concurrent readers.
Existing MVCC Systems~\cite{mysql,postgres} organize version chains for physical data in an append-only manner, but they differ in the ordering of versions.
MySQL adopts Newest-to-Oldest (N2O) ordering, where the current committed version resides in the primary data page, and past versions are chained backward through undo log entries~\cite{mysql-mv,Oracle95,DIVA}.
PostgreSQL adopts Oldest-to-Newest(O2N) ordering, where each version is stored as a separate physical tuple and newer versions are appended to the tail of the chain~\cite{postgres-mv,DIVA}.  }
}

Multi-version concurrency control (MVCC) maintains multiple versions of each data item, so that readers and writers do not block each other~\cite{ mysql,postgres,oracle,mvcc-eval}. Existing MVCC systems differ in version chain ordering. MySQL adopts Newest-to-Oldest (N2O) ordering, where the current version resides in the data page and past versions are chained backward through undo log entries~\cite{mysql-mv,Oracle95,DIVA}. PostgreSQL adopts Oldest-to-Newest (O2N) ordering, where each version is a separate tuple and newer versions are appended to the chain~\cite{ postgres-mv,DIVA}.

While MVCC eliminates RW and WR conflicts, WW conflicts remain. In N2O systems, a writer holds an exclusive lock on the current version until commit, blocking concurrent writers on the same record. This is unavoidable: in-place updates on the current version would otherwise cause lost updates~\cite{ansi-isolation}. In O2N systems, a writer must check whether another transaction is updating the same record before appending a new version, again forcing a wait.
MVCC thus decouples readers and writers, but offers no such mechanism among writers, leaving WW conflicts as the remaining bottleneck.

\comment{
When a transaction reads a data item, it acquires a shared lock on the appropriate version, and when it writes, it creates a new version while holding an exclusive lock on the current version. 
This allows readers and writers to proceed concurrently, as readers can access older versions without blocking writers, while still ensuring serializability through the locking protocol.
Next, we elaborate on how MVCC (MV2PL) ensures the correctness in isolation across concurrent transactions.

\module{Recoverability.}
Recoverability ensures that a transaction may commit only if all transactions from which it has read data have also committed, thereby preventing inconsistencies that could arise from reading uncommitted changes. 
In MVCC, this property is typically achieved by restricting transactions to read only committed versions, thus naturally enforcing recoverability.

\module{Cascaded Rollback}
\jhpark{design 장 또는 correctness 장으로 이동;}
Cascading rollback refers to the phenomenon where the abort of a single transaction leads to a chain of aborts in other transactions that have read its uncommitted changes. By allowing transactions to access only committed versions, MVCC can avoid cascading rollbacks, resulting in more robust and predictable recovery behavior.

\module{Serializability}
MV2PL can provide serializability because all transactions see a serialization order of update transactions that is consistent with their timestamp order.
}
To quantify the concurrency penalty of WW conflicts, we run TPC-C with 1 warehouse on two open-source MVCC databases under their default isolation levels: MySQL~8.4.5~\cite{mysql845} (\texttt{Repeatable Read}) and PostgreSQL~16.2~\cite{postgres-162} (\texttt{Read Committed}), using the setup described in Section~\ref{sec:exp-setup}.
\comment{
The workload consists of \texttt{New-Order} and \texttt{Payment} transactions in a 1:1 ratio, representing update-intensive core of TPC-C. Experiments run on a server with an AMD EPYC 9754 CPU (128 cores, 2.25\,GHz), 256\,GB DRAM, and Ubuntu~22.04. \hjin{rebirth처럼 각주?: 1 The experimental setup is described in Section 5.1.}
}


\begin{figure}[t]
  \centering
  \subfloat[Throughput]{%
    \includegraphics[width=.48\columnwidth]{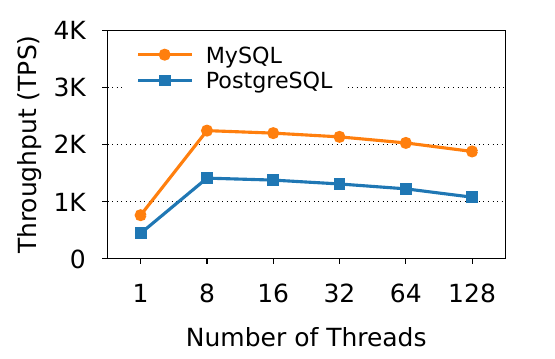}%
  }%
  \hfill
  \subfloat[Latency Breakdown]{%
    \includegraphics[width=.48\columnwidth]{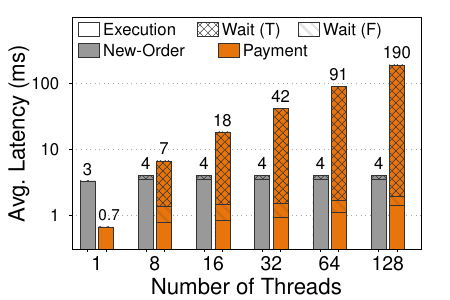}%
  }
  \caption{WW Conflicts Limit Concurrency (TPC-C, \revised{1 Warehouse})}
  \label{fig:tpcc-performance}
\end{figure}

\comment{
    Figure~\ref{fig:ww-varying-th-tps} shows the resulting transaction throughput. Despite abundant CPU resources, as the number of concurrent connections exceeds eight, TPS plateaus and the CPU utilization remains very low (\ie, at most 1.7\% with 128 concurrent clients), 
    indicating that the system is not compute-bound but limited by synchronization overhead.
    To identify the source of this scalability bottleneck, we analyze transaction latency and its components. Figure~\ref{fig:ww-varying-th-breakdown} breaks down the average latency into lock wait time as the number of threads increases.
    
    \fix{R1D2}{Note that the lock wait time measures the elapsed time that a transaction spends waiting for locks. Therefore, lock wait time also includes system-level overheads such as thread scheduling, context switching, and lock manager contention that arise from WW conflict-induced blocking.} 
    
    While the execution time remains largely unchanged across different concurrency levels, the time spent waiting for write locks grows substantially with the number of concurrent connections per warehouse. This increasing lock wait time directly translates into inflated transaction latency, particularly at the tail, as a growing fraction of transactions experience prolonged blocking.
    
    \swlee{False WW-conflicts between NO and PM: Columnar CC 복선} The Tuple-level locking creates false write-write dependencies between transactions that have no true data conflict at the column level, unnecessarily serializing concurrent updates (여기서 언급? 2.2 oppotunity로?) \jhpark{2.2에서 언급하는 게 좋을 것 같습니다.}
    
    \swlee{Latency 언급} (False WW conflicts의 결과?) Beside throughput, WW conflicts adversely  impact the latency of payment transaction surges as the number of concurrent threads increases. (1ms -> 128ms). 이유 설명 
    NO -> PM -> PM이 대기함으로써 PM의 latency가 길어지고, PM의 대기시간이 길어짐을 명시하면 Row level 충돌, Column level의 충돌이 존재함. (위를 보여줄 수 있는 그래프) 
}

    \autoref{fig:tpcc-performance}(a) shows transaction throughput under varying concurrency. Despite abundant CPU resources, TPS plateaus beyond eight connections, with CPU utilization remaining below 1.7\% at 128 clients --- indicating that the bottleneck is synchronization, not computation. To identify its source, we analyze transaction latency in MySQL.
    \autoref{fig:tpcc-performance}(b) breaks down average latency per transaction type into execution time and WW lock wait, further separating the latter into true conflicts (same columns) and false conflicts (different columns of the same record).  Note that WW lock wait measures total elapsed time during lock acquisition, including thread scheduling, context switching, and lock manager contention.

    While the execution time remains largely unchanged across concurrency levels, lock wait time grows substantially with the number of connections per warehouse.
    \texttt{New-Order} latency remains stable, but \texttt{Payment} latency grows sharply with concurrency because every \texttt{Payment} transaction updates the same warehouse record in the early phase of the transaction. Each \texttt{Payment} must wait for the \revisedhj{predecessor} to commit before acquiring the write lock, forming a serial chain of WW blocking. 
    In addition, record-level locking causes false conflicts between \texttt{Payment} and \texttt{New-Order} transactions that update disjoint columns of district records, unnecessarily serializing them and increasing the lock holding time. As a result, true and false WW conflicts prolong \texttt{Payment} latency from 0.7 ms to 190 ms at 128 clients.
\subsection{Opportunity for Speculative Write Visibility}\label{sec:oppor}
A promising solution to WW conflicts is speculative write visibility~\cite{PWV}, which exposes uncommitted writes to subsequent transactions.  Although prior speculative versioning systems~\cite{BAMBOO,rebirth_retire,IC3,MCC} have shown significant throughput gains over conventional locking under high contention,
two questions remain unexplored: 
what fraction of transaction execution time is held under WW-conflicting locks, and what characteristics of OLTP workloads make speculative write visibility effective. 
To answer these questions, we run six OLTP benchmarks on MySQL 8.4.5 under \texttt{Repeatable Read} with 128 clients and a scale factor of 1 for 30 minutes. \autoref{tab:oltp-characteristics} summarizes the results. Our analysis reveals that OLTP transactions exhibit two characteristics that enable speculative write visibility.

\begin{table}[t]
\setlength{\tabcolsep}{3pt}
\centering
\caption{Concurrency Characteristics
in OLTP Workloads}
\label{tab:oltp-characteristics}
\begin{threeparttable}
\begin{tabular}{lrrr}
\toprule
& \multicolumn{1}{c}{WW Wait$^a$}
& \multicolumn{1}{c}{Abort Rate} 
& \multicolumn{1}{c}{Re-update} \\
Benchmark
& \multicolumn{1}{c}{(\% of time)}
& \multicolumn{1}{c}{(\% of txns)}
& \multicolumn{1}{c}{(\% of txns)} \\
\midrule
TPC-C~\cite{tpcc}
& 93.9 & 1.6 & 0.05 \\
Wikipedia~\cite{wikipedia}
& 51.9 & 1.0 & 0.09 \\
YCSB-A ($\theta$=0.9)~\cite{ycsb}
& 50.6 & 0 & 0 \\
AuctionMark~\cite{auctionmark}
& $\sim\,0$ & 4.5 & 0 \\
TATP~\cite{tatp}
& $\sim\,0$ & 4.2 & 0 \\
Epinions~\cite{epinions}
& 0.04 & $\sim\,0$ & 0 \\
\bottomrule
\end{tabular}
\begin{tablenotes}[flushleft]
\item[]{\footnotesize $^a$ WW lock wait
time as a fraction of total TX exec.\
time.}
\end{tablenotes}
\end{threeparttable}
\end{table}

\comment{
    First, OLTP transactions are typically designed to commit successfully; aborts are the exception~\cite{TPBOOK}. Speculative write visibility introduces the risk of cascading aborts~\cite{PWV}: if a transaction aborts after exposing its uncommitted writes, all dependent transactions abort cascadingly. However, as shown in \autoref{tab:oltp-characteristics}, the abort rates remain low across all benchmarks. Thus, conservative blocking or abort-heavy designs miss opportunities to improve concurrency. 
    \hjin{user-intended abort가 낮아도 speculative read로 인한 cascading abort는 여전하다고 봐야하는 게 아닐지? 차라리 reupdate와 함께 이게 최종적으로 커밋될 버전일 확률이 높다로 가는 건?}
    \comment{YCSB-A and Epinions exhibit no aborts, TPC-C and Wikipedia show abort rates below 2\%, and AuctionMark and TATP have at most 5\% abort rates, \hjin{이미 구체적인 수치를 표로 제시해주고 있으니, 지면이 부족하다면 이 부분 줄여도 좋을 것 같습니다.} but their WW conflicts are negligible, so cascading aborts rarely occur.
    }
    
    Second, transactions typically update each record only once. In \autoref{tab:oltp-characteristics}, the re-update rates also remain below 0.1\% across all benchmarks. When a transaction re-updates a record, any concurrent transaction that has read its prior speculative version is invalidated and must be aborted. 
    \comment{
    However, since re-updates are rare in practice, such cascading aborts occur infrequently. The benefits of exposing speculative versions --- enabling concurrent transactions to proceed without waiting for predecessors to commit, thereby significantly improving concurrency while preserving serializability --- far outweigh the occasional cost of cascade abort.
    \msj{R1O1 리뷰 내용 반영(cascading aborts는 rare하다는 주장 제거.)
    Therefore, these low re-update rates offer an opportunity to improve concurrency by allowing subsequent transactions to access uncommitted writes.}
    }
    \revised{However, since re-updates are rare, there is an opportunity to expose speculative versions immediately after an update.}
}
First, OLTP transactions typically commit successfully~\cite{TPBOOK} and update each record just once. Exposing \revised{uncommitted writes} thus poses little risk: aborts propagate to dependent transactions only when a predecessor aborts or re-updates a record, invalidating any speculative read~\cite{PWV}.
As shown in \autoref{tab:oltp-characteristics}, both abort rates and re-update rates remain negligible across all benchmarks, indicating that a speculative version is highly likely to become the final committed value. This makes immediate exposure of speculative versions a safe and effective means of improving concurrency --- an opportunity that conservative blocking or abort-heavy designs leave unexploited.

Second, update-conflicting operations typically occur early in a transaction, so locks are held long after the conflict point, prolonging the blocking of subsequent transactions.
Such operations come early due to application logic, unique ID acquisition, or deadlock avoidance~\cite{Quro,MCC,Brook-2PL}.
Consequently, WW conflicts account for a large fraction of total transaction execution time. As shown in \autoref{tab:oltp-characteristics}, three benchmarks (TPC-C, Wikipedia, and YCSB-A) exhibit WW wait times exceeding 50\% of total execution time. Eliminating this blocking requires exposing uncommitted writes to subsequent transactions so that their executions can be pipelined.

Taken together, these characteristics reveal a clear opportunity to improve scalability in OLTP workloads through speculative write visibility. By leveraging them, update-conflicting transactions can proceed concurrently.
\rcc capitalizes on this opportunity to improve concurrency under WW conflicts.

\comment{
Understanding workload characteristics is essential for designing scalable concurrency control protocols. Although WW conflicts force MVCC systems to conservatively hold write locks until commit, thereby strictly serializing conflicting  transactions, realistic OLTP transactions exhibit three intriguing characteristics that enable opportunities for early write visibility. 

First, the impact of WW conflicts is further amplified by the fact that they arise at the very beginning of transaction execution. In many OLTP workloads, including TPC-C, write conflicts concentrate on a small set of hotspot records that are accessed immediately after a transaction starts. These early updates are not incidental; they are {\em intentionally} placed as synchronization points so as to prevent race conditions and reduce deadlocks~\cite{Quro}.
In TPC-C, for instance, every \textit{New-Order} transaction retrieves and increments the \texttt{D\_NEXT\_O\_ID} field in the \texttt{DISTRICT} table to obtain a new order identifier, while each \textit{Payment} transaction updates warehouse- and district-level aggregates such as \texttt{W\_YTD} and \texttt{D\_YTD}. These records act as global coordination points for transactions targeting the same warehouse or district. Similar early exclusive-access patterns are observed in other OLTP benchmarks, such as TPC-E~\cite{Quro}. 
Although these conflicting operations occupy only a small fraction of a transaction’s total execution time, their early placement has a disproportionate impact on concurrency. Once the synchronization phase completes, the remaining operations typically access disjoint tuples and proceed with little or no contention. However, because existing MVCC systems resolve WW conflicts by holding write locks until commit, short early conflicts are transformed into long blocking periods. As a result, transactions are serialized at the very beginning of execution, causing many concurrent transactions to wait for extended durations and further limiting scalability. If other transactions were allowed to access and build upon the results of these early synchronization-point updates—rather than being blocked until commit—conflicting transactions could be pipelined instead of serialized, leading to a far more scalable execution under contention.

Second, most OLTP transactions are expected to commit successfully: 
logical aborts are rare in production OLTP workloads, and transactions are, except for aborts caused by deadlocks or crash, typically programmed to complete once initiated~\cite{TPBOOK}. This high commit probability suggests that conservative blocking or abort-heavy designs waste substantial opportunities for more concurrency.

Third, transactions rarely update the same hot record more than once.
Particularly in terms of updates, OLTP transactions usually perform final writes rather than read–modify–write cycles on the same object. Even when select-for-update patterns are used, the interval between read and write is short, and subsequent accesses to the same object are uncommon.  This access-once behavior implies that, as long as commit order is preserved, exposing early write results does not inherently compromise correctness.

Taken together, these characteristics reveal a clear opportunity to improve scalability through early write visibility~\cite{PWV}. Given the high commit probability of OLTP transactions, exposing updates earlier can prevent conflicting transactions from waiting for long commit durations. Moreover, because write conflicts arise early and transactions typically access each object only once, cascading rollbacks are unlikely in practice~\cite{BHG}. As a result, early conflict resolution can be exploited to allow transactions to be pipelined rather than serialized, without incurring dirty reads or cascading aborts. As discussed later, R-MVCC capitalizes on these characteristics to improve concurrency under update contention.
}

\subsection{In-Place Updates as the Culprit}
\label{sec:culprit}

However, most existing schemes fail to exploit this ample opportunity for speculative write visibility.
We pinpoint as the main culprit the prevailing convention that active transactions apply their updates \emph{immediately and in place} to the \emph{current version} of a data object~\cite{pat-halland-cidr25}. When a transaction updates a data record, it must first acquire an exclusive write lock, then capture the redo log and modify the current record version immediately~\cite{COWBOOK,ansi-isolation}. The transaction should then hold the lock until commit so as to avoid isolation anomalies such as lost updates, dirty reads, and non-repeatable reads~\cite{ansi-isolation}. 

This convention causes three limitations. First, updating the \emph{current version} in place forces a transaction that updates a record to hold the write lock until commit, hiding its update from other transactions, although the transaction is highly likely to commit.
As a result, WW conflicts serialize update-conflicting transactions. 
Second, \emph{immediate modification of the current version} allows only one uncommitted writer at a time, so multiple uncommitted update states of the same record cannot coexist.
Lastly, when active transactions share and overwrite the current version, even speculative concurrency becomes unsafe: aborting a blocked transaction can trigger cascading rollbacks, and a transaction reading a record after updating it may observe another transaction's update, violating read-after-write consistency or repeatable-read semantics~\cite{ansi-isolation}.
Consequently, early conflicts are amplified into long blocking periods, and transactions that could proceed are serialized.

In summary, pessimistic write locking is not mandated by correctness, but rather an artifact of the \emph{update-in-place} convention, a fundamental obstacle to concurrency at scale~\cite{Gray81}.
This motivates a different approach: allowing update-conflicting transactions to proceed speculatively rather than blocking them.

\subsection{Prior Speculative Write \revised{Visibility} Schemes}
\label{sec:prior}
\comment{
\rcc is not the first speculative write versioning scheme to allow in-flight transactions to proceed without blocking on WW conflicts. Bamboo~\cite{BAMBOO} and its variant, Rebirth-Retire~\cite{rebirth_retire}, adopt speculative write visibility under locking, while OCC-based schemes such as Silo~\cite{silo} achieve the same effect through optimistic execution.
What distinguishes \rcc from these schemes are three design choices made possible by its redo log-based speculative versioning, detailed in Section~\ref{sec:design}. First, unlike Bamboo and Rebirth-Retire, which eagerly install updates to the data page to expose speculative versions to concurrent readers~\cite{Bamboo-impl}, \rcc defers installation to commit time, affording lightweight rollback: because no changes have yet been installed to the data page, a transaction aborts simply by discarding its speculative versions. 
Second, whereas existing schemes rely on pessimistic deadlock prevention --- wound-wait~\cite{BAMBOO}, which causes false aborts, and wound-retire~\cite{rebirth_retire}, which incurs per-rebirth overhead --- \rcc takes an optimistic stance with commit-time deadlock detection, eliminating both. 
\revisedmsj{-> Second, both schemes incur pessimistic deadlock handling overhead during lock acquisition. They examine conflicting transactions and compare their timestamps to determine whether to wound or rebirth a transaction~\cite{BAMBOO,rebirth_retire}. False aborts in Bamboo and rebirth operations in Rebirth-Retire further amplify this overhead. As more transactions contend for the same record, the cost of these checks increases, and the resulting delays accumulate as subsequent transactions queue for the same record, further increasing lock wait time (\autoref{fig:ycsb_varying_clients_threads}). In contrast, \rcc adopts an optimistic approach to deadlock handling, avoiding these costs on the lock acquisition path.}
Third, because \rcc represents updates as redo log deltas, it naturally enables columnar concurrency control that suppresses false WW conflicts on disjoint columns --- an optimization that is simply unavailable to schemes without delta-based versioning. Together, these three techniques allow \rcc to scale gracefully under high contention workloads such as TPC-C, whereas existing schemes suffer severe throughput degradation beyond 32 concurrent threads.
}

\revised{
\rcc is not the first speculative write visibility scheme to allow in-flight transactions to proceed without blocking on WW conflicts. Bamboo~\cite{BAMBOO} and its variant, Rebirth-Retire~\cite{rebirth_retire}, adopt speculative write visibility under locking. However, two design choices limit the benefits of speculative write visibility.

First, in-place updates in their implementations\footnote{While neither Bamboo nor Rebirth-Retire explicitly documents this behavior, we confirmed it by inspecting their publicly available source code~\cite{bamboo-inplace, rebirth-inplace}.} \revisedhj{make aborts costly}. 
Speculative write visibility allows subsequent transactions to read and overwrite a shared record before their predecessors commit.
If a \revisedhj{predecessor} aborts, all dependent transactions must also abort\revisedhj{, triggering cascading aborts}. 
Since these transactions have overwritten the same record, their updates must be rolled back to restore a consistent state, and a new lock request on the record waits until the restoration completes~\cite{bamboo-rollback, bamboo-lock-wait}. 
\revisedhj{As dependency chains grow, cascading aborts can involve more transactions, increasing in-place undo costs and lock wait time.}

\revisedhj{
Second, both schemes prevent deadlocks pessimistically at lock acquisition.
As illustrated in \autoref{fig:deadlock_handling}(a) and (b), these wound-wait-based schemes compare the timestamps of conflicting transactions.
Bamboo~\cite{BAMBOO} aborts younger transactions even when no cycle exists, causing false aborts. 
Rebirth-Retire~\cite{rebirth_retire} avoids such false aborts by checking for a cycle with topological sorting and, if no cycle would form, performing \emph{Rebirth}, which assigns larger timestamps to the requester and its dependents.
}

In addition, both schemes track dependencies at the record level, introducing unnecessary commit dependencies. Even when transactions update disjoint columns of the same record, they are treated as conflicting. These false dependencies lengthen dependency chains and prolong commit waits.
}

\comment{
\msj{Subsection C + D 합친 버전으로, E 확인 후에 Subsection C,D 삭제 필요}
\swlee{Subsection C + D 로 남겨두는게? 다시 생각해보자구} 
\subsection{\textcolor{purple}{Challenges of Speculative Write Visibility}}
\textcolor{purple}{Prior concurrency control schemes, such as Bamboo~\cite{BAMBOO} and Rebirth-Retire~\cite{rebirth_retire}, allow in-flight transactions to proceed without blocking on WW conflicts by exposing uncommitted writes. 
However, two challenges prevent existing schemes from fully exploiting the ample opportunity for speculative write visibility.
}

\module{\textcolor{purple}{In-Place Updates.}}
As the main culprit, we pinpoint the prevailing convention that all active transactions must apply their updates \emph{immediately and in place} to the \emph{current version} of a data object~\cite{pat-halland-cidr25}. When a transaction is to update a data record, it must first acquire an exclusive write lock, then capture the redo log and modify the current record version immediately~\cite{COWBOOK,ansi-isolation}. This convention causes several concurrency-related limitations. 
When multiple active transactions share and overwrite the current version, even speculative concurrency becomes unsafe: aborting a blocked transaction can trigger cascading rollbacks. \hjin{postgres?DBx1000?}
\textcolor{purple}{
a preceding transaction can roll back only after all subsequent transactions have completed their rollbacks.
This requires cascading rollbacks to proceed serially in reverse update order to restore the record correctly~\cite{txsql}, increasing the time required to complete cascading rollbacks. 
Moreover, new transactions requesting the record must wait until all cascading rollbacks complete, preventing additional transactions from becoming involved in the cascading rollback.
Consequently, serialized rollback increases lock wait time, limiting the concurrency benefits of speculative write visibility.
}

\module{\textcolor{purple}{Deadlocks.}}
\textcolor{purple}{
Although speculative write visibility alleviates WW blocking, dependencies among transactions can still lead to deadlocks.
Existing schemes handle potential deadlocks during lock acquisition. Bamboo~\cite{BAMBOO} adopts wound-wait to prevent deadlocks, potentially aborting transactions even when no deadlock exists. Rebirth-Retire~\cite{rebirth_retire} instead reassigns transaction timestamps to avoid unnecessary aborts, but requires deadlock detection before reassignments. \hjin{디자인장과 중복?}
These approaches introduce deadlock handling overhead directly on the critical path of transaction execution. 
Such aborts may delay subsequent lock acquisitions until the cascading aborts complete. 
Under high contention, frequent conflicts amplify these overheads and wait time for rollback.
}
}


\comment{
\module{WW conflicts dominates in TPC-C, making TPS inscalable.} Even under read committed mode, where WR and RW conflicts are resolved due to multiversion concurrency scheme in MySQL/InnoDB, WW conflicts 아주 빈번. For instance, payment 1 TX때문에, 평균 ? transactions are write-blocked -- 결과적으로 CPU util 낮음, TPS 낮음. core \# 늘려도 TPS saturated. (어떤 TX/Table/Statement간 blocking/blocked/ {\bf 특히, blocked duration} 분석 제시) existing record들에 update에 의한 WW conflict (not insertion or deletion) 주로 발생. 
    The dominating WW blocking makes even the existing MVCC schemes which are known to overcome RW and WR conflicts not scalable to the increasing concurrent threads and CPU cores. As shown in Figure/Table, the TPS of the TPC-C benchmark does not change at all even when the number of concurrent threads is increased. 
    
\begin{itemize}
\item New Order: warehouse dist\_id에 의해 serialize
\item Payment:  in a warehouse, ten district 따라서, ten concurrent까지 TPS 증가후 saturate
\item NO+PM: 결과적으로 몇 thread에서 saturate?
\end{itemize}

\module{Transaction Logic Analysis.} 
We make observations on concurrency-related characteristics in popular OLTP workloads. Those observations will play crucial role in justifying the core concepts of \stx and related assumptions. This section discusses three intriguing characteristics about XXX we have observed from realistic OLTP workloads write traces: most transaction successfully commits, update conflicts happens at the early phase of transactions, and they access data mostly once. We capitalize on these observation in designing \stx. 

In addition, to better understand the severity of inscalable concurrency with TPC-C benchmark, we analyze the characteristics of two major transaction types in TPC-C, NO and PM and make the following observations. 
    
\begin{itemize}
\item (Most transaction commits) transactions are programmed to commit without logical abort. -> deadlock abort 제외한 대부분의 transaction은 commit할 것이라는 optimism. 

\item (Conflicts at early stage of transaction execution) 
In the TPC-C benchmark, most transaction conflicts arise at the very beginning of the New-Order and Payment procedures. This is because every New-Order transaction must retrieve and increment the same D\_NEXT\_O\_ID field from the DISTRICT table to obtain a new order identifier, while each Payment transaction updates the warehouse- and district-level sales aggregates (W\_YTD, D\_YTD). These shared records serve as global synchronization points, forcing all concurrent transactions targeting the same district or warehouse to serialize their execution. Once this early phase completes, subsequent operations—such as inserting new orders or order lines—touch disjoint tuples and therefore proceed with little or no contention. As Wang et al.~\cite{Wang} observe, the D\_NEXT\_O\_ID update becomes the dominant contention hotspot, making the conflicting logic front-loaded in TPC-C transactions rather than evenly distributed across their lifetime 

most WW conflicts among transaction 은 early phases of transactions. WH and DIST의 경우, NO와 PM의 가장 앞 부분에 locate하고, 나머지 부분은 충돌없이 진행됨. 단일 수행 시간 기준으로 1/10th에 위치(?) \reminder{혜진,NO/PM 각각 수행시 초당 TPS 측정-NO/PM 상대적 수행 길이 비교, WH/DIST 수행 시점 측정 및 표기 부탁}\hjin{128core에서 NO(2ms), PM(0.5ms) 소요. NO: 평균 45query 중 2번째(DT) query, PM: 10 query중 1(WH), 2번째(DT) qeury에서 WW conflict 발생} 따라서, other researchers~\cite{BAMBOO,Quro}에서 지적한 것 처럼, hot object에 대한 write lock을 획득한 transaction만 수행하고 commit시까지 other tranactions을 blocked시켜서, 즉, strictly serialization 시켜, concurrency를 아주 낮게 만듦. 
    {\bf Why? schema 구조/application logic에 기인}; (NO long vs. PM short) long tx에 의해 short tx long blocking시 long TX commit시까지 오래동안 blocking 됨. 
   \comment{
    The D\_NEXT\_O\_ID value in the DISTRICT table is a major contention point because all New-Order transactions in the same district must update this value to obtain a new order ID~\cite{Wang}. Yan et. al. made the similar observation and suggested the {\em Quro} scheme which delays the hotspot contention later so as to avoid the long blocking time blocking time among concurrent transactions, thus maximing the concurrency~\cite{Quro}. 
    만일 conflicting region이 backend에 있다면, blocked duration 줄여 concurrency 저하 최소화할 수 있음~\cite{BAMBOO}) -- 만일 앞부분의 WW conflict를 resolve해서 contending transaction들이 pipeline될 수 있으면, CPU util과 concurrency를 아주 높일 수 있음.  
    \reminder{혜진, conflict point를 TX 끝부분 미루면, TPS 등 지표 얼마나 개선? --> conflict 유발 WW point를 뒤로 옮겼을 경우, TPS/CPU util measure/present?. Bamboo~\cite{BAMBOO} 실험 결과?} 
    }
\item  (Transactions tend to accesses objects only once) Access once 각 TX는 object에 대해 (대부분?) 논리적으로 one access (즉, 대부분 access는 final read/write임, select-then-update \reminder{SFU 경우에도 대부분 duration 짧음?}포함, final access 개념 TPC-C 예시로 설명); 또한 설령 그렇지 않더라도 analysis/detect/preprocessing을 통한 final read/write auto hint 가능함~\cite{PWV,BAMBOO}. 따라서, 어떤 TX T1가 update한 object를 다른 TX가 다시 read/update해도, T1이 commit하면 이후 TX들 consistency 문제 없음. 
\item Practical OLTP system often acquire exclusive access to critical records near the beginning of a transaction. This intentionally serialize the execution path, prevents race conditions on hotspot tuples and ensures that business logic observes a consistent state. In TPC-C, \textit{New-Order} updates the district’s next-order ID before creating the order; \textit{Payment} similarly updates warehouse and district aggregates before touching the customer record. In TPC-E, \textit{Trade-Result} acquires an early SELECT FOR UPDATE on the relevant trade tuple, ensuring that no concurrent transaction can modify it while the order-processing logic proceeds. To preserve correctness, systems often acquire exclusive access to the records to be updated before any modification occurs, ensuring that no concurrent transaction can alter those values and thereby preventing race conditions.

\item Another practical reason is deadlock reduction: imposing a consistent lock order by requesting key locks earlier can reduce cycles in the wait graph \cite{Brook-2PL}

\end{itemize}

\module{Why the existing MV/optimistic CCs are not such successful in spite of opportunity?} 
WW conflict시 early write lock release를 통해 (즉, early write visibility) concurrency를 높일 기회요인에도 불구하고, existing MVCC protocols also make update in-place and immediately로 인해 lost update, non-repeatable read, dirty read의 위험성때문에, conservative접근해서 기회를 살리지 못하고 있다.  OCC의 경우, hot object에 의한 충돌과 과도한 abort로 인해 부적합하다.  
}

\comment{ 
\item \stx 관점 다른 특징? 
\item (Columnar CC 복선 - row-level false conflicts; no columnar conflict)  또한, much less conflict in column-granule, not tuple-granule level: \item Serializable vs. RR mode 차이:  serializable에서 FK 의 CC에 대한 영향있고, RR서는 없음? 
}

\comment{ 
마찬가지로, FK를 disable시키고 conflict point를 뒤로 미루어서 deferred check을 하면 얼마나 개선?))
}

\comment{이 부분은 실험장에서 \stx 효과로 대체; vanilla 대비 STX 효과 더 높고, SSD 입장에서 IOPS도 높아져서, CPU / IO util 모두 높임. 
    \item (I/O에 의한 concurrency 악화) 관련해서, as the execution time of transaction holding the write lock is prolonged, other blocked transactions has to wait longer. To verify this, we ran the same TPC-C benchmark with smaller buffer cache, in which transactions may experience page misses. 이처럼 blocking TX가 I/O 수반 시, blocked TX의 waiting 더 심해져서, 따라서 throughput과 latency 악화됨. For example, in the cae of {\tt new\_order} transaction in TPC-C, WW conflict 유발 table/record는 hot해서 cached되어 있지만, other large tables such as orderline and stocks are very likely to cause I/Os. 이 경우, WW blocking이 make less IO requests하는 지? -- Implication: \stx 등으로 인해 concurency 개선 뿐 아니라 I/O util도 개선? I/O bound 상황에서 CC 저하 현상이 덜 나타날 수도 있음 (혜진 결과에서 일부 확인)
}

\comment{
\subsection{The Myth of Lost Update}~\label{sec:mvcc}
Under MVCC, a transaction attempts to update the most recent version of a data item that is visible under its snapshot.
This approach inherently leads to the possibility of Write-Write (\textbf{W-W}) conflicts, as concurrent transactions may simultaneously seek to modify the latest committed version.
If a concurrent transaction has already updated (committed) that item, system detects W-W conflict. In this case, the update may fail or the transaction may be rolled back.
Such W-W conflicts are very common in OLTP workloads. In our experiments using the TPC-C benchmark, we observed a high rate of W-W conflicts on both MySQL and PostgreSQL. \reminder{TPC-C @mysql 또는 postgres 실험중 w-w conflict 비율 높은 것 제시}
The rationale behind requiring transactions to update most recent visible version is to prevent the \textbf{Lost Update} anomaly~\cite{lostupdate:mit99, lostupdate:vldb14, ansi-isolation}.
The Lost Update anomaly occurs when multiple transactions attempt to update the same data item concurrently, and the update made by an earlier transaction is overwritten by a subsequent transaction, resulting in data inconsistency.
Even though MVCC maintains multiple versions of data items to allow concurrent reads and writes, write-write (W-W) conflicts are still inevitable.

Instead of directly modifying the latest committed version, transactions can only record their changes as redo logs and defer applying these changes to the data pages until commit time.
Moreover, this forms dependency chains, where a series of transactions transitively wait for their predecessors to commit. 
By tracking transaction dependencies during this period, the system can reduce the W-W conflicts and, verifies that no Lost Update anomalies have occurred prior to commit, even if temporarily allowing conflicting updates (\ie, implicitly systems permit the temporal lost update anomaly) between transactions.
This enables us to extract more concurrency, improving transaction throughput in MVCC protocols.

\comment{
참고 from Generic Version Control: Configurable Versioning for Application-Specific Requirements
Isolation and the Commit Validation Phase. MVCC systems offer isolation levels by implementing access rules and conflict resolution mechanisms. Most MVCC implementations use metadata fields, like begin and end timestamps, to set access rules to avoid dirty reads, non-repeatable reads, and phantom reads. They also prevent the lost update anomaly by using a first-writer-wins rule [10, 14, 15], ensuring isolation levels up to repeatable read and snapshot. To achieve serializable isolation, read-write conflicts are resolved during commit validation, using methods like dependency graphs [2, 19], read and scan sets [5, 10], or predicate logging [14].
}

\comment{
storyline:
reader/writer는 서로를 block하지 않지만, w-w conflict는 발생함.
root cause: under MVCC, 트랜잭션들은 항상 현재 버전에 데이터에 업데이트를 하려고 함-> 그 이유?
Meanwhile DBMSs can have three timelines: past, present, and future.  
Why do transactions in MVCC can only update the latest version of data, thereby causing write-write (W-W) conflicts?
MVCC에서 만약 latest version이 아니라 old 버전에 업데이트를 허용하면 발생하는 문제점:
\begin{itemize}
    \item  이미 다른 트랜잭션에 의해 변경된 데이터에 대한 충돌 발생 - 직렬화 보장 어려움 및 복잡성 증가
    \item  multiple transaction들이 동시에 서로 다른 버전에 대해 충돌할 수 있는 문제가 발생함. 최신버전에 대해서만 업데이틀르 강제한 다면 위와 같은 문제를 쉽게 detect할 수 있고,  충돌발생할 경우, 다른 트랜잭션을 rollback시켜서 간단 해결 가능.
    \item  lost update 문제 발생
\end{itemize}
}
\comment{
\begin{itemize}
    \item MVCC 개요 - 핵심: single logical data에 대해 multiple physical versiond을 관리하는 multi-versioning; mv를 동시성제어 기법과 결합하여 사용함.
    그중 대부분의 commercial dbms mysql, postgres, oracle 등이 MV2PL을 차용함. 
    핵심: writers (readers) do not block readers (writers);
    \item MVCC의 구현은 다양한 version storage 형태로 지원함: N2O, O2N, Delta, Time Travel; 여기서 N2O, O2N 중심으로 설명 
    \item Garbage collection and index;
    \item Under MVCC, DBMS에서는 undo log (past), data page (present), redo log (future)가 3개의 timeline 공존; time travel이 가능함; (단, SAP HANA와 같이 version stroage레벨에서 snapshot을 제공하는 것이 아니라, 트랜잭션 레벨에서 
    \item \textbf{stereotype}: update는 data page의 최신 버전에 써야 한다! 결국 w-w conflict는 MVCC를 써도 피할 수 없으며, 성능의 병목! (Section ~\ref{sec:oltp-ww} 에서 내용 보강.
    \item 앞서 설명한, (tx-level) time travel을 한다면?
\end{itemize}
}

\module{MVCC} Pavlo paper~\cite{mvcc-eval} 서론 참고; MVCC is prevalent. Taxonomy about CC, version storage, garbage collection and indexing from Pavlo~\cite{mvcc-eval}.  

{\bf MVCC+2PL(MV2PL) 가정} In this paper, we assume MV2PL(MVCC+2PL)~\cite{mvcc-eval} for \stx for the simplicity of discussion. And because most popular DBMSs including Oracle, MySQL, Postgres, and SAP HANA are taking the approach. 

\module{MVCC and timestamp} In most MVCC system, each transaction and data version has its timestamp, and using the timestamp, each transaction determines the version to read~\cite{mvcc-eval}. \stx도 똑같은 가정을 함(?).   

\module{Cascaded Rollback}

\module{Recoverable schedule} ~\cite{BHG}

\module{Commit Ordering} ~\cite{BHG}

\module{Last bottleneck: WW conflict 그리고 Current Version 관련 Dilemma} update는 data page의 최신 버전에 써야 한다는 생각의 Dogma 갖혀 있음. 
}

\revised{
\section{Optimistic Deadlock Handling}\label{sec:opt_dl}



\begin{figure}[t]
  \centering
  \subfloat[Bamboo]{%
    \includegraphics[width=.43\linewidth]{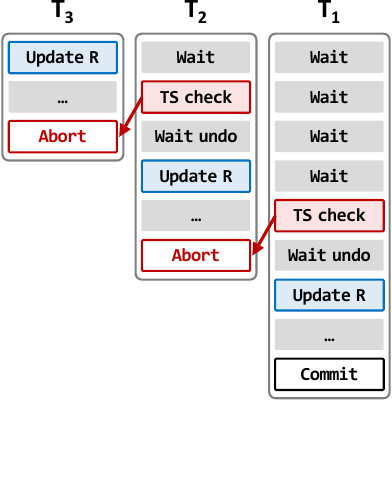}%
  }%
  \hfill
  \subfloat[Rebirth-Retire]{%
    \includegraphics[width=.43\linewidth]{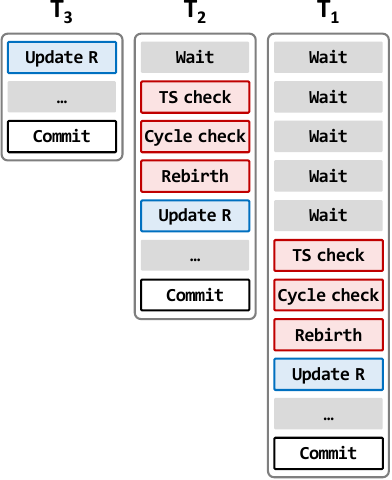}%
  }

  \par
  \subfloat[\rcc]{%
    \includegraphics[width=.99\linewidth]{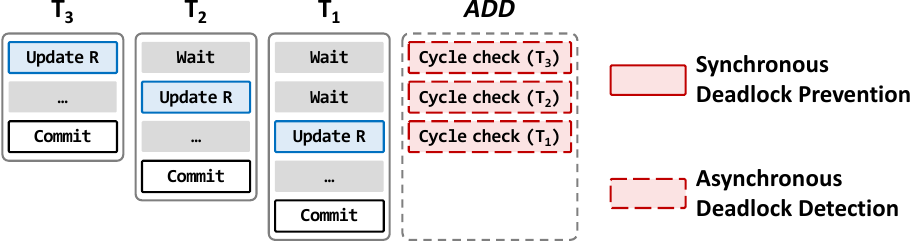}%
  }

  \caption{\revised{Deadlock handling overhead when no dependency cycle exists. Unlike Bamboo and Rebirth-Retire, \rcc checks cycles asynchronously, off the lock acquisition path. Each schedule shows transactions updating a hot record $R$, with $ts(T_1) < ts(T_2) < ts(T_3)$ (smaller is older).}}
  \label{fig:deadlock_handling}
\end{figure}

When a transaction reads the speculative version of another active transaction, it is said to depend on that transaction. 
This creates a \emph{commit dependency}: the dependent cannot commit before its \revisedhj{predecessor} and must abort if the \revisedhj{predecessor} aborts or re-updates its speculative version~\cite{raz:93}. 
When commit dependencies form a cycle, the involved transactions wait for each other to commit, resulting in a deadlock.
Prior lock-based speculative write visibility systems, such as Bamboo~\cite{BAMBOO} and Rebirth-Retire~\cite{rebirth_retire}, prevent deadlocks with timestamp-based deadlock handling at lock acquisition.
Bamboo follows wound-wait and aborts younger transactions.
Rebirth-Retire aborts a younger transaction only when a cycle check confirms a deadlock.
\comment{
Under high contention, more transaction contend for the same records, increasing the cost of synchronous deadlock prevention.
In our YCSB experiment ($\theta=0.9$), the time Bamboo's committed transactions spend on or waiting for deadlock prevention rises from 2.9\% their latency at 16 threads to 41.4\% at 128 threads, nearly fourteenfold.
Since these operations must finish before a conflicting lock request can proceed, they delay lock acquisition for subsequent transactions, whether or not a deadlock exists.
}
\revisedhj{
Under high contention, longer lock lists increase timestamp checks, while larger dependency graphs increase Rebirth-Retire's cycle-checking and rebirth costs~\cite{rebirth_retire}.
These synchronous checks delay every subsequent conflicting lock request even when no deadlock exists  (\autoref{fig:deadlock_handling}(a) and (b)).
In our YCSB experiment ($\theta=0.9$), the time Bamboo's committed transactions spend performing or waiting for deadlock prevention rises from 2.9\% of their latency at 16 threads to 41.4\% at 128 threads, nearly fourteenfold.
}



\rcc instead adopts \emph{optimistic deadlock handling}.
We call it optimistic because it assumes that dependency cycles are rare and pays for them only when one actually forms, rather than on every conflicting lock request.
Rather than incurring synchronous prevention overhead on conflicting lock requests, it permits cycles to form and detects them asynchronously. Once detected, the cycles are resolved by aborting the involved transactions.
To the best of our knowledge, \rcc is the first lock-based speculative write visibility scheme to handle deadlocks optimistically.
This optimistic design delivers three key benefits.
First, \rcc grants locks without checking for potential cycles, which removes deadlock handling from lock acquisition.
\revisedhj{
As shown in \autoref{fig:deadlock_handling}(c), {\tt T1}, {\tt T2}, and {\tt T3} update {\tt $R$} without waiting for cycle checks, which run asynchronously off the lock acquisition path.
A transaction is thus not delayed by synchronous cycle detection during its execution.
Under high contention, this sustains the pipelining that speculative write visibility provides.
}
Second, \rcc avoids false aborts caused by \revisedhj{wound-wait} deadlock prevention~\cite{rebirth_retire}. Unlike Bamboo, \rcc aborts transactions to resolve deadlocks only after an actual cycle is detected.
Finally, \rcc does not impose commit waits on \emph{RW anti-dependencies}.
An RW anti-dependency arises when a transaction updates a record that another active transaction has read~\cite{lostupdate:mit99}. 
Unlike WR or WW dependencies, RW anti-dependencies \revisedmsj{impose} no commit ordering: the writer \revisedmsj{can} commit without waiting for those readers, since it consumes no value they wrote.
However, an RW anti-dependency can participate in a cycle closed by a reader after its writer has committed~\cite{PSSI}.
\rcc detects such cycles and resolves them by aborting the reader.
To minimize wasted work caused by transactions in a cycle and their dependents, \rcc detects cycles using a dedicated background thread, the \emph{asynchronous deadlock detector} (\emph{ADD}, detailed in Section~\ref{sec:design:add}).
By detecting cycles shortly after they form, \emph{ADD} keeps dependency chains from growing, which reduces cascading aborts.
However, deadlocks can still form before \emph{ADD} detects them.
When deadlocks or aborts are frequent, cascading aborts and wasted work can outweigh the gains of speculative write visibility~\cite{SCC,NeurCC}.
}

\begin{figure}[t]
\centering
\includegraphics[clip,width=0.48\textwidth]{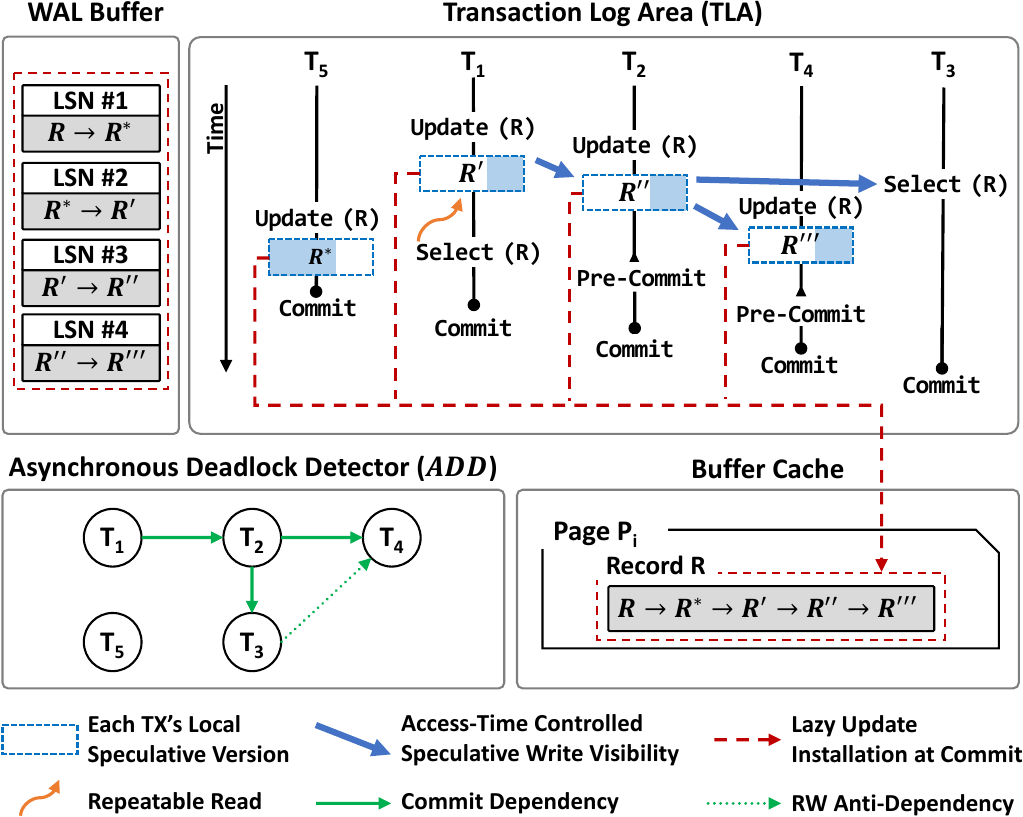}
\caption{\rcc Overview: Shaded regions indicate the modified columns of hot record $R$. $T_1$, $T_2$, and $T_4$ update Col A, while $T_5$ updates Col B.}
\label{fig:stx-archi}
\end{figure}

\section{Design of \stx}\label{sec:design}
Section~\ref{sec:background} shows that OLTP transactions typically commit and encounter WW conflicts early, yet existing systems serialize them by updating records in place under exclusive locks held until commit.
\rcc instead materializes speculative write versions out of place from redo logs, installs them at commit, and orders them through dependency tracking.
Conflicting transactions thus run in a pipeline with isolation preserved.

\subsection{Key Idea: Redo Logs as \revised{Speculative Write} Versions}\label{sec:motiv}

Modern DBMSs maintain three versions of each record: the current version, undo log entries keeping past committed versions, and redo log entries recording intended future updates. While the undo log is leveraged as version storage to resolve RW and WR conflicts in existing MVCC schemes, redo logs have been used solely to guarantee durability by persisting transactional updates at commit and replaying them upon recovery. However, redo logs naturally capture a transaction's \emph{update intention} and, crucially, require no undo on abort --- they are simply discarded. This motivates a key question: 
\emph{Can redo logs designed for recovery also serve as write versions for concurrency control?}

\rcc answers this question by repurposing redo log entries as \emph{speculative write versions} --- uncommitted, abortable versions of a record materialized by active transactions. For brevity, we refer to them simply as \emph{speculative versions} hereafter.

When updating a record, a transaction reads the speculative version from its predecessor or the current version from the data page, appends a redo log entry, and creates a new speculative version by applying this redo log.
When multiple transactions are updating the same record, each active transaction thus maintains its own speculative version of the record, enabling multiple active transactions to update the same record instead of overwriting the single current state. The speculative versions of an active transaction are installed to their corresponding current versions in data pages lazily at commit time. Unlike OCC, where transactions execute without coordination and require validation at commit to detect conflicts, \rcc enforces a total order among conflicting transactions through lock acquisition ordering, serializing their commits without validation.

Treating active transactions' redo logs for a record as speculative versions --- thereby decoupling them from the record's current version --- enables a fundamentally different approach to concurrency control.  First, dirty reads are precluded: if a transaction reads a speculative version produced by another, it becomes dependent on the producer. Commit ordering ensures the consumer commits only after the producer; should the producer abort, the consumer is cascadingly aborted as well, ensuring no transaction ever commits based on data that was never committed~\cite{raz:93}. Second, repeatable reads are guaranteed because each transaction retains its own speculative version even if other active transactions are updating the same object. Lastly, lost updates are also prevented because concurrent writers create their speculative versions in lock acquisition order, each from the version of its predecessor, so no update overwrites another.

\comment{To realize the full potential of {\em redo log as a version}, \stx proposes two novel concepts,
\textcolor{blue}{as discussed below. The first removes cycle detection from the execution path, deferring the abort decision to commit time and detects cycles caused by RW anti-dependencies in the background. The second tracks dependencies at column granularity to eliminate false WW conflicts.}}
\comment{
The redo log-based design enables three novel techniques.
First, \emph{lazy \revisedhj{update} installation} defers update propagation to commit time, enabling lightweight abort: since no changes have been applied to the data page, a transaction aborts simply by discarding its speculative versions. 
Second, \rcc moves deadlock detection off the lock conflict point and performs it asynchronously before commit, keeping the execution path free of cycle detection overhead while aborting truly deadlocked transactions early.
Third, \emph{columnar concurrency control} tracks dependencies at column granularity using redo log deltas, eliminating false WW conflicts among transactions that modify disjoint columns of the same record.
}
The redo log-based design enables 
\revisedhj{
\emph{lazy \revisedhj{update} installation}, which defers update propagation to commit time and allows lightweight rollback: since no changes have been applied to the data page, a transaction aborts simply by discarding its speculative versions.
Based on this, \rcc introduces two 
}
novel techniques.
First, \rcc moves deadlock detection off the lock conflict point and performs it asynchronously before commit, keeping the execution path free of cycle detection overhead while aborting truly deadlocked transactions early.
Second, \emph{columnar concurrency control} tracks dependencies at column granularity using redo log deltas, eliminating false WW conflicts among transactions that modify disjoint columns of the same record.

\subsection{Architecture}\label{sec:design-obj-archi}

\comment{
\autoref{fig:stx-archi} illustrates \rcc as an extension to N2O MVCC systems such as MySQL.
Since \rcc maintains speculative versions out-of-place independent of the underlying MVCC architecture, O2N system such as PostgreSQL can also be naturally extended with \stx, as discussed in \autoref{appendix}. 
}

\module{Constructing Speculative Versions.} 
\comment{
As shown in \autoref{fig:stx-archi}, when a transaction updates a data record (\ie, the \emph{current version} of the record having the last committed value), it acquires an exclusive lock on the record, captures the corresponding redo log, and creates a speculative version using the redo log (\eg, the blue dashed rectangles in \autoref{fig:stx-archi}).  The speculative version provides \emph{repeatable reads} to the transaction, even when other transactions have also updated the same record. 
\stx shares a common principle with OCC~\cite{OCC} and prior speculative versioning schemes such as Bamboo and Rebirth-Retire~\cite{BAMBOO,rebirth_retire}: each active transaction maintains a private copy of the records it accesses.
}
As shown in \autoref{fig:stx-archi}, when a transaction updates a data record (\ie, the \emph{current version} of the record having the last committed value), it acquires an exclusive lock on the record, captures the corresponding redo log, and creates a speculative version using the redo log (\eg, the blue dashed rectangles in \autoref{fig:stx-archi}). The speculative version guarantees \emph{repeatable reads} for the transaction, even when other transactions have also updated the same record. 
\comment{\stx shares a common principle with OCC~\cite{OCC} and prior speculative versioning schemes such as Bamboo and Rebirth-Retire~\cite{BAMBOO, rebirth_retire}: each active transaction maintains a private copy of the records it accesses. 
}
Note, however, that unlike Bamboo and Rebirth-Retire, which install each record update to the corresponding data page \emph{before commit} so that conflicting transactions can access the updated record, \rcc defers installation to commit time, as indicated by the red dashed arrow in \autoref{fig:stx-archi}.

\module{Transaction Log Area.}
Rather than introducing a separate version store, \stx leverages redo logs directly as speculative versions. However, constructing per-record speculative versions from the WAL alone is inefficient: the WAL is organized at the transaction level, so extracting relevant updates for a given record requires scanning the entire log --- prohibitive 
when log records must be retrieved from disk.

To mitigate the overhead, \stx introduces the \emph{Transaction Log Area} (\emph{TLA}), a per-transaction in-memory heap that co-locates redo log entries with their corresponding speculative versions, enabling direct access without repeated WAL reconstruction.
Although speculative versions can be reconstructed by applying redo logs, storing them separately in the \emph{TLA} eliminates lock contention and version traversal overhead.
The space requirement of the \emph{TLA} depends on the number of active transactions and the number and \revised{size of records being accessed}. For instance, with 128 concurrent threads running {\tt New-Order} and {\tt Payment} transactions in TPC-C, the maximum \emph{TLA} capacity is about 8 KB. 
This amount of additional \emph{TLA} space, though not negligible, is well justified by the increased concurrency. 

\module{\revised{Lazy Update Installation.}}
\rcc materializes each update as a speculative version in the \emph{TLA} rather than modifying the record in place, and installs it into the data page only at commit.
We call this \emph{lazy update installation}.
Data pages thus retain the most recently committed value.

Lazy update installation also makes cascading aborts lightweight.
In \rcc, cascading aborts arise when a predecessor aborts, or when it re-updates a record whose speculative version a dependent has already read.
In \autoref{fig:stx-archi}, for instance, if {\tt T2} reads the speculative version {\tt $R^{\prime}$} written by {\tt T1} and {\tt T1} subsequently aborts, the dependent transaction {\tt T2} must also abort. For another example, if {\tt T1} \emph{re-updates} the record {\tt $R$} after {\tt T2} had already read its speculative version ({\tt $R^{\prime}$}), \rcc prevents stale reads by aborting {\tt T2} that has read the first update.
Such aborts can cascade along long dependency chains under high contention.
Since no update has been installed into the data page, however, \rcc aborts transactions in parallel by discarding their speculative versions in the \emph{TLA}; unlike in-place schemes, it \revisedhj{does not revert current versions}.
\revised{This cheap abort is what makes optimistic deadlock handling attractive: a deadlock can be resolved cheaply after it forms.}

Lazy update installation adds overhead even without WW conflicts.
Because subsequent transactions may read uncommitted updates and the updater may later abort, \rcc materializes a speculative version for every update and applies it to the data page at commit.

\module{Access-Time Conflict Ordering.}
Each transaction acquires an exclusive lock before updating a record, using the same lock table mechanism as existing DBMSs (\autoref{fig:stx-archi}). \stx then introduces \emph{speculative unlock}: rather than holding the lock until commit, the transaction releases it immediately after creating its speculative version, allowing subsequent transactions to proceed without waiting. The lock table entry is retained for dependency tracking, so that \rcc can detect re-updates and cascade-abort dependent transactions.
A subsequent transaction updating the same record acquires the exclusive lock and creates its own speculative version based on its predecessor's. For example, in \autoref{fig:stx-archi}, {\tt T2} acquires the lock on {\tt $R$} and creates {\tt $R^{\prime\prime}$} from {\tt T1}'s version {\tt $R^{\prime}$}.

\stx guarantees conflict-serializable schedules~\cite{COWBOOK} by enforcing that commit order follows lock acquisition order. In this sense, \stx satisfies commit order consistency~\cite{raz:93}: 
once the update order among conflicting transactions is fixed, speculative execution is safe as long as commits follow that order.
Unlike OCC~\cite{OCC,silo}, which defers conflict detection to validation and incurs frequent aborts under high contention, \stx enforces lock-ordered access and pipelines conflicting transactions to completion.

\revised{
\module{Dependency Tracking.} \label{sec:design:add}
\rcc represents commit dependencies and RW anti-dependencies in a dependency graph and detects cycles. RW anti-dependencies are tracked only under \texttt{Serializable} (Section~\ref{sec:read}). Active transactions accessing the same record form a dependency chain~\cite{CLV,PWV}, managed in the lock table.

To track dependencies involving committed transactions, \rcc records the latest committed writer and the current readers of each record. This yields WR dependencies from a committed writer to later readers of its version, and RW anti-dependencies from committed readers to a transaction that subsequently updates the record.
When a transaction commits and installs a new version, \rcc sets the record's committed writer to the committing transaction and clears the prior reader list.
\rcc retains edges after the writer commits until they can no longer participate in a cycle, so that any cycle subsequently closed by a reader is still detected and the reader is aborted.

As transactions generate dependencies, \emph{ADD} records them in its dependency graph and checks for cycles in parallel with transaction execution. Before a transaction commits, \rcc waits until \emph{ADD} has finished checking all of its dependency edges.
According to our experiments, \emph{ADD} completes this check before the transaction reaches the commit point in most cases, so the committing transaction rarely waits for it. 

When a cycle involves only writers, all transactions in the cycle are aborted. Otherwise, when the cycle includes readers, \rcc prefers a reader over a writer as the abort victim to reduce cascading aborts.
For example, consider a commit-dependency chain $T_i \to T_j \to T_k \to T_l$, where an RW anti-dependency $T_l \to T_i$ closes the cycle. Aborting $T_i$ aborts all four transactions, whereas selecting reader $T_l$ as the victim breaks the cycle with only one abort.
Tracking record accesses to capture future dependencies incurs graph management overhead even without conflicts.
Under low contention, where avoiding pessimistic deadlock handling yields little benefit, this overhead can diminish the benefit of \rcc.

\comment{
When a transaction reads the speculative version of another active transaction, it is said to depend on that transaction. The reading transaction is called a \emph{dependent} and the one whose version is read is called a \emph{preceding} transaction. \rcc tracks two kinds of dependencies, as shown in the lock table of Figure~\ref{fig:stx-archi}.

A \emph{commit dependency} arises when a dependent reads the speculative version of its preceding transaction, as in WW and WR conflicts.
The dependent cannot commit before its preceding transaction and must abort if the preceding transaction aborts or re-updates its speculative version.

An \emph{RW anti-dependency} arises when a transaction updates a record that another active transaction has read~\cite{lostupdate:mit99}. Unlike commit dependencies, RW anti-dependencies impose no commit ordering: the writer commits without waiting for those readers, since it consumes no value they wrote. However, \rcc retains the RW edge after the writer commits until it can no longer participate in a cycle, so that any cycle subsequently closed by a reader is still detected and the reader is aborted.

Note that \rcc tracks RW anti-dependencies only under \texttt{Serializable} (Section~\ref{sec:read}). \rcc represents both dependencies as edges in the dependency graph and detects cycles, as described below. Active transactions accessing the same record form a dependency chain~\cite{CLV,PWV}, managed in the lock table. A transaction that has completed execution but has unresolved commit dependencies waits in the \emph{pre-commit} state until they are resolved (Section~\ref{sec:commit}).

To track dependencies involving committed transactions, \rcc records the latest committed writer and the current readers of each record. This yields WR dependencies from a committed writer to later readers of its version, and RW anti-dependencies from committed readers to a transaction that subsequently updates the record.
When a transaction commits and installs a new version, \rcc sets the record's committed writer to the committing transaction and clears the prior reader list; existing dependency edges are preserved in the dependency graph until cycle detection resolves them.
Tracking record accesses to capture future dependencies incurs graph management overhead even without conflicts.
Under low contention, where avoiding pessimistic deadlock handling yields little benefit, this overhead can diminish the benefit of \rcc.
}
}

\comment{
    \textcolor{blue}{
    \module{Commit-Time Deadlock Detection.}
    A cycle in the dependency graph arises in two cases.
    In the first, pre-committed transactions wait for each other through commit dependencies, forming a deadlock.
    In the second, an RW anti-dependency closes a cycle after the writer has committed, violating serializability.
    While \rcc removes lock-induced blocking during execution, its dependency chains remain susceptible to such cycles, shifting the bottleneck from lock contention to cycle detection.
    This concern applies to all speculative write versioning schemes~\cite{BAMBOO,rebirth_retire},
    Prior schemes handle deadlocks pessimistically.
    Bamboo~\cite{BAMBOO} extends wound-wait to prevent deadlocks, and Rebirth-Retire~\cite{rebirth_retire} checks for a cycle at each conflicting lock acquisition before allowing the dependency to form.
    Both place deadlock handling on the execution path of every transaction, regardless of whether a cycle exists.
    }
    
    \rcc instead takes an optimistic approach.
    Transactions execute without interruption from cycle handling, and \rcc aborts a transaction only when a cycle involving it is confirmed before its commit.
    We call this \emph{commit-time deadlock detection}.
    \textcolor{red}{
    This design suits \rcc for two reasons.
    First, \rcc resolves only actual cycles and thus avoids the \emph{false aborts} of deadlock prevention~\cite{rebirth_retire}.
    Second, lazy update installation keeps cascading rollback lightweight.
    Aborted transactions discard their speculative versions independently, without undoing shared records in reverse dependency order, so cascading rollbacks proceed concurrently.
    When true deadlocks are frequent, the optimistic approach also allows more speculative work to proceed before a cycle is resolved, increasing wasted work and cascading aborts.
    }
}
\comment{
\module{Commit-Time Deadlock Detection.}
Lock-based speculative write versioning schemes, including \rcc, cannot eliminate cycles in the dependency graph. Existing schemes handle this pessimistically on the execution path. Bamboo~\cite{BAMBOO} adopts wound-wait to prevent deadlocks, aborting younger transactions even when no cycle exists. Rebirth-Retire~\cite{rebirth_retire} detects cycles at each \revised{conflict} that conflicts with a younger transaction before allowing a dependency to form. 
First and foremost, both introduce deadlock handling overhead on the lock acquisition path, delaying subsequent transactions. This delay increases further when either scheme triggers aborts, since subsequent transactions must wait for cascading rollbacks to complete.

\rcc instead takes an \emph{optimistic} approach: to keep the execution path free of deadlock handling overhead, \rcc grants locks without regard to potential cycles, allowing transactions to execute without interruption. As a result, cycles can form in the dependency graph in two cases: pre-committed transactions waiting on each other through commit dependencies (\emph{deadlock}), or an RW anti-dependency closing a cycle after its writer has committed (\emph{serializability violation}). \rcc detects and resolves such cycles only at commit time --- a strategy we call \emph{commit-time deadlock detection}. This design resolves only true cycles, avoiding the false aborts of prevention-based schemes~\cite{ rebirth_retire}, and leverages lazy update installation to keep cascading rollback lightweight: aborted transactions discard their speculative versions independently, without undoing shared records in reverse dependency order, so cascading rollbacks proceed concurrently.

\module{Asynchronous Deadlock Detection} 
This optimistic approach, however, can cause wasted work in \rcc: transactions in a cycle and their dependents continue executing or wait in the pre-commit state until detection, consuming CPU cycles on work that will ultimately be aborted. In particular, more dependents may accumulate before a cycle is detected, amplifying the cascade of aborts when it is finally resolved.

To resolve cycles earlier, \rcc detects them asynchronously using a dedicated background thread, the \emph{asynchronous deadlock detector} (\emph{ADD}) 
As transactions generate dependencies, \emph{ADD} collects them into its own view of the dependency graph and checks for cycles in parallel with transaction execution. Before a transaction commits, \rcc waits for \emph{ADD} to finish examining all of its dependency edges. In practice, \emph{ADD} completes cycle examination for the committing transaction before it reaches the commit point, so the committing transaction rarely needs to wait for \emph{ADD} (Section~\ref{sec:eval-ycsb}).


When a cycle involves only writers, all transactions in the cycle are aborted. Otherwise, when the cycle includes readers --- as occurs when RW anti-dependencies close a cycle --- \rcc prefers a reader over a writer as the abort victim so as to reduce cascading aborts. 
For example, consider a commit-dependency chain $T_i \to T_j \to T_k \to T_l$, where an RW anti-dependency $T_l \to T_i$ closes the cycle. Aborting $T_i$ aborts all four transactions, whereas selecting reader $T_l$ as the victim breaks the cycle with only one abort.

When multiple transactions must be aborted together --- as in cascading aborts or cycle resolution --- \rcc can abort them in parallel and in any order. Because no changes have been installed to data pages, each transaction independently discards its speculative versions in the TLA without coordinating with others. This stands in contrast to prior speculative versioning systems such as Bamboo~\cite{BAMBOO}, which perform in-place updates and must therefore abort transactions sequentially in reverse dependency order to maintain consistency~\cite{txsql}.

\revisedmsj{
Yet deadlocks can form during speculative execution before \emph{ADD} detects them. When deadlocks or aborts are frequent, cascading aborts and wasted work can outweigh the gains of speculative write visibility~\cite{SCC, BAMBOO}. When they are rare, however, \rcc fully exploits its pipelining potential, resolving occasional deadlocks asynchronously with lightweight rollback.
}
}

\comment{
\revisedhj{
\module{Optimistic Deadlock Handling.}
}
Speculative write visibility under locking, including \rcc, must handle cycles in the dependency graph.
Prior schemes, such as Bamboo~\cite{BAMBOO} and Rebirth-Retire~\cite{rebirth_retire}, employ deadlock prevention based on wound-wait. 
However, this pessimistic approach incurs overhead from timestamp-based deadlock handling during lock acquisition and can cause false aborts in Bamboo.

\rcc instead takes an \emph{optimistic} approach: to keep the execution path free of deadlock handling overhead, \rcc grants locks without regard to potential cycles, allowing transactions to execute without interruption. 
As a result, cycles can form in the dependency graph in two cases: pre-committed transactions waiting on each other through commit dependencies (\emph{deadlock}), or an RW anti-dependency closing a cycle after its writer has committed (\emph{serializability violation}). 


When a cycle involves only writers, all transactions in the cycle are aborted. Otherwise, when the cycle includes readers --- as occurs when RW anti-dependencies close a cycle --- \rcc prefers a reader over a writer as the abort victim so as to reduce cascading aborts. 
For example, consider a commit-dependency chain $T_i \to T_j \to T_k \to T_l$, where an RW anti-dependency $T_l \to T_i$ closes the cycle. Aborting $T_i$ aborts all four transactions, whereas selecting reader $T_l$ as the victim breaks the cycle with only one abort.

\revisedhj{
This optimistic design offers three benefits.
First, \rcc resolves only true cycles, avoiding the false aborts of prevention-based schemes~\cite{rebirth_retire}.
Second, by removing deadlock handling overhead from the lock acquisition path, transactions proceed without synchronous cycle detection.
Finally, \rcc optimistically assumes that RW anti-dependencies will not form cycles, allowing writers to commit without waiting for their readers and thus avoiding unnecessary commit waits. 
}



\revisedhj{
To minimize wasted work caused by transactions in a cycle and their dependents, \rcc detects them using a dedicated background thread, the \emph{asynchronous deadlock detector} (\emph{ADD}).
By detecting cycles shortly after they form, ADD limits the accumulation of additional dependents, thereby reducing cascading aborts.
}
As transactions generate dependencies, \emph{ADD} collects them into its own view of the dependency graph and checks for cycles in parallel with transaction execution. Before a transaction commits, \rcc waits for \emph{ADD} to finish examining all of its dependency edges.
In practice, \emph{ADD} completes cycle examination for the committing transaction before it reaches the commit point, so the committing transaction rarely needs to wait for \emph{ADD} (Section~\ref{sec:eval-ycsb}).


\revised{
Deadlocks can still form during speculative execution before \emph{ADD} detects them.
When deadlocks or aborts are frequent, cascading aborts and wasted work can outweigh the gains of speculative write visibility~\cite{SCC}. \rcc resolves such deadlocks asynchronously with lightweight rollback without interrupting pipelined execution.
}
}

\comment{
\module{Dependency Chain and Pre-Commit.} 
When a transaction accesses the speculative version of another transaction, the former transaction is said to depend on the latter one. The former and the latter are called a \emph{dependent} and \emph{preceding} transaction, respectively. 
\textcolor{purple}{We refer to such a dependency, which constrains the dependent's commit order, as a \emph{commit dependency}.}
Active transactions accessing the same record constitute a dependency chain~\cite{CLV,PWV}, which is managed in the lock table. Even if a transaction reaches the commit point, it cannot commit when it \textcolor{purple}{has unresolved commit dependencies.} 
We refer to transactions in this state as \emph{pre-committed}. A pre-committed transaction will commit once all its preceding transactions \textcolor{purple}{on which it has commit dependencies} commit, and abort if any of the preceding transactions \textcolor{purple}{on which it has commit dependencies} aborts or re-updates its speculative version. For instance, transaction {\tt T2} in Figure~\ref{fig:stx-archi} must abort if {\tt T1} re-updates its version {\tt $R^{\prime}$} after {\tt T2} creates its version {\tt $R^{\prime\prime}$}. Otherwise, {\tt T2} experiences a dirty read anomaly.
When the leading transaction in a dependency chain --- one with no remaining uncommitted predecessors --- commits, \textcolor{purple}{the corresponding commit dependency is resolved.} 
This releases the pre-commit wait of its immediate successors, triggering their commits in turn and propagating sequential commits along the chain.
}


\comment{
\module{Assumption about SELECT FOR UPDATE}  \reminder{아래 comment 내용 중 최소한 언급}
transaction이 (insert/delete 제외) 존재하는 object를 update할 모든 record는 항상 SELECT-FOR-UPDATE (SFU) mode의 lock 획득 가정 (그리고, lock은 항상 current version 대상으로 획득 - WHY?);  {\bf Rationales:} 1) In \stx, (optional SELECT)-then-(final) update 후 immediate lock release하기 때문에 blocking 시간이 짧음. 2) Unlike existing MVCC schemes where one of two WW-conflicting transactions should abort, SFU lock 대기를 통해 WW-conflicting transaction간 serilization을 시켜서 불필요한 rollback을 방지 (물론, 동시에 한 Transaction의 blocking 시간 최소화함!!)하기 위해 의도적  (unlike existing MVCC, WW conflicting 조건 check 회피 및 abort 방지 {\bf SFU 역할/목적:} (가정: 한 TX에서 한 object O에 대해서는 한번만 read-then-modify함)  1) read시 X lock 획득 후 update 직후 X lock release를 가능하게 함. 2) read - update사이 X lock 유지를 통해서 O에 대한 접근을 serialize시킴(?); \reminder{SFU의 문제점, 단점?} \hjin{비결정적인 read-modify-write에 불리. select 후 조건에 따라 update를 진행하는 경우에는 write lock 낭비될 수도 있음}
}

\comment{ 
\begin{verbatim}
T1: UPDATE emp SET sal = sal + 10 WHERE empid = 101;
T2: UPDATE emp SET sal = sal + 20 WHERE empid = 101;
...
T1, T2 모두 101 "이전 SAL 값" (즉, read committed) 읽음.
==> commit시 first commiter/writer따라, 둘중 하나 ROLLBACK!!

반면, FOR UPDATE 경우
T1: SELECT sal WHERE empid = 101 FOR UPDATE;
    UPDATE emp SET sal = sal + 10 WHERE empid = 101;
T2: SELECT sal WHERE empid = 101 FOR UPDATE;
    UPDATE emp SET sal = sal + 20 WHERE empid = 101;
...
==> NO PROBLEM! 
SFU 선택 이유: 기존방법은 X-lock blocking 폐해 크지만, 
\end{verbatim}

\module{Two type of read}
\begin{itemize}
    \item ------------------------------------------------------------------------> Time
    \item T1: UPDATE(O1)\hspace*{2cm} Read(O1)
    \item T2:\hspace*{2cm} UPDATE(O1)
    \item ------------------------------------------------------------------------>
    \item Select, Select for Update 총 두 가지의 Read를 제공
    \item (1) T1 timestamp 값 기반 모든 객체에 대해 일관된 시점 snapshot 읽기 (Select)
    \item (2) 다른 transaction들의 final write인, 가장 최신 version을 읽기 (Select for Update)
    \item (2)를 선택하는 경우, 피할 수 있는 anomaly가 많음. 
    \item 단, "TX에서 읽고 쓰기를 동시에 하는 객체에 대해서는 반드시 select for update로 접근한다"는 가정이 필요함.
\end{itemize}
}
\comment{
260106 meeting
WAL log는 업데이트 시점에 남기기 / Pre commit record는 안 만드는 걸로
WAL log에 쌓인 순서로 dependency 확인 가능
T1이 커밋한 것으로 T2, T3까지 group commit 했다고 확인 가능
Apply -> Version install로
기존: commit 시 redo log를 parsing하여 physical record에 apply
변경: TLA의 future version을 physical record에 바로 install
install은 T1이 하고, T2&T3는 pass하게
T1~T3 group commit 시, T1이 T3의 future version을 install하고(LSN은 T1의 것으로) T2, T3는 해당 vers는 skip
}

\comment{
\textcolor{purple}{
\module{Serialization Detector}
When a transaction depends on a predecessor, \rcc defers cycle detection to commit time. This design, termed \emph{commit-time deadlock detection}, ensures that all transactions proceed through the execution pipeline without any interruption from deadlock resolution logic. However, delayed detection allows transactions involved in a deadlock and their dependents to continue executing or remain in the pre-commit state until detection. This increases wasted work and transaction latency, while additional dependents may accumulate and increase cascading aborts.
}
\textcolor{purple}{
To resolve cycles earlier outside the transaction's critical path, \rcc adopts \emph{asynchronous deadlock detection}, implemented by a dedicated background thread \emph{serialization detector} (SRD). As transactions generate dependencies, \emph{SRD} promptly collects them into its own dependency-graph view and checks cycles parallel with transaction execution. Before commit, \rcc ensures that all of its edges have been examined by \emph{SRD}. \hjin{ensure 방법: SRD가 edge 모두 검사 후 approve할 때까지 Pre commit에 잠시 대기 -> 이건 implementaion 장에?}
}

\textcolor{purple}{\module{Deadlock Handling.}}
While \stx reduces lock-induced blocking during execution, its dependency chain remains susceptible to deadlocks, shifting the bottleneck from lock contention to deadlock management. Deadlock handling is indeed a critical design concern not only for \stx but for all speculative write versioning schemes, including Bamboo and Rebirth-Retire~\cite{BAMBOO,rebirth_retire}. 

\textcolor{purple}{
Prior speculative-visibility schemes take a \emph{pessimistic} approach to deadlock handling. Bamboo~\cite{BAMBOO} adopts \emph{deadlock prevention} by extending Wound-Wait, whereas Rebirth-Retire~\cite{rebirth_retire} performs \emph{conflict-time deadlock avoidance} checking whether the resulting dependency would create a deadlock before allowing it to form. Such proactive handling is relevant to cycles involving dirty-read-induced dependencies: once a cycle has formed, aborting one victim can force all transactions in the cycle to abort, with the cascading abort further propagating to downstream dependents~\cite{BAMBOO,PWV}.
}

\textcolor{purple}{
To maximize the benefits of speculative write versioning, \rcc adopts an optimistic \emph{Deadlock Detection}. This design is well suited to \rcc for four reasons: 
First, \rcc resolves only actual cycles, avoiding the \emph{false aborts} from \emph{deadlock prevention}~\cite{rebirth_retire}. 
Second, \emph{TLA}-based deferred installation makes cascading rollback lightweight. Aborted transactions can independently discard their speculative versions without undoing shared-records updates in reverse dependency order, allowing cascading rollbacks to proceed concurrently. 
Third, the \emph{SRD} performs cycle detection asynchronously in parallel with transaction execution, whereas Rebirth-Retire performs its deadlock check as part of the conflicting lock acquisition~\cite{rebirth_retire}.
}

\textcolor{purple}{
Finally, \rcc limits abort amplification by distinguishing RW anti-dependencies from commit dependencies. Because the writer does not consume the reader's speculative value, an RW anti-dependency requires neither commit ordering nor cascading aborts. \rcc exploits this property in two ways. 
First, it preferentially selects a reader as the victim when possible, reducing the propagation of cascading aborts. This may also increase the risk of reader starvation under repeated cyclic conflicts.
Second, it allows the writer to commit without waiting for the reader's commit, shortening its commit latency and limiting the growth of commit-dependency chains. Nevertheless, RW anti-dependencies are still tracked by \emph{SRD} to detect cycles, as detailed in \S~\ref{sec-read}.
}

\textcolor{purple}{
When true deadlocks are frequent, however, this optimistic approach may allow more speculative work to proceed before a cycle is resolved, increasing wasted work and cascading aborts.
}
}


\comment{

\module{\textcolor{red}{Commit-Time and Asynchronous} Deadlock Detection.} 
While \stx reduces lock-induced blocking during execution, its dependency chain remains susceptible to deadlocks, shifting the bottleneck from lock contention to deadlock handling. Deadlock management is indeed a critical design concern not only for \stx but for all speculative write versioning schemes, including Bamboo and Rebirth-Retire~\cite{BAMBOO,rebirth_retire}. 

These schemes handle deadlocks through avoidance~\cite{BAMBOO} or synchronous detection~\cite{rebirth_retire}, both premised on the pessimistic assumption that deadlocks are frequent and require immediate intervention.  As a result, deadlock handling logic is placed directly on the latency-critical execution path. Rebirth-Retire incurs per-lock coordination overhead on every lock acquisition, regardless of whether a deadlock actually exists. Bamboo's wound-retire triggers \emph{false aborts} 
--- preemptively killing transactions before a true deadlock is confirmed --- imposing wasted work far exceeding the cost of true deadlocks. Together, these drawbacks prevent speculative write versioning from fully realizing its pipelining potential, as transactions are repeatedly interrupted by deadlock handling during execution.

\comment{To maximize the benefits of speculative write versioning, \rcc adopts an \emph{optimistic} approach to deadlock handling: rather than detecting deadlocks \emph{synchronously} at the point of conflict --- as in prevention-based schemes~\cite{BAMBOO} or per-lock detection~\cite{rebirth_retire} --- \rcc defers detection to commit time. This design, termed \emph{commit-time deadlock detection}, ensures that all transactions proceed through the execution pipeline without any interruption from deadlock resolution logic. In particular, transactions not eventually involved in any deadlock cycle commit without incurring any deadlock-related latency overhead, significantly reducing their
end-to-end latency. We opt for the design because true deadlocks are rare in practice~\cite{TPBOOK,pg-deadlock,mysql-deadlock}.

The cycle detection is triggered once per transaction at commit time: each transaction traverses its local dependency chain to check for cycles, requiring no per-conflict synchronization on the execution path. Upon deadlock detection, \rcc aborts all transactions in the cycle. While this may abort multiple transactions at once, the overhead remains negligible thanks to \rcc's lightweight abort.\swlee{SSI like optimization 언급 어디서?}

\revised{
Despite its benefits, commit-time deadlock detection can increase wasted work.
Transactions involved in a deadlock and their dependent transactions continue executing until detection, resulting in wasted work when they eventually abort. Dependent transactions waiting in the pre-commit state waste CPU cycles idling before being cascadingly aborted, degrading both resource utilization and throughput. Moreover, the number of dependent transactions may grow before detection, further increasing wasted work.

To mitigate this, \rcc adopts \emph{asynchronous deadlock detection}, which can detect deadlocks after the conflict point but before commit time, removing deadlock resolution entirely from the critical execution path. 
}
}

To maximize the benefits of speculative write versioning, \rcc adopts an \emph{optimistic} approach to deadlock handling: rather than detecting deadlocks \emph{synchronously} at the point of conflict --- as in prevention-based schemes~\cite{BAMBOO} or per-lock detection~\cite{rebirth_retire} --- \rcc defers detection to commit time. This design, termed \emph{commit-time deadlock detection}, ensures that all transactions proceed through the execution pipeline without any interruption from deadlock resolution logic. The cycle detection is triggered once per transaction at commit time. \textcolor{purple}{This eliminates per-conflict synchronization on the execution path and aborts transactions only when an actual cycle is confirmed.}
\textcolor{purple}{Upon detecting a cycle, \rcc selects a transaction as a victim to break the cycle. \rcc preferentially selects a reader, since aborting a reader avoids cascading aborts compared with writer whose speculative versions may have already been consumed by other transactions.}
\msj{writer가 reader를 대기하지 않는 장점에 대한 trade-off로 cycle에서 reader만 abort되므로 reader starvation issue에 대해서 언급 필요}

\textcolor{purple}{Despite its benefits, commit-time deadlock detection can increase wasted work. Transactions involved in a deadlock and their transactions continue executing until detection, resulting in wasted work when they eventually abort. Dependent transactions that have completed execution may remain in the pre-commit state only to be aborted after cycle detection, degrading resources and increasing latency. Moreover, the number of dependent transactions may grow before detection, further increasing wasted work. }

\textcolor{purple}{
To mitigate this, \rcc adopts \emph{asynchronous deadlock detection}, implemented by a dedicated background \emph{serialization detector} (SRD). \emph{SRD} asynchronously checks the dependency graph for cycles after the conflict rather than being deferred until commit, removing deadlock resolution entirely from the critical execution path. Under Serializable, the same mechanism also detects serialization cycles involving RW anti-dependencies, as described \S~\ref{sec-read}) \hjin{마지막 문장 굳이?}
}

}
\comment{
\jhpark{ww보다 commit-time detection이 오히려, deadlock이면, cycle 낭비 더 심한 것 아닌가?
대응: 1) ww early abort가 오히려 false abort 높음; 2) younger tx가 rollback 될때까지 synchronous wait. 3) true deadlock 비율 낮음점 - wasted work 비율이 낮음. 4) TLA discard 기반 lighweight rollback
In addition, \rcc 보다 eager하게 deadlock detection 수행하는 \rccd 성능 떨어지는 점 언급; + \rccw 보다 성능 약 2.7배 좋은 것도 언급
}
\msj{DL Detect revisit한 이유, Bamboo: wait for graph가 커진다는 이유로 Wound-wait 채택 -> }
\jhpark{
기존 mysql,postgres 에서 사용하는 DL\_DETECT 방식과 차이점;
Bamboo에서 dl\_detect할 수 없는 점;
}
\revised{Bamboo chose wound-wait to avoid the complexity of wait-for graph maintenance. What about wait-for graph in \stx?? defense 필요! 
\comment{ 
}
--- wound-wait 등의 false wait를 피하고 rare하게 발생하면 abort시키면 되는 이유로, 
--- MySQL/Postgres/Oracle 등 producion DBMS서 주로 사용함. 
}
\revised{In \stx, deadlock은 pre-commit 상태에 있는 tx들이 wait time이 detect된다.}
}

\subsection{Database Operations in \stx}

With the added functions and data structures described in Section~\ref{sec:design-obj-archi}, \stx performs normal database operations such as {\tt update}, {\tt  read}, {\tt commit}, and {\tt abort} differently from traditional DBMSs.
This section describes how \stx redefines when and how these operations interact to enable speculative write versioning with \revisedhj{speculative unlock}. 
While explaining how each database operation works in \stx, we also discuss the challenges and our solutions. 
\comment{
\begin{itemize}
\item 1) cascading abort - 데이터를 직접 바꾼게 아니라 log를 별도로 관리,  2) serialization overhead - tx chain threshold ?, 3) log/lock mgmt 오버헤드 - tx-wise 로그를 관리하기 때문에 centric log mgmt가 아니라 mutex 오버헤드 최소화할 수 있는점. 기존 Lock 구조체 수정 최소화...등등, 
4) 복잡도? : 모든 action에 대해 \stx 를 적용하는 게 아니라, update에 대해서만 적용함. 
\item  w-w conflict 검증 cost, 직렬화 보장 비용 - transaction dependency 관리 그 결과 commit latency 증가, cascade rollback, read consistency - 다른 트랜잭션에 의한 최신 업데이트를 볼 수 없음; 읽기 일관성 보장 어려움.
\end{itemize}
}
\subsubsection{Update}\label{sec:design:update}
Conventional RDBMSs serialize concurrent writes by holding the exclusive lock on a record's current version until commit.
\rcc breaks this serialization through speculative unlock.
Consider {\tt T1}, {\tt T2}, and {\tt T4} sequentially updating record {\tt $R$} in \autoref{fig:stx-archi}.
{\tt T1} acquires the exclusive lock on $R$, applies its modification, and stores the resulting speculative version $R^{\prime}$ together with its redo log entry in its own \emph{TLA}.
When {\tt T2} subsequently acquires the lock, it waits only for {\tt T1} to complete its update rather than for {\tt T1}'s commit.
{\tt T2} reads $R^{\prime}$ from {\tt T1}'s \emph{TLA}, produces $R^{\prime\prime}$ in its own \emph{TLA}, and releases the lock immediately.
\rcc records a dependency edge {\tt T1}$\rightarrow${\tt T2}, enforcing that {\tt T2} cannot commit before {\tt T1} and must abort if {\tt T1} aborts.
{\tt T4} proceeds analogously for $R^{\prime\prime}$, yielding the dependency chain {\tt T1}$\rightarrow${\tt T2}$\rightarrow${\tt T4}.
If {\tt T1} re-updates $R$ after {\tt T2} has read $R^{\prime}$, \rcc cascade-aborts {\tt T2} and its successors ({\tt T4}) to invalidate their stale reads. 
The current version in the data page remains unchanged throughout this sequence.

\comment{
In conventional RDBMSs, when a transaction {\tt T2} updates a record {\tt $R$} (in Figure~\ref{fig:stx-archi}), it acquires an exclusive lock on {\tt $R$}'s current version, captures a redo log entry to the global log buffer, immediately applies the modification in place, and holds the lock until commit.
If another transaction {\tt T1} holds the lock, {\tt T2} must wait until that transaction commits, serializing concurrent writers.

With \stx, the record update operation proceeds as follows.
In Figure~\ref{fig:stx-archi}, {\tt T1}, {\tt T2}, and {\tt T3} sequentially update record {\tt $R$}.
When {\tt T1} updates {\tt $R$}, it acquires an exclusive lock on {\tt $R$}'s current version and applies its modification. 
The resulting speculative version ({\tt $R^{\prime}$}), along with its corresponding redo log entry, is stored in {\tt T1}'s \emph{TLA}.
Unlike conventional systems that serialize access to a centralized log buffer, this transaction-wise log management eliminates mutex contention during updates, reducing logging overhead.
When {\tt T2} subsequently updates {\tt $R$}, it waits only until {\tt T1} completes its update, not for {\tt T1}'s commit. Upon acquiring the lock, {\tt T2} reads {\tt T1}'s speculative version ({\tt $R^{\prime}$}) from {\tt T1}'s \emph{TLA}, performs its update, and stores the resulting speculative version ({\tt $R^{\prime\prime}$}) in its own \emph{TLA}.  Since {\tt T1} is still active, \stx records a dependency edge {\tt T2}$\rightarrow${\tt T1} to ensure serializability, enforcing that {\tt T2} cannot commit before {\tt T1} and must abort if {\tt T1} aborts. {\tt T2} then immediately releases the lock.  {\tt T3} follows the same procedure, reading {\tt $R^{\prime\prime}$} from {\tt T2}'s \emph{TLA} and creating {\tt$R^{\prime\prime\prime}$} in its \emph{TLA}, with dependency edge {\tt T3}$\rightarrow${\tt T2}. Note that, the current version in the data page remains unchanged. Each transaction maintains its own speculative version in its \emph{TLA}, enabling concurrent progress while preserving the serialization order established by lock acquisition.
}

\comment{
\module{Lock Management.}
\stx extends the existing lock table with transaction dependency tracking while preserving the standard lock structure. 
Locks are enqueued in their acquisition order, forming a linear sequence where each holder becomes the predecessor of the next, thereby tracking commit dependencies while maintaining full compatibility with existing lock-based protocols.

\module{Insertion and Deletion.}
\stx applies speculative versioning only to update operations. 
Insertions create new records with no prior version, so WW conflicts do not arise.
Deletions are handled by setting a delete flag on the record~\cite{mysql-purge,postgres-vacuum}, allowing \stx to treat them in the same manner as updates.

\module{Index Management.}
\stx requires no modification of index structures. 
Primary key values should be chosen to be rarely or never changed~\cite{db_sys_concepts}, so primary index entries remain valid throughout speculative versioning.
For secondary indexes, entries point to the current version in the data page, not to speculative versions. 
When a transaction reads a secondary index, it copies the speculative version from the preceding transaction without traversing the dependency chain.
Secondary index maintenance occurs only at commit time when speculative versions are installed into the data page, following the same procedure as conventional systems.
}

\comment{

\module{Lock acquisition and dependency management} TX Ti가 record Ri을 SFU mode로 접근시, {\bf lock 대상:} current version에 대해 W (or SFU or new type lock) lock을 획득 후, update 진행; 다른 TX Tj가 Ri's current version에 SFU lock hold 중인 경우 (WW conflict), Tj의 update 기다린 후, SFU lock 획득. Tj가 active이면, Tj-Ti dependency edge를 dependency graph에 추가. (24.7.17 SJ 대화와 동일?)
\comment{ 
}
(Lock contention/acquisition) With regard to lock acquisition, \stx does not assume or prefer any specific implementation. Instead, it When multiple transactions are competing for a lock of a hot object, \stx simply follows the implementation of lock acquisition in specific DBMS.  \reminder{Postgres/MySQL 상세 구현은 4장서 언급; 어떤 방식이 \stx에 유리한지 규명필요?}

(Read 관점) 1) Figure X에서 보여진 것 처럼, active version 중 가장 최신 version 값을 읽음 (Serializable mode?), 2) (단순 MV 관점?) TX timestamp 기준 old version 값을 읽음 (current version에서 undo log 따라가서) 

\module{Lazy update} Ri에 대해 SFU lock 획득 후, Ti의 Ri update를 redo log로 capture해서 자신의 TLA에 남기고, and, as explained above, 실제 record update는 수행하지 않고 추후 Ti commit시 lazy하게 current version에 propagate함. 
update내용을 current version에 immediate하게 반영하는 기존 DBMS와 달리, Ti의 Ri update를 current version에 반영하지 않음. 대신 나중 Ti가 commit시 current version으로 delayed 반영됨. 
(alternative 1 and 2) \comment{ 
1. WW conflict 없을 시, current version in data page에 immediate update하고, old committed version을 undo log (that is, old version store)에 기록함. 
2. update redo log에 대해, 원래 처럼 redo log를 남기고 LSN도 동일하게 확보하고, per-tx redo log area에는 LSN 도 같이 남기고 그 update가 최종 반영될 때, pageLSN도 반영하는 방법도 있으나, TX가 long일 경우, 그의 오래된 LSN이 CP\_LSN에 left behind해서 checkpoint 지연 또는 recovery 문제가 있음.} 

(부가 장점: no WAL and undo recovery) In \stx adopting lazy updates at commit, data pages at buffer cache are guaranteed to always keep only committed versions. Therefore, when such pages are flushed to the storage, non-committed update가 storage에 미반영되므로 WAL protocol 필요없고 failure시 undo recovery도 불필요. (+ CP flush write도 적게 발생(?))  what about redo recovery?   
({\bf redo log format: record version VS. physiological log}) 만일 current version에 write lock이 set되어 있고 final write이면, updating transaction T1이 redo log를 capture해서 자신의 redo log region에 새로운 record version copy (\reminder{Or, 이전 record에 대한 pure update log (column update시 PPL 형태?}) 를 남김. redo log 형태 기존 WAL format 확장 필요? 

{\bf (Early lock release)} In \stx assuming that transactions mostly updates records only once, record update에 해당하는 redo log를 TLA에 생성 후, current version에서 active TX의 SFU lock을 바로 release한다. (lock duration: SFU W 획득 - 마지막 update (not till commit) {\bf Update Visibility in \stx} update직후 lock release하고 바로 visible함 because, 2.3장에서 언급한 것처럼, most transaction update an object only once; ((what if multiple updates?))  dependant TXs를 cascading aborts하면 됨; 또는 EOT까지 hold하는 lock mode나 final update hint mode? In addition, {\bf commit point와 같은 update visibility 시점을 결정하는 개념~\cite{PWV} 미도입 이유.} 1) 아주 낮은 확률로 logical abort, 2) 적은 cascade abort overhead

(TX에서 record update 후에 TX's Read) 각 active TX는 자신이 생성한 redo log 기반의 record version을 읽음 (repeatable read 보장)

\module{Dependency Graph and Deadlock Resolution} Owner와 dependant TX간의 dependency graph (node - tx, edge - dependency b/w TXs, cyclic DAG?)를 유지함.  MySQL/Postgres/Oracle 같은 locking 기반의 CC를 지원하는 DBMS에선느 이미 TX간의 dependency graph를 기반으로 deadlock resolution 및 XXX 용도로 사용함. \stx서는 WW에 대해 go하지만, 동일하게 dependency 관계 유지 필요; cascades abort와 deadlock detection 목적으로 dependency graph 유지가 필요함.  (\stx 목적에 추가로 필요한 정보?) (edge insertion 시점) T1-T2간 WW conflict에서 T1 update 후에 T2가 lock을 획득했을 경우, T1이 아직 active상태이면, T1-T2 edge를 dependency graph에 추가함.  

((Deadlock resolution)) Deadlock condition in \stx = cycle in dependency graph?  (deadlock detection VS. avoidance) 
둘중 어느 방식? 성능 차이 크지 않다면 이유 서술 필요; ((Deadlock Detection:)) Update마다 cycle detection cost가 크기 때문에 나중에 Deadlock Detection을 하는게 더 낫다. (deadlock detection 시점?) edge insertion or precommit or commit?
(Deadlock Avoidance:) Pre-commit단계에서 체크하면 Cascade Abort cost가 커서 미리 체크/방지 방지 중요함. cycle detection cost vs. cascade abort cost 고려할 때,\reminder{혜진: \stx에서는 어떤 방법이 유리해서 어떤 방법을 선택했음 기술할 것. ++ commit시에 deadlock check하는 실험 결과?}  \hjin{Cycle check를 하기 위해서 잡는 Global lock-Page에 latch contention overhead 때문입니다. MySQL8.0에서는 background thread가 waiting thread를 대상으로만 짧게 mutex를 잡고 snapshot을 찍은 후 deadlock detection을 진행합니다. 즉 Lock-Page Latch contention overhead > Benefit of Early abort to reduce wasted work}

Do we need {\bf limit the number of TXs depending on a transaction} to reduce the penalty of cascaded aborts? 

\module{Avoiding False Deadlocks: \stx 기회?} \reminder{포함?} \comment{ 
Some Deadlocks in 2PL --> False, \stx and Group Commit (T1/T2 동시 commit --> deadlock free임.; (논문 서버 0.O2PL 폴더 참고 - DaMon Submission Version); 그렇다면, 진짜 deadlock은 어떤 경우? 
Under 2PL, each of the following two schedules will result in a deadlock situation. But, if both transactions commit, both schedules are serializable (True? Write Skew??) . In this sense, we regard those deadlocks as false one. And, though they will form a cycle in dependency graph, they are not true deadlocks. Thus, O2PL does not kill any victim transaction. Instead, both transactions will commit under O2PL. For Case 2, the schedule is not false only if both transactions successfully commit. For this, O2PL introduces the mechanism of group commit, where two transactions depending each other can commit atomically. As a future work, we will classify the taxonomy of diverse deadlocks in O2PL and investigate how to avoid all the false
deadlocks in O2PL.
\begin{verbatim}
r1[x], r2[y], w1[y], w2[x]: true deadlock  
w1[x], w2[y], r1[y], r2[x]: true deadlock  

- r1[x]값 기반 w1[y]이면 write skew! (Under SI)
- otherwise, r1[x]과 w1[y] 무관하면 false deadlock
- ****************************************************
- CC module 구분 불가(?), thus 보수적 Write Skew 판단? **
- ****************************************************
\end{verbatim}
}

{\bf Cycle Detection in PostgreSQL, TX longevity}(포함?)
\comment{ 
\begin{enumerate}
    \item ShadowTX의 blocking을 기다릴 필요없이 go 시킨다.
    \item commit wait을 진행하지 않아서 blocking은 줄어드는데, 바닐라에 없는 cycle이 발생할 확률이 높다.
    \item Cycle Detection을 언제 할 것인지가 중요함.(commit 시점에서 detection 비교 v.s. 업데이트 시점에서 detection 비교)
    \item 커밋시점에서 진행하는 경우, Tuple1에 대해서는 T1 <= T2가 관계를 갖고,  Tuple 2에 대해서는 T2 <= T1의 관계를 가지면은, Tuple1이 먼저 업데이트돼서 winner로 하는 경우 T1도 abort해야하는 문제 발생
    \item 따라서 업데이트 하는 시점에서 eager하게 cycle detection하는 로직이 필요함.
    \item 다만 업데이트마다 위 cycle detection을 하는것은 overhead가 크므로 해결할 방법이 필요함 (주요한 Challenge)
    \item TPC-C 에서 W-W blocking 관계 표시, TPC-C에서의 트랜잭션 길이를 조사해서 트랜잭션 길이를 활용하여 최적화가 필요하다고 언급할 필요가 있음. (기존 우리에게 유리한 Scenario, T1은 Commit까지 대기하는 시간이 100초라면, T2가 commit까지 대기하는 시간이 1초인경우, T2가 수행되는데 101초가 걸림.)
    \item 위 상황이 실제 워크로드에서 발생하는 일인지, TPC-C에서는 위와같은 상황이 얼마나 발생하는지, 다른 워크로드에도 존재하는지 확인이 필요함.
    \item 그리고 만약 abort상황에서 어떤 트랜잭션을 살릴 것인지 판단을 할 때 트랜잭션 길이를 바탕으로, 제일 긴 트랜잭션만 살림.
    \item 우리 Shadow-TX 입장에서는 블락킹 안당하고 가니까, 트랜잭션 ID 생성 순서 = 트랜잭션 길이 이므로  트랜잭션 ID 생성 순서가 제일 빠른 애들만 살리고 감.
    \item Pre-Commit 상태면 User Level에서 새로운 트랜잭션 issue할 수 있게 (tpcc 수정) -> 구현 불가능?
\end{enumerate}
}

\module{Handling record insertion/deletion} \stx is applied only for updating existing objects, not for insertion/deletion. (기존과 동일하게) Record insertions are directly made at current page and will be 기존 방식으로rollbacked. For record deletion, delayed deletion으로 처리?

\module{Index management} record가 update되었을 때, primary key 또는 update되는 column에 대한 second index가 수정된 active version을 consistent하게 pointing할 필요 있음~\cite{mvcc-eval} 참고 설명필요. Primary key는 update 안되므로 문제 없음. 하지만, Postgres와 같이 O2N에서는 index entry에 multiple entries를 유지함. \reminder{update되는 column에 대한 secondary index consistency는 어떻게 보장?}
}
\subsubsection{Commit} \label{sec:commit}
When {\tt T4} commits in \autoref{fig:stx-archi}, it must wait for all its predecessors ({\tt T1}, {\tt T2}) to commit. Otherwise, if any predecessor aborts after {\tt T4} commits and makes its update durable, {\tt T4}'s effects become irrecoverable~\cite{COWBOOK}. 
As described in Section~\ref{sec:design-obj-archi}, a transaction waits in the pre-commit state until all its commit dependencies are resolved and \emph{ADD} has confirmed that the transaction is not involved in any cycle. Once both conditions are met, the transaction commits and releases the pre-commit wait of its immediate successors, triggering their commits in turn along the dependency chain. Unlike SSI~\cite{SI:TODS09}, \rcc requires no read validation at commit time, as commit ordering already ensures that all dependencies are consistent.
Upon commit, a transaction flushes its commit log record to storage and installs its speculative versions from the \emph{TLA} into the corresponding records in data pages.

In \autoref{fig:stx-archi}, when {\tt T1} commits, {\tt T2} and {\tt T4} commit sequentially, advancing the current version of {\tt $R$} from {\tt $R^{\prime}$} through {\tt $R^{\prime\prime}$} to {\tt $R^{\prime\prime\prime}$}. Note that {\tt T4} has two kinds of dependency: commit dependencies on {\tt T1} and {\tt T2}, which require {\tt T4} to wait until both commit, and an RW anti-dependency on {\tt T3} ({\tt T3} $\rightarrow$ {\tt T4}), which imposes no commit ordering. {\tt T4} can therefore commit while {\tt T3} is still executing.


\comment{
When {\tt T3} commits in Figure~\ref{fig:stx-archi}, it must wait for all its predecessors ({\tt T1}, {\tt T2}) to commit first. If {\tt T3} commits while {\tt T1} or {\tt T2} is still active, the system becomes irrecoverable should {\tt T1} or {\tt T2} subsequently abort.
To cope with this case, \stx introduces a {\em pre-commit} protocol. When {\tt T3} completes execution but has active predecessors, {\tt T3} enters a {\em pre-commit} state rather than committing immediately. In this state, {\tt T3} has finished its operations but defers durability until all predecessors commit. This ensures that {\tt T3}'s effects become durable only after its predecessors have successfully committed, preserving recoverability~\cite{CLV}.

\module{Version Installation}
Once all predecessors have committed, \stx installs the speculative versions into the data pages. 
Version installation applies the speculative version from the \emph{TLA} to the data page and appends the corresponding redo log to the WAL buffer before writing the commit record.
This process ensures that data pages are updated in the same order as transactions accessed conflicting objects.
As shown in Figure~\ref{fig:stx-archi}, when {\tt T1} commits, \stx sequentially commits its pre-committed successors ({\tt T2}, {\tt T3}) in dependency order, installing each speculative version into the data page.
The current version of record {\tt$R$} is updated from {\tt $R^{\prime}$} through {\tt $R^{\prime\prime}$} to {\tt $R^{\prime\prime\prime}$}.
Each version installation assigns an LSN to the redo log entry and stores it as the \texttt{page\_lsn}.
The WAL preserves this dependency order, allowing recovery to restore transactions in the correct sequence.
}

\comment{
\module{Asynchronous Commit}\label{sec:async-commit}
In \stx, a pre-committed transaction is likely to wait longer for the I/Os of its owner transactions, which can offset the concurrency benefits. 
ore importantly, even if a transaction is ready to commit before its predecessors, it cannot issue its commit until all predecessors commit first. This commit ordering constraint forces short transactions to wait for long ones in the dependency chain.
To mitigate this, \stx can allow the session to initiate the subsequent transaction immediately after invoking pre-commit, rather than waiting for the final commit acknowledgment.
For pre-committed transactions, \stx defers acknowledgment until all predecessors have committed and the group commit completes.
This approach is analogous to external synchrony~\cite{xsyncfs}, where the system defers durability guarantees while allowing applications to proceed.
Furthermore, this approach follows the eventually durable model~\cite{ED:VLDB24}, where systems trade strict durability for improved performance without compromising correctness.
The asynchronous commit mechanism that decouples client-side request issuance from server-side commit completion does not affect correctness, as commit ordering and dependency constraints are still enforced by \stx at the server side\reminder{redundant claim?}.
\comment{ 
Asynchronous commit in RCC decouples client-side transaction issuance from server-side commit completion, but does not alter the correctness guarantees of the protocol. Although a client may submit a new transaction immediately after a previous one enters the pre-commit state, the server still enforces commit ordering and dependency constraints. A transaction is only made visible to others after its predecessors have committed, and dependency-based ordering is strictly preserved. As a result, asynchronous commit affects only when clients are allowed to issue subsequent transactions, not the serialization order or visibility rules enforced by RCC. Therefore, asynchronous commit does not compromise serializability or isolation guarantees.
}
}

\comment{
When  transactions T1 commits in \stx, dependency 관계 때문에 all its depending TXs가 모두 commit했으면, T1 can commit normally. 
Otherwise, if any of  its depending TXs, T2, is still active, T1 could not commit normally. If T1 commits, T1 become irrecoverable if T2 aborts. 
To cope with this case, \stx introduces the concept of {\bf pre-commit}. 
When a transaction T1 is to commit and any of its depending TXs (say T2) is still active, T1 follows the {\em pre-commit} protocol. If T1 commits (that is, its changes are made durable) while its depending TX, T2, is not yet committed and thus may abort due to deadlock or system fails before commit, then the system will be irrecoverable. To avoid this, T1 is finished but durable하게 하면 안됨; 대신 나중에 T2 commit 후 T1's effect의 durablity는 delay함 (cite IBM Fast Path, CLV etc);  \reminder{In this way, with \stx, the commit records of transactions are made durable in redo log in the same order they are allowed to access conflicting objects. (commit ordering).}
단, 사용자에게는 T1 ack해서 다음 tx issue를 허용해서 concurrency를 높일 수는 있음. (cite Eventual Durability 논문)

{\bf Redo Log, LSN 관리 및 pageLSN 수정} TX commit시 TLA의 모든 redo log를 한꺼번에 log buffer에 copy하면서 LSN 부여받고 이 LSN을 각 대응하는 page LSN으로 기록 (이 경우, DB에 발생한 event에 대한 chronological order와 LSN 부여 순서가 달라져서 문제??); TLA redo log는 Ti commit시 한꺼번에 global log buffer에 copy후 삭제됨. (일종의 
{\bf Garbage collection for active versions?}) 그리고 나서, {\tt pre-commit} record를 log buffer에 남김; Transaction table의 T2's status는 pre-committed로 남김; 

{\bf update LSN 값 부여 방식 2:} (기존 동일) record update의 redo log를 log buffer에 기록하고, LSN을 할당받고, 그 LSN을 redo에 기록하고, 나중에 apply할 때 반영; 이때 이슈는, 여러 Tx가 동일 page내 동일 record를 여러 개 update할 때, 각 update LSN들을 LSN increasing 순서대로 page에 반영; Otherwise, Page\_Last\_Modified 값이 out-of-order로 기록될 수 있음; 이 경우, Recovery만 더 많이 하게 되고 recovery나 consistency 측면 문제 없을지?  

\module{WAL log management} 기존 WAL semantic/protocol 가급적 유지; active log를 언제 WAL 반영? 1) 매 update, 2) precommit시, 3) (group) commit시  

\module{Asynchronous transaction issue with pre-commit}  In \stx, with small buffer cache, a precommitted transaction is likely to wait longer for the I/Os of its owner transactions, thus concurrency could be compromised. Thus, after the current transaction in a connection precommits, the next transaction from the connection need to be initiated for higher concurrency. For any depending TX; precommitted transaction을 issue한 session의 경우, 다음 transaction를 issue 가능함. Precommitted TX에 대해 commit ack 안 줌; 나중 group commit시 ack함; 이 측면에서 precommit은 external synchrony~\cite{xsyncfs}와 유사.
\textbf{pre-commit interface} (from ED:VLDB25) \stx are willing to some failure risk in exchange for better performance. That is, it can perform additional work, including executing additional transactions, without waiting for the fast transaction to become durable.~\cite{ED:VLDB24}.

pre-committed TX는 자신의 direct owner TX에게 Redo log를 전달\reminder{혜진: 왜? 추가설명 부탁} \hjin{owner TX가 나의 redo log를 apply해줘야 하기 때문에. Lock은 전달할 필요 없고, dependency graph edge만 owner로 옮겨주면 됨.} 한 후, user에게 pre-commit ack를 전달. Upon receiving pre-commit ack, the session can issue new transaction. TPC-C Benchmark (Client) 코드 수정/실험 통한 TPS 개선 측정 필요

\module{Group commit protocol}  protocol for owner TX: depending TXs which are in pre-commit state들도 동시에 commit됨. When Ti commits, \stx commits together all Ti-dependant TXs in pre-committed status while traversing dependency graph starting from Ti in DFS/BFS(?). \reminder{1. commit record (Ti, dependantTXs id list) WAL에 atomic하게 write하면 group commit 성공?, 2. (all TX's redo logs --> redo log flush), lazy update of TXs, and commit record flushing 3가지 연산의 순서?)}

\module{Lazy update installation on group commit} 
By applying only the last update in an ordered group of conflicting writes to the data page, \stx reduces in-place updates, undo generation, and index maintenance by up to 5× for hotspot records, translating into substantial runtime and scalability gains in TPC-C workloads. 

\reminder{When a transaction is pre-committing, it flushes its redo log to the storage if its depending transactions have committed. Otherwise, it waits?}

{\bf Group commit protocol's property:} According to partial Strictness in 2PL~\cite{Partial-Strictness}, Owner TX Ti commit record를 log buffer에 insertion시에, (without waiting for Ti's commit to be flushed to disk) 모든 depending TX의 commit record도 같이 log buffer에 insert해도 무방;;

{\bf Group commit: Issues:} 1) Atomic group commit 보장 방안, 2) group commit시 REDO log flush 시점?,  3) Group commit시 deadlock detection 통해 (selective) cascade abort 수행 필요, 4) (Failure시 Precommited TX의 commit 해석:) group commit에도 불구하고, atomic하게 precommitted T0 경우 dependency graph 기록해두고 owner TXs 모두 commit하면 T0 commit으로 해석

\module{Commit Ordering Serializability} \stx는 First-Writer 중심 commit ordering에 따른 serializability 보장.

\module{TX's lazy updates atomicity} delayed update 관련, Ti가 two or more record updates를 해서 여러 개 record Oi update시, Naive하게 여러 Oi에 “delayed update로 ATOMIC”하게 반영 필요; (예: T1이 O1, O2를 update한 경우에, P1에 O1 redo log, P2에 O2 redo log apply하고, P1, P2의 페이지 상태를 TX-Consistent한 것으로 한꺼번에 mark;); 이런 TX-consistent pages들만 storage flush하면, DB storage는 TX-consistent status --> no undo recovery! (Haeder paper 참고); 

\module{Group commit에 따른 lazy update 최적화}  만일 T1 <-T2 <-T3가 모든 Ri에 대해 update를 순차적으로 수행했고, group commit하는 경우, T1, T2, T3 차례대로, Ri current version에 수정 내용을 3번 반영할 수도 있지만, 마지막 T3's redo log만 반영해도 OK. 

\module{DeWitt Group Commit에서 precommit과 유사/차이점(?)}
STX에서는 TX logic을 다 수행했고 logical abort없이 commit intention 역할 vs. DeWitt에서는 + visible point 역할;  Transaction model에 포함해서, default로 session에서 ack시 pre-commit; 나중 true-commit시 async하게 ack함; session은 다음 TX를 바로 이슈 TX concurrency 제고. 만일 app에서 true commit시 ack 원하는 경우, TX mode를 pre-commit vs. commit으로 구분; commit mode에서는 depending TX의 경우, owner TX commit시 TX commit ack return.
((Eventual Durability @ VLDB 25 참고, Rethink the sync와도 유사))
}


\subsubsection{Abort}\label{sec:abort}

Unlike conventional systems where the abort operation must undo in-place updates, \rcc simply discards the \emph{TLA} of each aborted transaction without page-level rollback, because of the lazy update installation. Allowing transactions to read speculative versions improves concurrency, but \rcc must then cascade the abort to all successors in the dependency chain when a predecessor aborts. For example, in \autoref{fig:stx-archi}, if {\tt T1} aborts, {\tt T2} must also abort because {\tt T2} has read {\tt T1}'s uncommitted update. Even so, the cascade is cheap: because no update has been installed into the data page, it proceeds \emph{in parallel} across the dependent transactions. 
To keep this safe, any transaction attempting to access a record affected by an ongoing cascading abort must wait until the abort completes.



\comment{
\stx enables high concurrency by allowing transactions to access a speculative version. However, this introduces the potential for cascading aborts when a predecessor transaction fails.
In conventional RDBMSs, cascading aborts are expensive as it requires undoing in-place modifications to data pages and can lead to unrecoverable states. 
In contrast, \rcc can handle cascading aborts efficiently by maintaining speculative versions as redo logs in the isolated \emph{TLA} rather than modifying the data page directly. Aborting a transaction simply discards its \emph{TLA} without any page-level rollback.
In Figure~\ref{fig:stx-archi}, consider {\tt T2} in a pre-committed state. If its predecessor {\tt T1} aborts, {\tt T2} must cascading abort because {\tt T2} has read {\tt T1}'s uncommitted update.\comment{,which constitutes a dirty read.}
If its successor {\tt T3} aborts, only {\tt T3} needs to rollback since {\tt T2}'s execution is not affected, and \stx removes all dependency edges pointing to {\tt T3}.
Both cases discard only the \emph{TLA} of each aborted transaction.
\comment{In both cases, abort handling simply discards the redo log entries and speculative versions from each transaction's \emph{TLA}, requiring no page-level modifications.}
Furthermore, when a predecessor aborts, \stx eagerly aborts all successors in the dependency chain to avoid unnecessary execution.
To prevent cascading aborts, \stx enforces that any transaction attempting to access a record modified by a transaction that is aborting or involved in a cascading abort must wait until the abort completes.
}

\comment{
(Cascaded Abort in \stx) 기존 방식~\cite{BHG} 대비 \stx cascade abort 동작방식; The cascade abort in \stx is light-weight in that the abort process for a transaction can be carried out by simply discarding the redo logs in TLA for transactions to be aborted due to lazy update and TLA redo log 기반

{\bf (Abort for Precommited TXs)} Cascading abort 문제 (T1 -> T2시 T2가 먼저 commit시 기존에는 해결 불가)에 대해, precommit은 임시 commit이고, T2의 update가 current version에 미반영되어 때문 undo 불필요; recovery시에서 group commit 없으면 실패한 tx로 간주가능; 단, T2의 WW blocking 관계없는 다른 object에 대한 udpate에 대해 undo 필요 (따라서, Undo log와 WAL protocol 준수 여전히 필요?); a TX’s redo logs without (explicit?) commit record는 recovery시 무효화.--> lazy update와 결합해서, pre-committed depending TXs에 대해 abort record 남기거나 (또는 단순히 commit record를 남기지 않을으로써) light-weight cascade rollback 가능하게 해 줌.
}

\subsubsection{Read}
\label{sec:read}

\stx supports two read modes.
A plain {\tt SELECT} follows the conventional MVCC snapshot read protocol, returning a committed version consistent with the transaction's snapshot timestamp by traversing the version chain without acquiring locks or creating dependency edges~\cite{mysql-mv,postgres-mv}.
Speculative versions are invisible to snapshot reads.
A {\tt SELECT FOR UPDATE}~\cite{tpccmysql,mysql-locking-read,sfu-1,sfu-2} acquires an exclusive lock and reads the latest speculative version from the predecessor's \emph{TLA}, or the committed version if no active predecessor exists.
This creates a dependency edge from the predecessor to the reader, as described in Section~\ref{sec:design:update}.
All read-modify-write operations in \stx require {\tt SELECT FOR UPDATE} to ensure updates are based on the most recent state in the dependency chain.
A transaction that mixes both modes on records updated by the same predecessor observes values from different points in time.
The snapshot read returns the committed version, while the speculative read returns the uncommitted version.
This is consistent with \vanillamy under Repeatable Read, where a plain {\tt SELECT} and a {\tt SELECT FOR UPDATE} within the same transaction can return different versions of the same record~\cite{mysql-locking-read}.
The choice between the two read modes determines how \rcc tracks dependencies and what isolation guarantees it provides, as described next.

\comment{
\stx supports two read semantics depending on the transaction's intent.
For read-only accesses, \stx follows the conventional MVCC snapshot read protocol. 
The transaction reads a version consistent with its snapshot timestamp by traversing the version chain, without acquiring exclusive locks or creating dependency edges~\cite{mysql-mv,postgres-mv}.
This preserves the concurrency benefits of traditional MVCC for read-heavy workloads.
In Figure~\ref{fig:stx-archi}, when {\tt T1} re-reads {\tt $R$} after its update, it retrieves its own speculative version {\tt $R^{\prime}$} from its own \emph{TLA}, regardless of speculative versions created by {\tt T2} and {\tt T3}.
When a transaction reads a record with the intent to update ({\tt select-for-update}), it acquires an exclusive lock on the current version and reads the latest state, which is either the committed version or the predecessor's speculative version if the predecessor is still active.
This approach is prevalent in OLTP workloads, where \texttt{select-for-update} is used to prevent lost updates in read-modify-write sequences~\cite{tpccmysql,mysql-locking-read,sfu-1,sfu-2}.
As described in Section~\ref{sec:design:update}, this mode establishes a dependency edge when reading from an active predecessor, ensuring that the transaction's update is based on the most recent state in the dependency chain.
}


\revised{
\rcc preserves the isolation levels of the underlying MVCC engine, and each level determines which dependencies \rcc tracks. 
}
Under \texttt{Repeatable Read}, \rcc tracks only WW dependencies from speculative writes. Snapshot reads and read-only transactions behave identically to vanilla MVCC systems. Thus, if a transaction uses plain {\tt SELECT}, \rcc is also susceptible to write skew~\cite{SI:TODS09}, as in conventional \vanillamy.
Under \texttt{Serializable}, \rcc additionally tracks RW anti-dependencies to prevent write skew~\cite{ansi-isolation}. 
To this end, plain \texttt{SELECT} statements are treated as \texttt{SELECT FOR UPDATE}, acquiring locks and performing speculative reads~\cite{mysql-locking-read} to detect potential RW conflicts with concurrent transactions.
In \autoref{fig:stx-archi}, when {\tt T3} subsequently performs a plain \texttt{SELECT} on {\tt $R$} after {\tt T2} updates it, {\tt T3} reads the latest speculative version from {\tt T2}, creating a WR dependency edge {\tt T2} $\rightarrow$ {\tt T3} in the dependency graph.
\comment{ Write Skew 시나리오에 대한 언급 삭제, Isolation level에서 언급할 내용 X \textcolor{purple}{For example, extending the scenario in \autoref{fig:stx-archi}, suppose that {\tt T3} subsequently modifies another record {\tt X} that {\tt T4} has speculatively read, the resulting RW anti-dependency edge ({\tt T4 $-$$\rightarrow$ T3}) creates a cycle in the dependency graph, and \rcc breaks the cycle by aborting a victim transaction to prevent write skew~\cite{ansi-isolation}. } }
Section~\ref{sec:correctness} formalizes these correctness guarantees, showing that \rcc's dependency graph is equivalent to the conventional conflict serialization graph under \texttt{Serializable}.

\subsubsection{Insert and Delete}
\stx applies speculative versioning only to update operations. 
Insertions create new records with no prior version, so WW conflicts do not arise.
Deletions are handled by setting a delete flag on the record~\cite{mysql-purge,postgres-vacuum}, allowing \stx to treat them in the same manner as updates.

\comment{
\stx supports multiple isolation levels through selective dependency tracking. For snapshot isolation, \stx tracks only WW dependencies, allowing transactions to read from their snapshots while serializing concurrent writes to the same record.
Read-only transactions behave identically to traditional MVCC systems, performing snapshot reads from their transaction start timestamp~\cite{mysql-mv,postgres-mv}.
For example, in \autoref{fig:stx-archi}, consider any transaction (say {\tt T4}) that starts before {\tt T1}'s commit and accesses record {\tt $R$}.
If {\tt T4} is a read-only transaction, it reads the current version {\tt $R$} from the data page, as speculative versions remain invisible.
However, if {\tt T4} operates under Serializable isolation or performs select-for-update on {\tt $R$}, it reads the latest speculative version $R^{\prime\prime\prime}$ from {\tt T3}'s \emph{TLA} and \stx records a dependency edge {\tt T4}$\rightarrow${\tt T3}. \jhpark{표현 수정!!!}
This read-write dependency tracking prevents write skew~\cite{ansi-isolation} and ensures serializability.
}


\comment{
\module{Two read semantics in \stx} (1) update처럼 Select-For-Update lock 획득 후 {\bf most recent version read}; 이 경우에는 current version일 수도 있고, owner tx redo version일 수도 있고, 자신이 update한 경우에는 자신의 redo log일 수도 있음. 따라서, repeatable read 보장, (2) timestamp 기준 old version 중 적당한 version 읽기;; \reminder{어떤 semantic인가에 따라 \stx 보장 isolation level, SR or SI, 결정? read/write skew 발생?} ---  (1)의 경우, read만 한 경우에도 T1-T2간 (T1이 active 상태이면) dependency edge 추가 필요함. 

\stx는 두 가지 모두 자신이 읽어야 할 version tracking 가능하고 따라서 repeatable read 보장함. 반면, 기존 MVCC에서는 current version만 읽게 하면, repeatable read 보장 못함. \reminder{기존 MVCC 등 처럼 timestamp 관리 필요?}

\module{기존 SI/MVCC와 차이점} 기존 MVCC서는 active transaction이 current version에 대해 immediate/direct update 하므로 TX Read를 위해 committed undo log를 읽음. 또한 단순 select는 TX timestamp snapshot으로 읽음 vs. \stx는 이런 snapshot read 방법도 사용가능하고 또한, update를 current version 대신 TLA에 redo log로 남기므로 current version or old version read, 또는 자신의 TLA redo log version을 읽어 repeatable read 달성.
\begin{itemize}
    \item \stx에서 hot object에 대해 각 TX가 lock 획득 순서대로 순차적으로 TLA에 future version 관리시 각 dependant는 {\bf 각 record별 current version에 대해 'SELECT-FOR-UPDATE' lock 획득 후}\hjin{Select for update 말고 exclusive lock을 쓰는 건 어떨까요?}, latest version에 대해 update 수행; recall that 이전 lock hold TX가 final W 후에 lock release함. 
    \item Select for update에서는 MVCC 최신 version 을 읽어야 함. \stx서도 최신 버전 (하지만 \stx서는 latest future version)을 읽어야 함. 
    \item Read-only transaction이 precommit 상태의 값을 읽어서, distributed tx로써 다른 site에 copy하면?(in CLV paper); \stx에서 어떻게 처리하면 되나? 
\end{itemize}

\module{SELECT vs SELECT for update}
select의 경우 내 스냅샷을 기준으로 repeatable read를 보장하기 위해 과거 값을 읽는다. select for update는 기존 select for update의 의미를 그대로 사용하여 '가장 최신의 버전'을 읽는다. 
(r: read, w: write, rw: read-and-write(select for update and update), c: commit, a: abort)
\begin{itemize}
\item S1: r1[x=100]w1[x=110]...r2[x=100]...(a1 or c1)...c2
\end{itemize}
이때 r2는 w1 이전의 x 값을 읽는다. r2는 w1의 값을 읽지 않았으므로 dependency를 형성하지 않았기 때문에 w1의 commit 여부와 상관 없이 commit할 수 있다.
\begin{itemize}
\item S2: r1[x=100]w1[x=110]...r2[x=110]w2[x=120]...(c1c2 or a1a2)
\end{itemize}
이때 r2는 w1 이후의 future version을 읽고, w2는 in-place update가 아닌 TLA에 future version으로 작성한다. r2는 w1을 읽었으므로 w1이 abort하면 cascading abort하게 된다. c2는 반드시 c1이 선행해야 한다.

write/read 관련해서, when dirty pages are flushed to the storage, 1) memory 상주 future logs들은 어떻게 관리; 반영후 write? or 그냥 반영된 내용만 write/read해도 추후 문제없을지? 2) system LCN 정보등의 변경에 따라, memory에 상주하는 future log apply시 any 문제점?

\module{Isolation Levels in \stx: Repeatable Read << \stx << Serializable}
Select For Update를 통한 Uncommitted version을 읽음으로써 기존 Repeatable Read에 존재하는 Anomaly가 해결이 가능함. 특히, 어떤 값을 Depending하여 Update하는 경우에는, Write Skew 해결이 가능함

1) Is the isolation level in \stx serializable?;  2) \stx는 RepeatableRead 보장하면서도 더 빠르게 동작? (확인 필요) 3) Different iso. level 지원, 4) Read/write skews? \reminder{아래 포함/논의?}

(SI or SSI) Unfortunately STX allows the write skew problem (cite), which ... , because the short read lock is immediately released immediately for higher concurrency. When two transactions read and update two objects interchangeable. For this reason, the default isolation level by STX is SI. STX can avoid the write skew problem and thus achieve the serializable isolation level in two ways. 1) by tracking rw dependency as in Postgres SSI, 2) holding the read long and thus blocking other transaction’s write till the commit. By swallowing the run time overhead or compromising the concurrency STX can support Serializable in addition to SI.

\begin{itemize}
\item  {\bf Write Skew in SI}: When T1 reads r1 and sets r2 while T2 reads r2 and set r1, SI results in write skew. In \stx, if read r1 and r2 as SFU, then will result in deadlock, thus avoiding WS. three alternatives: 1) every read in SFU mode, 2) RW dependency - deadlock 귀결 - detect/resolution? 
3) accept WS, 
\item  {\bf Lost update}: When T1 sets r1 and r2 as 100 while T2 sets r2 and r1 and 200, 2PL results in deadlock due to WW conflicts. In \stx, deadlock will be detected, but if we group commit either T1-T2 or T2-T1 order and apply each TX's blind redo logs, then the result will be consistent either in (100,100) or (200,200) 
\end{itemize}

\comment{ 
\begin{itemize}
    \item  Is the isolation level in \stx serializable?;  
Repeatable Read 기반?;;;  Final Write를 읽고 + Cascading Abort 지원하니까 Committed version을 읽는다고 볼 수 있음;  mysql도 동일하게 작동

\begin{verbatim}
T1: Select(a)           SelectForUpdate(a')
T2:          Update(a')
\end{verbatim}

\item \stx는 RepeatableRead 보장하면서도 더 빠르게 동작? (확인 필요)
\begin{itemize}
    \item MySQL repeatable read에 대해 SFU에 대해서는 RR과 동일한 semantic 보장
    \item What about Postgres? SI? 
    \item Truly Serializable level 지원: future work으로... challenge는 어떤 것이 있는지?
    \item Write skew 발생? 
    \item Serializable?: MVCC는 대부분 SI 지원. Can we make NV-IPL CC Serializable by leveraging IPL log?(\ie, write skew 해결 가능?)
\end{itemize}

\item Diff Isolation Levels in InnoDB: under R.R. mode(default in InnoDB), W-W conflict 도움; Serializable mode에서는 R-W conflict으로 인해 성능 감소하는데, \stx로 R.R 만큼 성능 개선 가능?; 

\item Read / Write  Skew~\cite{Berenson95} under \stx? 
\begin{itemize}
\item Write Skew in S.I. is a disguised version of deadlock in 2PL CC. Can we prevent it? How?
\item R.R. mode에서 read/write skew 문제 발생하는데, IPL log와 \stx로 해결 가능? 
\item Read Skew in S.I. 
\end{itemize}
\item  ShadowTX는 결국에는 dirty read/write라 write skew 문제는 해결할 수가 없음. 하지만 기존 PostgreSQL에서도 Serializable에서 Write skew를 해결하기 위해서 SSI (Serializable Snapshot Isolation)를 사용하여 해결하고 있음 \cite{SSIpostgres}. 따라서 ShadowTX도 SSI 처럼 별도의 Data Structure와 Lock을 도입하면 Serializable도 해결할 수 있을지?
\begin{verbatim}
r1[x], r2[y], w1[y], w2[x]: true deadlock  
\end{verbatim}

\begin{verbatim}
r1[x],w2[x],w2[y],c2,r1[y]:read skew(일종의 non-RR)
\end{verbatim}
\end{itemize}
}
\subsubsection{Checkpoint} CP in \stx? lazy update + commit에 따라 checkpoint가 덜 필요? 
}

\subsection{Columnar RCC (RCC-C)}\label{sec-columnar}
\rcc resolves WW conflicts by allowing update-conflicting transactions to be pipelined through speculative versioning with commit ordering. However, dependency tracking in \rcc operates at the record level. Even when two transactions update different columns of the same record, \rcc creates a dependency edge between them, treating them as conflicting.
For example, in TPC-C~\cite{tpcc}, \texttt{New-Order} updates \texttt{d\_next\_o\_id} of a \texttt{District} record, while \texttt{Payment} updates \texttt{d\_ytd} of the same record~\cite{columnar-locking}.
Since \rcc ~\revisedhj{handles deadlocks optimistically}, \emph{false conflicts} between \texttt{New-Order} and \texttt{Payment} --- which are independent at the column level --- accumulate under concurrent execution. 


These false conflicts, inherent to record-level concurrency control, cause two performance problems.
First, short transactions (\eg, \texttt{Payment}) in the pre-commit state must wait idle for long transactions (\eg, \texttt{New-Order}) to commit, increasing commit latency.  Second, and more critically, false conflicts at the record level can form false deadlock cycles, causing unnecessary aborts even though no true conflict exists at the column level, severely degrading throughput.

Existing approaches to column-level concurrency control, such as column-level locking~\cite{PCLL}, static analysis~\cite{IC3}, and static timestamp splitting~\cite{TSTO}, either impose prohibitive overhead or require prior knowledge of workload access patterns, making them impractical for production systems.
Prior speculative versioning schemes~\cite{BAMBOO,rebirth_retire} track dependencies at the record level and thus cannot distinguish column-disjoint updates.

\rcc maintains redo logs as speculative versions in the \emph{TLA}. Since these redo logs employ physiological logging~\cite{mysql,oracle}, each entry naturally captures which columns a transaction modifies as (offset, length, value) deltas, requiring no static workload analysis. \rcc exploits this to refine WW dependency tracking to the column level: when two transactions modify disjoint columns of the same record, \rcc does not create a dependency edge between them. At commit time, each transaction installs its redo log deltas directly into the original data page at the corresponding column offsets, updating only the modified columns in place. Since column-level deltas on disjoint offsets are independent of each other, their installation order does not affect the final record state, allowing these transactions to commit in parallel without violating correctness. We refer to this columnar RCC as \emph{\rccc}. 

As illustrated in \autoref{fig:stx-archi}, {\tt T1}, {\tt T2} and {\tt T4} sequentially update {\tt Col A} of record {\tt $R$}, forming a dependency chain {\tt T1}$\rightarrow${\tt T2}$\rightarrow$ {\tt T4}. When {\tt T5} updates {\tt $R$}'s {\tt Col B}, its delta does not overlap with those of {\tt T1}--{\tt T4}, so \rccc creates no dependency edge between {\tt T5} and {\tt T1}, {\tt T2} and {\tt T4}. {\tt T5} reads the current version of {\tt $R$} from the data page (since no speculative version of {\tt Col B} exists in any \emph{TLA}), stores its own delta as a speculative version ({\tt $R^{*}$}) in its \emph{TLA}, and commits independently without waiting for {\tt T1}, {\tt T2}, {\tt T4}. 
If \texttt{T5} were to read {\tt Col A} speculatively, \rcc would create a dependency edge {\tt T4} $\rightarrow$ {\tt T5} as in base \rcc. When the modified columns of two transactions overlap, \rccc falls back to record-level dependency tracking, preserving the same correctness guarantees as base \rcc. 

\comment{ 
Note that columnar refinement applies only to WW dependencies. Under Serializable, RW anti-dependencies are tracked at the record level to ensure sound write skew detection.
Each transaction reads records at the row granularity, while version installation applies only the modified columns from the redo log.
}
\rcc's redo log-based design enables column-level fine-grained concurrency control without the prohibitive overhead of per-column lock management~\cite{PCLL} or the static workload analysis~\cite{IC3,TSTO}. By eliminating false dependencies between transactions accessing disjoint columns, \rccc avoids unnecessary pre-commit waits, false deadlocks, and the resulting aborts, improving both transaction latency and throughput over base \rcc.
\rccc is currently supported only on MySQL, as it requires physiological redo logs that capture column-granular update deltas. 
\revised{
PostgreSQL's WAL stores full tuple images, and DBx1000~\cite{dbx1000-logging, taurus} copies the entire tuple on each write.
Neither system maintains column-level deltas, making both incompatible with \rccc without architectural changes to their logging layer. We leave such extensions as future work.

}



\comment{
    \revised{WW conflict는 column별로 관리하고 스냅샷은 row 단위로 관리함. 스냅샷 시점 자체는 처음으로 row lock을 획득하는 시점임. + 같은 row에 다른 column을 업데이트하는 경우는 SFU를 다시해야함. T1: SELECT colA, 
    T2: SELECT colB, T3: SELECT colA, colB인 경우에는 T3가 디펜던시 엣지를 두 개 생성}
    \hjin{columnar CC의 장점+ - local version 불필요; tuple level에서는 predecessor가 다른 칼럼을 수정했을 때 해당 내용만 redo log에 남게 되므로, 다른 칼럼의 speculative version을 읽기 위해서는 모든 predecessors의 redo log를 확인해야 목적 column정보를 읽을 수 있었음. 하지만 columnar cc로 column단위로 dependency를 관리한다면 그냥 바로 앞 predecessor의 redo log만 즉시 apply하여 읽기만 가능 like MVCC(version 저장 불필요 - 실제 read 발생 시에만 apply) / 실험 진행 필요}
    \msj{T1(Col A) -> T2(Col B) -> T3(Col A) -> T4(Col B) -> T5(Col A,B)의 경우, dependency를 가지는 모든 칼럼에 대한 최신 버전만 redo log apply 가능(Col A: T3, Col B: T4)? 즉, apply 시간은 column 수 +  boudning }
}


\revised{
\subsection{Implementation}\label{sec-impl}
We have implemented \rcc on three systems: MySQL~(v8.4) with N2O MVCC, PostgreSQL~(v16.2) with O2N MVCC, and the in-memory research prototype DBx1000~\cite{dbx1000}.
Due to space constraints, we describe the MySQL implementation in this section.
\footnote{Implementation details for MySQL, PostgreSQL and DBx1000, along with experimental results on PostgreSQL, are available in our technical report~\cite{rcc_arxiv}.}
In MySQL, \rcc extends the lock manager to add dependencies to \emph{ADD}’s graph and to maintain a bitmap for columnar dependency tracking.
We preallocate 4 KB for an initial \emph{TLA} heap when each transaction object is initialized, and reuse it across transactions to reduce memory allocation overhead during speculative version construction.
We extend the commit path with a pre-commit state, where a transaction waits until its commit dependencies are resolved and \emph{ADD} completes asynchronous cycle detection. Then the transaction applies the redo logs to data pages using the recovery module. 
}



\comment{ 
\stx also reduces I/O overhead as a side benefit of its design. First, eliminating WW conflicts shortens transaction lifetimes. In vanilla systems, a transaction waiting for a write lock remains active for longer. During this period, its previously modified pages can be flushed to disk through checkpointing or buffer eviction, and each such flush requires WAL writes including full page images. \stx removes this blocking period, so transactions complete faster and fewer page flushes occur during transaction execution. Second, \stx's deferred update policy keeps data pages clean from uncommitted modifications. Vanilla systems apply updates in place immediately, so flushing a page with uncommitted modifications requires persisting redo logs, undo logs, and associated page images. \stx stores speculative versions only in the \emph{TLA} and installs them at commit. These pages remain unmodified, which reduces page write traffic and the associated WAL overhead.
}

\comment{
\jhpark{251102 meeting: 
- recovery 기능 어떻게 서술? (lazy apply, pre-commit thread)
- precommit record에 트랜잭션간 dependency 정보를 넣으면서 최적화 방안 탐구
precommit은 async하게 동작; 
T1 - T2 - T3;   pc(T3)-pc(T2) 정보를 미리 알고 있으면 logging/recovery 관점에서 어떤 효과?
committed 로그에 대해서만 WAL로그에 남기는 net effect?

RecA: T1<-T3 / RecB: T2<-T3<-T4 의 conflict 상황이 있을 때
1. T4 commit 시도: T3와의 conflict가 해제되지 않았으므로 pre-commit 후 commit-wait
2. T3 commit 시도: T1과의 conflict가 해제되지 않았으므로 pre-commit 후 commit-wait
3, T1 commit 시도: 커밋 가능, commit record 남기고 lock 해제 -> wake_up(T3)
4. T3: T2와의 conflict 해제되지 않았으므로 다시 commit-wait
5. T2 commit 시도: commit record 남기고 lock 해제 -> wake_up(T3)
6. T3: 모든 conflict 해제되었으므로 future log apply 후 commit record 남기고 -> wake_up(T4)
7. T4: commit
위 시나리오 상 commit record는 (T1, T2, T3, T4) 순서로 남게 됩니다. (T1과 T2 끼리는 순서 상관 없음)
precommit 없으면, T1,T2,T3,T4 각각 커밋할 때 wal buffer에 있는 wal mutex, flush 없이도, 
}
}
\comment{
Like conventional MVCC systems, \stx employs ARIES-style recovery with redo and undo logs to ensure atomicity and durability.
In MySQL, \rcc defers log materialization until all predecessor transactions have committed.
Once dependencies are resolved, \rcc copies redo logs from the \emph{TLA} to the global WAL buffer, assigns LSNs and applies changes to data pages.
Since the \emph{TLA} is volatile, crash recovery relies solely on persisted WAL records, following the same ARIES recovery.
In PostgreSQL, speculative versions are stored as separate tuples rather than \emph{TLA} entries. However, the same WAL protocol applies, as tuple creation and visibility metadata follow the same vanilla logging scheme~\cite{COWBOOK,postgres-recv,mysql-recv}.
}

\comment{ 
Like conventional MVCC systems, \stx employs ARIES-style recovery with redo and undo logs to ensure atomicity and durability.
While \vanillamy applies updates immediately and generates logs during execution, \stx defers log materialization until pre-commit.
While \vanillamy applies updates immediately and generates logs during execution, \stx defers log materialization until all predecessor transactions have committed.
Once dependencies are resolved, \stx copies redo logs from the \emph{TLA} to the global WAL buffer, assigns LSNs, and applies changes to data pages.
This approach maintains identical WAL content to vanilla systems, enabling standard ARIES recovery.
In PostgreSQL, speculative versions are stored as separate tuples rather than \emph{TLA} entries. However, the same WAL protocol applies, as tuple creation and visibility metadata follow the same vanilla logging scheme.

Upon crash recovery, \stx follows the same three-phase process as vanilla systems: Analysis, Redo, and Undo.
During Analysis, \stx scans the WAL to identify committed and uncommitted transactions. 
The Redo phase replays all committed logs to restore data pages. The Undo phase rolls back uncommitted transactions using undo logs. 
In MySQL, since the \emph{TLA} is volatile and lost upon crash, recovery relies solely on persisted WAL records.
In PostgreSQL, speculative tuples persist on disk, so it simply replays WAL records to restore visibility metadata.
As a result, \stx can integrate tightly with the existing recovery logic of vanilla MySQL and PostgreSQL, requiring no changes to the core recovery mechanisms. 
}

\comment{
Like conventional MVCC systems that supports ARIES-style recovery, \stx also employs redo and undo logs to ensure data atomicity and durability. 
Updates are applied immediately in both vanilla MySQL and PostgreSQL.
In contrast, \stx retains modifications in the \emph{TLA} and applies them to the data pages only when the corresponding transaction invokes a commit operation.
When a transaction pre-commits, \stx copies the corresponding redo logs in the \emph{TLA} to the global redo log buffer. At this point, \stx assigns LSNs, applies the changes to the current version of the data page, and updates the page\_lsn. 
Thus, similar to \vanillamy that employ per-page basis logging scheme, \stx adopts a lazy logging, where log records are materialized and persisted only when the transaction invokes a pre-commit operation.
This approach enables both efficient version construction and robust recovery.
While \stx, like conventional systems, employs WAL to ensure atomicity and durability, relying solely on WAL for version construction is inefficient, as the log is maintained at the transaction level rather than the page level. 
Constructing future versions of a data page in this manner requires scanning the entire WAL to extract relevant modifications, incurring significant overhead, particularly when log records must be retrieved from disk. 
To overcome this limitation, \stx introduces a dedicated \emph{TLA} structure that supports efficient page-level version construction without the need for costly WAL scans, thereby achieving rapid version build while preserving the reliability of traditional WAL-based systems. 
Furthermore, by deferring the propagation of updates until all transaction conflicts and dependencies are resolved at commit time, \stx guarantees that only consistent and finalized changes are recorded and applied to data pages. 
This design simplifies recovery, as only committed updates need to be replayed, and enables the batching of related updates during group commit, which amortizes log write costs and further streamlines the recovery process.

\subsubsection{Prototype Implementation}
We implement the recovery scheme on \vanillamy, which employs ARIES-style recovery~\cite{COWBOOK}.
When a transaction commits, \stx copies its redo log entries from the \emph{TLA} to the global WAL buffer, assigns LSNs, and applies the changes to data pages. 
Since the WAL contains identical information to \vanillamy, \stx leverages the \vanillamy design with minimal modifications.
The ARIES-like recovery module both MySQL and PostgreSQL involves three steps: Analysis, Redo, and Undo~\cite{innodb-recovery, postgres-recv}.
Here, we outline how \stx expedites each phase by introducing the \textbf{TLA}, enabling \textbf{lazy logging}.

\module{Analysis}
\stx scans the WAL log file from the last checkpoint to identify committed transactions for redo and uncommitted transactions for undo.
While scanning and parsing redo logs, the redo log entries are copied to a volatile buffer and indexed into a hash table with the page as the key for fast lookup.

\module{Redo}
Once redo logs are scanned and parsed, the recovery process replays the redo logs from the earliest point in the WAL file identified in the analysis phase.
Since \stx materializes \emph{TLA} contents to the WAL at commit time, the redo logs contain identical information to \vanillamy.
Upon completion of the redo phase, \stx restores the state of the system to the state it was in at the time of failure. 
In addition, \stx also restores on-disk organization structures (\ie, rollback segments, undo segments) for undo as in \vanillamy.

\module{Undo}
When the redo phase is complete, \stx rolls back all effects of uncommitted transactions to ensure atomicity.
This follows the same approach as \vanillamy, utilizing undo logs to revert all modifications made by transactions that were active at the time of the crash.
Asynchronous rollback threads handle this process, ensuring that the database returns to a consistent state.
\jhpark{However, because \stx applies updates only at commit time, transactions that were active but not yet committed at the time of failure have no page modifications to undo.}
}


\comment{
\reminder{NV-PPL 버전: SQLite/PPL, /SSL 도 참고해서 }
\jhpark{
1. \stx에서 apply를 하면 LSN 관리는 어떻게 하는지;  (-> Ti precommit시 Ti TLA에 있는 redo log들에 대해 gloabl redo log buffer에 copy하고 LSN 부여받고, 이를 각 data page current version에 반영함. page\_LSN으로 setting)    
2. atomicity durability 관점에서 문제가 없는지.
3. \stx redo log랑 WAL로그 별도로 관리하는 이유 설명 
성재: System Redo log는 sequential하게 저장된다. System Redo log는 TRX 별로 또는 Page 별로 저장되지 않기 때문에, 추후 미래 버전을 생성하는 과정에서 Sequential log를 전부 scan해야하는 (또는 log index를 활용해야하는) overhead가 있다. 따라서 트랜잭션 별로 자신이 업데이트한 내용에 대해서 저장하는 apply용 리두로그가 필요하다. (MySQL에서 구현되어있는 Undo log도 트랜잭션 단위로 저장하는 이유가 , rollback의 단위가 트랜잭션 단위이기 때문에 별도로 저장하는 것과 비슷하다고 말할 수 있다면? 좋을 것 같습니다. )
System Redo log는 I/O overhead가 발생한다. System Redo log의 경우 Redo log buffer -> Redo log file에 저장하게 된다. 커밋되지 않은 트랜잭션에 대한 redo log가 먼저 DIsk로 내려갔다면, 미래버전을 생성하기 위해서 I/O overhead가 발생하므로, apply용 redo log를 메모리에 별도로 관리하는 것이 필요하다.
}

System failure is inevitable in large-scale database servers, particularly in cloud environments where hardware failures are prevalent~\cite{azure-recv,instantrecovery}. Upon system restarts, databases typically undergo a multi-phase recovery process to ensure the atomicity and durability of transactions.

However, \stx simplifies the restart procedure compared to conventional systems. Updates are made durable in NVDIMM, clustered by page, and ordered chronologically. This structure aligns effectively with existing recovery schemes such as ARIES~\cite{Aries}, where pages serve as the recovery unit.

\module{일반적인 ARIES style recovery assume} 

\module{Precommit 관련 recovery} (초안)트랜잭션들이 Precommit을 진행한 것을 Depending TRX와 함께 Redo log에 남김. 예를들어 TX 1에 depending하고 있는 TX 2가 Precommit을 진행하면 {Precommit, TRX 2, TRX 1}상태를 남김. TRX 1에 커밋하고 System Crash가 나서 recovery를 진행하면 Redo Log를 읽어서 TRX 2도 commit처리함

\module{Group commit 관련 recovery}

The {\bf transactional recoverability} indicates the property that all of a committed transaction's writes are made durable, and that none of an aborted transaction's writes are made durable or observed by committed transaction~\cite{PWV,BHG} In this sense, \stx satisfies both of these conditions. 특히, 두 번째 관련해서, update는 durable하지 않고, 진행중인 transaction에 의해 observed write는 결국 owner transaction이 commit해야 eventually valid하기 때문. 

\module{Recoverable schedule~\cite{BHG,PWV}} This discipline prevents T1 from aborting after T2 (which read T1’s data) has already committed. Schedules that satisfy this property are called recoverable~\cite{BHG}. 이 관점에서 \stx는 recoverable schedule을 지원함. 

\module{TCC 관련 (포함?)}
\reminder{Recovery 구현 overview - 필요하면 아래 uncomment}
\comment{ 
\subsubsection{Prototype Implementation}\label{sec:recv}

\stx leverages the existing \vanillamy design, making it compatible with other DBMSs that use ARIES-style recovery.
By using durable per-page redo logs in NVDIMM, NV-PPL enables redo-less and undo-less recovery for most \revised{PPL} pages.
The ARIES-like recovery module in MySQL involves three steps: Analysis, Redo, and Undo~\cite{innodb-recovery}.
Here, we outline how \stx expedites each phase through the use of persistent PPLs in NVDIMM.

\module{Analysis}  \stx constructs a mapping table by scanning the PPL region in NVDIMM and parsing the redo log file. This process identifies the PPLs in NVDIMM, enabling instant access to them during page read operations. Additionally, \stx records the transaction IDs of active transactions at the time of failure, ensuring proper handling of PPL pages modified by these transactions.

\module{Redo phase}  \stx handles redo recovery differently for two types of pages: normal and \texttt{PPL} pages. For all normal pages, recovery follows the same procedure as in \vanillamy. However, for most \texttt{PPL} pages, \stx can achieve redo-less recovery. Since all updates for each \texttt{PPL} page are already durable as PPLs, \stx can skip applying redo logs for those pages. Thus, any changes made to \texttt{PPL} pages by committed transactions are guaranteed to persist even in the event of crashes.

When a page is de-PPLized, it ensures that the page has been safely flushed to disk before removing the PPLs of the page.
If a failure occurs during the flush process, it is crucial to retain updates made before the de-PPLization to ensure their persistence.
To do this, \stx orchestrates the interaction between PPLs and redo logs in the WAL file by introducing a \texttt{ppl\_lsn} that records the LSN of the latest PPL.
For a de-PPLized page where the \texttt{ppl\_lsn} is greater than the lsn in the page header, it indicates that the page has not been safely flushed to disk. In this case, \stx first applies PPL to the old disk page  and then replays the redo logs from the WAL file to rectify the situation.
Conversely, if the \texttt{ppl\_lsn} is smaller than the lsn in the page header, \stx skips applying redo logs and returns the PPL blocks of the page, since the de-PPLized page has already been persisted to disk.

\module{Undo phase}  \stx restores uncommitted updates for normal pages using undo logs while it skips replaying PPLs from incomplete transactions for PPL pages. This distinction arises because \stx keeps the original pages on disk unchanged, while PPLs are persisted in NVDIMM, enabling the identification of uncommitted transaction IDs within the PPL.
The undo process in \stx similarities with the logical revert mechanism used in \cite{azure-recv}, which restores the committed version to the data page. Unlike this undo-less scheme\cite{azure-recv}, \stx allows the system to simply skip PPLs associated with incomplete transactions.

The recovery process in \stx is {\em idempotent}. \stx determines the need for redo application by comparing the page's LSN with that of the log records. If the page's LSN exceeds that of the log records, it indicates that the page has not only been flushed but also that any applicable PPL records have been applied after the application of this record, ensuring idempotency within \stx.
}
}

\comment{
\subsection{Limitations} \reminder{Omit?}
While \stx effectively reduces lock contention and improves transaction throughput through speculative versioning, it introduces certain trade-offs that we discuss below.

\module{Page Split and Merge}
 \stx does not apply speculative versioning to page split and merge operations. Unlike tuple updates, these operations require immediate installation to maintain page consistency.
 \stx thus follows the vanilla protocol for such operations.
 Extending speculative versioning to structural modifications remains an interesting direction for future work.
 
\module{Deadlock and Cascade Abort} \stx may encounter more deadlocks in certain scenarios, as transactions holding speculative versions can block other transactions waiting for the same resources.
When predecessor transactions abort, all dependent transactions must also abort.
This overhead increases with longer dependency chains and interleaved update patterns. 
However, because \stx does not install speculative updates to data pages, aborting a transaction requires only discarding its \emph{TLA} without modifying the underlying pages. 
The concurrency benefits of \stx outweigh this overhead due to the lightweight abort scheme (detailed in Section~\ref{sec:perf:serializable}).
}

\comment{ 
\begin{itemize}
\item Scalable to Core \#?: if so, why?
\item Robust to access skewness?
\item Limitations? (Review crititics?) Concurrent short vs. long transaction interference? 
\end{itemize}

\module{On Final Write} \reminder{아래 포함?}
Until now, we made two assumptions. 1) transactions access each object only once, and 2) (read+update의 경우) 항상 SFU로 접근. 1)의 경우, 실제 많은 경우에 성립함. 그리고, \stx에서 read의 경우, 여러번 하더라도 동일 redo log version을 접근함으로써 repeatable read 보장해서 문제 없다. 다만, 특정 object를 owner TX가, depending TX가 존재하는 경우, 다시 update하면, all depending TXs should be cascadingly aborted. 이런 경우가 자주 발생하면, 과도한 abort가 발생하고, aborted tx의 execution을 waste시키게 됨. 논문서 다루지않지만, final access user hint나 auto-detect 결합하면 {\bf (LLM-assisted detect가능??)}, abort/waste 회피하면서도 동시성 최대화 가능.
\begin{itemize}
    \item Final access hint/auto-detect 가능~\cite{PWV,BAMBOO,tx-assumption}: Canned transactions에서 static analysis 통해~\cite{CLV}. Bamboo justification 참고. 
    \item 또한, better concurrency를 위한 user responsibility 부각됨 - 아래 URL은 CIDR'23~\cite{cidr23-cc} 제시한 weak isolation에서 CC 문제점 및 consequence 설명 \\ 
(\url{https://docs.google.com/spreadsheets/d/17Kt55T7dwadYnfZMYfMwm7bXJvn7v9GFBUnaOZNHrIY/edit#gid=2053744148})
\end{itemize}

단, final write의 경우, hint/detect를 통해 explicity하게 ensure하면 불필요한 수행 및 rollback을 막아서 \stx 장점 극대화. 

((Final Read/Write hint 불필요)) 지금까지 final W/R hint를 assume했지만, \stx can still work without such hints. When an owner transaction T1 update an object, \stx can simply abort all depending (with respect to the object) and not-yet-commmitted transactions  (<-T1 commit없이는 T2 commit 못함.) Of course, compared to the case with final W hint,  불필요하게 rollback을 하는 overhead, though very rare, 가 있음. In the case of final R, T1의 O1 last read 후 T2가 O1을 업데이트 해도, T2가 current version을 덮어쓰지 않고, T1은 자신의 redo log를 통해 자신의 version 값을 읽을 수 있어서 final R hint 없이도 repeatable read를 자연스럽게 보장함.
}

\comment{
\subsection{Optimistic Locking}
\label{sec:opt-locking}
Optimistic locking is a new locking protocol that leverages final
access patterns to allow eager lock releases to allow more concurrency. Updates are
made immediately available once the locks are released. This allows speculative
reads and necessatates transactions to commit in dependency order to guarantee
serializability. In this section, we formally define and analyze optimistic
locking,  followed by discussions.  Section~\ref{sec:impl} then provides a
concrete design for implementing optimistic locking in modern main-memory
database systems.





\module{Protocols}

\module{Correctness}
Now we formally define and prove optimistic locking's correctness properties.  We
begin defining our transaction mode, final accesses and schedules that
consist of them. 

Our proof follows a similar methodology to that of partially strict 2PL (early
lock release) by Soisalon-Soininen and Yl\"{o}nen~\cite{PS-in-2PL}.

{\bf Final Access}: Action $Wi[x]$ or $Ri[x]$ (Write or read of record X by transaction $Ti$) is {\em final}
if no read or write accesses from $Ti$ occur before $Ti$ attempts to commit. Final read and write actions are
respectively denoted as $FRi[x]$ and $FWi[x]$.
 
{\bf Final access history}: a history $H$ with the following conditions hold:
 
\begin{enumerate}[leftmargin=*]\setlength\itemsep{0em}
\item Whenever $H$ is in the form: \dots $Wj[x]$ \dots $Ri[x]$ \dots, where
$Ri[x]$ reads from $Wj[x]$, then $FWj[x]$ precedes $Ri[x]$ but occurs after
$Wj[x]$.

\item Whenever $H$ is of the form: \dots $Wj[x]$ \dots $Wi[x]$ \dots, then
$FWj[x]$ precedes $Wi[x]$ but occurs after $Wj[x]$.

\item Whenever $H$ is of the form: \dots $Rj[x]$ \dots $Wi[x]$ \dots
(anti-dependency), then $FRj[x]$ precedes $Wi[x]$ but occurs after $Rj[x]$.
\item Commit ($Ci$) or abort ($Ai$) actions always occur after final access actions.
\item For each case $1-3$: $Ci$ must proceed $Ci$ or $Ai$.
\item For cases $1$ and $2$: if $Tj$ aborts then $Ai$ must follow $Aj$ (cascading abort)
\end{enumerate}

Final access histories are recoverable.

\begin{proof}
Let $H$ be a final access history. Assume that $H$ is not recoverable.
Then, by definition, it is of the form: \dots $Oj[x]$ \dots $Oi[x]$ \dots $Ci$
\dots, where $Ok[x]$ can be $Rk[x]$ or $Wk[x]$, and $Oj[x]$ and $Oi[x]$ are
conflict operation; and $Cj$ does not occur in $H$. Because $H$ is a
final access history, by Definition 2 (conditions $1-3$), $FOj[x]$ must
precedes $Oj[x]$. Thus, by condition $4$ of Definition 2, $Cj$ msut precede $Ci$
in $H$ which contradicts with the assumption that $Cj$ does not occur in $H$.
\end{proof}

Final access histories are equivalent to read committed histories.

\begin{proof}
Let $H$ be a final access histoy. For the sake of contradiction, assume that
$H$ is not a read committed history. Then, by definition, $H$ must exhibit action
sequences in the following form: \dots $Wj[x]$ \dots $Ri[x]$ \dots $Wj[x]$ \dots
$Ci$ \dots, where $Ri[x]$ reads from $Wj[x]$ but another $Wj[x]$ occurs after
$Ri[x]$.  As $H$ is a final history, $FWj[x]$ must proceeds $Ri[x]$ (by
condition $i$ in Definition 2).  By Definition 1, no $Wj[x]$ action is allowed to
occur after $FWj[x]$, contradicting with previous assumption.
\end{proof}

A final access history H exhibits commit ordering guarantee if actions in
Condition 3 of Definition 1 are absent in $H$.

\begin{proof}
By Definition 2, $Cj$ always preceeds $Ci$ for conflicting operations and $Oj[x]$
always precedes $Oi[x]$. Therefore, final access histories give commit ordering.
\end{proof}

A final access history H that exhibits commit ordering is serializable
~\cite{raz:93}.

A more detailed description about commit ordering can be found
elsewhere~\cite{raz:93}.
}

\comment{ 
\reminder{Columnal CC and Discussion : 필요하면 아래 uncomment}
\subsection{Discussions}
a spectrum of protocols that unify the design space
\begin{itemize}
    \item \stx가 N2O랑 O2N 중 어떤 architecture에서 더 효과적인지.
    \item write skew issue: \stx에서 
    \item cascade-abort vs. optimistic transaction processing via resolving write-write conflict
    \item Cloud-Native DB와 Fit: Aurora, Socrates, Neon 등에서 redo log 기반 동작하는데 CC는 기존과 동일함. redo log와 commit에 따라 durability 보장하고 no checkpoint / replacement write 등과 잘 어울림. 
    
    \item  Page Split:  page split 이 일어나는 테이블에 대해서는 Shadow 불가능 (orderline, orders)
\end{itemize}

\subsection{Columnar Concurrency Control}
\stx에 의한 새로운 기회요인으로, Tuple-granule CC를 columnar로 확장해서, 
false dependency and false deadlock avoidance 완화를 통한 동시성 제고; 실제 TPC-C, LinkBench에서 column-level에서 contention이 아님에도 불구하고, tuple-granule CC에 의해 blocking 경우 빈번함({\bf 실제 stats 측정 후 제시 필요})
\begin{itemize}
    \item TPC-C Benchmark 수행중 DBMS에서 Read Column 식별/구분 가능? (for Columnar CC) 
    \item 같은 레코드의 다른 칼럼을 업데이트하는 경우, 롤백이 발생하면 롤백을 해야하는가?
    \item 같은 레코드의 다른 칼럼을 write하는 경우도, 충돌로 봐야하는가? -> 이 경우는 서로 다른 칼럼을 업데이트하는 것이기에, w-w blocking이 아니라고 봄
    \item 따라서 Sel for update도 대기하지 않아도 되며, pre-commit도 하지 않아도 되어 트랜잭션 동시성을 더 높일 수 있다.
    \item 일반적으로 T1이 O1.a' 수정하고, T2가 O1.b'를 수정하면 T1 abort시 T2는 cascade abort할 필요가 없을 수도 있음. vs. tuple-grained CC에서는 dependency 관계; but column-grained에서는 false dependency일 가능성 높음. {\bf BUT,} write는 redo log에서 칼럼별 구분 가능하지만, read 관련해서 {\bf 각각 어떤 column을 읽었는지 CC module서 구분 불가}. T1이 update한 O1.a' column값을 읽고, T2가 O1.b -> O1.b' 수정시, T1 abort하면 T2도 abort해야 함. 
    \item 본 논문에서는 MySQL의 경우 메모리에 추가적으로 남기는 Future version에 대한 Redo log를 활용하여 다른 칼럼에 대한 w-w blocking을 판별하고 회피할 수 있다.
    \item 진짜 w-w blocking(같은 column에 대한)인 경우를 판별을 해야하는 경우 오버헤드가 발생할 수 있는데, 추후 방법을 찾아야함.
    \item YogaByte DB의 경우,  이미 Column-level로 fine-grained locking을 사용하고 있음. -- \href{https://vladmihalcea.com/yugabytedb-column-level-locking/}{YugabyteDB Column-level Locking – Vlad Mihalcea}
    \item 보통 자기 칼럼 읽어서 자기 칼럼 쓰는 거 → 어보트 안 해도 됨 (tpcc) 하지만 다른 칼럼 읽어서 자기 칼럼 쓰면 어보트 해야 함 RW 디펜던시
\end{itemize}

\begin{itemize}
    \item Why revisit Column-level CC in \stx?? Redo log 기반이라 overhead 거의 없이 column-level CC 지원 가능(?) : Google Spanner (Becoming SQL SIGMOD '07 / F1 paper VLDB '13) 
    \item  ...
    
\item {\bf ChatGPT:} 기존 concurrency control 방법들 중 tuple granule이 대부분인데, column granule 이유? 

ChatGPT의 말: 대부분의 DBMS에서 tuple-level (row-level) granularity를 concurrency control(CC)의 기본 단위로 삼는 이유는 단순한 구현뿐 아니라 효율성과 충돌 가능성 최소화의 균형 때문입니다. 반면, **column-level CC(column-granular locking 또는 multiversioning)**이 거의 사용되지 않는 이유는 아래와 같은 이론적/실용적 제약 때문입니다.
\begin{enumerate}
\item {\bf 읽기/쓰기 충돌 경계가 불명확}: 예를 들어, T1이 column A를 읽고, 다른 T2가 같은 row의 column B를 갱신한다고 하더라도, 내부적으로는 같은 tuple을 다루는 것이므로 row-level isolation을 깨뜨릴 수 있습니다. 실질적으로 대부분의 DB 연산(예: WHERE, JOIN, UPDATE)은 여러 컬럼에 걸쳐서 동작하기 때문에, 실제로는 column granularity가 충돌을 줄이는 효과가 거의 없습니다.

\item  Locking/MVCC 구조가 복잡해짐: Row-level MVCC에서는 tuple마다 버전이 하나만 있으면 되지만, column-granular CC를 하려면 각 컬럼마다 버전/lock 정보 따로 저장 해야 함. 이는 **version vector (e.g., one per column)**를 유지해야 하며, 동기화, GC(garbage collection), undo/redo 등 모든 작업이 훨씬 복잡해집니다.

\item 성능 저하 (metadata overhead): row마다 lock을 거는 대신 column마다 걸게 되면, lock의 개수가 폭발적으로 증가하며, lock manager의 부하와 메모리 사용량이 급증합니다. 특히 OLTP 시스템에서 짧은 트랜잭션이 많은 경우, 이런 오버헤드는 감당하기 어렵습니다.

\item 일관성 보장 관점에서 의미가 모호함:  많은 경우 애플리케이션은 row 단위로 의미 있는 **불변 조건(invariants)**을 갖습니다 (e.g., balance >= 0).

컬럼 단위로 locking이나 versioning을 하게 되면, 이러한 불변조건이 violate될 가능성이 높아지고, 이를 체크하거나 보장하기 위해 application-level locking이 필요할 수도 있습니다.
\end{enumerate}

\item {\bf 추가 QnA:} 그런데, TPC-C의 transaction 중 new order와 payment의 경우 같은 tuple에 대해 다른 column을 update하는 경우가 많은데, 이 경우에도 column level CC로는 한계가 있을까? 

ChatGPT: 
\begin{enumerate}
\item Tuple-level consistency invariant:  TPC-C에는 **고객 상태에 대한 불변조건(invariant)**이 있습니다. 예: c\_balance = 총주문액 - 결제금액

이런 조건을 유지하려면 여러 컬럼을 동시에 보는 원자적 관찰이 필요합니다.
column-level CC는 column 단위로 버전이 나뉘면 불일치된 상태를 다른 트랜잭션이 관찰할 위험이 있습니다. 즉, column 단위 MVCC를 하게 되면 consistency 보장이 어려워지고, 이를 막기 위해 결국 row-level locking 수준의 동기화가 필요해짐.

\item Write skew 또는 update-conflict detection 어려움: 
예를 들어 Payment와 New-Order가 같은 c\_balance와 c\_payment\_cnt를 동시에 변경하면, column 단위로는 write-write 충돌이 감지되지 않거나 늦게 감지될 수 있습니다. 이는 예기치 않은 **논리적 충돌(logical conflict)**이나 dirty write로 이어질 수 있습니다.

\item 복잡한 lock management \& overhead:  실무에서는 단순한 row-level lock으로도 충분히 관리되는데, column-level lock을 적용하면: Lock granule 수가 증가 → Lock table이 폭증;  Deadlock detection 및 scheduling overhead 증가; 복잡한 lock escalation 정책이 필요;

\item 스토리지 구조와의 부조화: 대부분의 OLTP 시스템은 row-store 기반이며, 단일 tuple을 페이지 단위로 I/O합니다. column 하나만 업데이트해도 결국 전체 row를 I/O해야 하므로, I/O 효율성에서 column-level isolation이 무의미해질 수 있습니다.

\end{enumerate}

\end{itemize}
}
\section{Correctness}
\label{sec:correctness}

\smallskip\noindent\textbf{Definitions.}
A history $H$ is the sequence of operations and commit or abort events of a set of transactions.
Let $x = (r, \Gamma)$ denote a \emph{conflict unit}, where $r$ is a record and $\Gamma$ is a subset of columns of $r$.
Two conflict units $(r, \Gamma_i)$ and $(r, \Gamma_j)$ \emph{conflict} iff $\Gamma_i \cap \Gamma_j \neq \emptyset$.
A read conflicts with every write on the same record.
\noindent Let $C(H)$ denote the committed transactions of $H$, and $CG(H)$ their conflict graph~\cite{Bernstein09}, with an edge $T_i \rightarrow T_j$ whenever an operation of $T_i$ precedes and conflicts with an operation of $T_j$.

A write is \emph{speculative} (denoted $SW_i[x]$) if $T_i$ stores the update in its \emph{TLA} and releases the lock before commit; the update is installed to $r$ only when $T_i$ commits.

A read $R_i[r]$ is \emph{speculative} (denoted $SR_i[r]$) if $T_i$ issues {\tt SELECT FOR UPDATE} (\ie, locking read) and reads an uncommitted value written by another transaction $T_j$.
Under \texttt{Read Committed} and \texttt{Repeatable Read}, a plain {\tt SELECT} returns the version visible under $T_i$'s snapshot and creates no dependency, as in the underlying MVCC system.
Under \texttt{Serializable}, \rcc treats every plain {\tt SELECT} as {\tt SELECT FOR UPDATE}, so read dependencies from plain {\tt SELECT}s (\eg, RW and WR dependencies) are also captured in $SG(H)$.

\noindent \textbf{Speculative Dependency Graph.}
The speculative dependency graph $SG(H)$ of a history $H$ contains a directed edge from $T_i$ to $T_j$ whenever an operation of $T_i$ conflicts with a later operation of $T_j$ while $T_i$ is in $SG(H)$.
A transaction is added to $SG(H)$ upon its first record access and removed when it aborts, or when it has committed and no predecessor of it remains in the graph.
An edge is a \emph{commit dependency} if $T_j$ reads or builds upon $T_i$'s speculative version, as in WR and WW conflicts between active transactions. Commit dependency edges also impose the constraint that $C_i$ precede $C_j$. Other edges, such as RW anti-dependencies and conflicts with committed transactions, impose no commit ordering and exist solely for cycle detection.

\rcc applies speculative versioning only to record updates, including deletes, but not to inserts.
An inserting transaction holds its lock on the inserted record until commit, so no other transaction can access it before the inserting transaction commits.
Conflicts with an inserting transaction are therefore resolved by conventional lock waiting rather than by \rcc's commit dependencies.


\noindent \textbf{Commit Protocol.}
Commitment ordering~\cite{raz:93} requires $C_i$ to precede $C_j$ for every edge $T_i \rightarrow T_j$ of $CG(H)$ and thereby guarantees serializability, but it forces a writer to wait for every reader it conflicts with.
\rcc enforces commitment ordering only for commit dependencies and guarantees serializability by cycle detection. 
\rcc enforces two invariants on every speculative history $H$.
First, for every commit dependency edge $T_i$ $\rightarrow$ $T_j$, $C_i$ must precede $C_j$, and if $T_i$ aborts, $T_j$ also aborts.
Second, the subgraph of $SG(H)$ induced by $C(H)$ is acyclic.
Before a transaction commits, \emph{ADD} checks for a cycle through it and, if one exists, aborts an uncommitted transaction on the cycle.
Every edge $T_i \rightarrow T_j$ is created by an operation of $T_j$ and is in the dependency graph before \emph{ADD} checks $T_j$.
No transaction on the cycle has been removed, since a transaction is removed only after all of its predecessors are removed, and the transactions on a cycle are predecessors of one another. 
When \emph{ADD} checks the last remaining transaction on the cycle, all transactions on the cycle are therefore still in the dependency graph and the cycle is detected.

\begin{lemma}
Under \texttt{Serializable}, every edge of $CG(H)$ that lies on a cycle is in $SG(H)$.
\end{lemma}

\begin{IEEEproof}
Let $T_i \rightarrow T_j$ be an edge of $CG(H)$.
Under \texttt{Serializable}, every read is a locking read, so two conflicting operations pass through the lock queue of the record, $T_i$ ahead of $T_j$.
If $T_i$ is active or retained at that time, the edge is recorded.
If $T_i$ has been removed from the graph, every predecessor of $T_i$ was removed earlier, and a committed transaction receives no new incoming edge.
By induction on removal order, no transaction that operates after the removal of $T_i$ has a path to $T_i$ in $CG(H)$.
Since $T_j$ operates after that removal, the edge lies on no cycle.
The same argument applies to predicate conflicts.
An update to a record that $T_i$ has read is a record conflict as above.
An insert into a range that $T_i$ has read passes through the queue of $T_i$'s range lock, so the edge is recorded under the same condition.
\end{IEEEproof}

\begin{thm}
Speculative histories are recoverable.
\end{thm}
\begin{IEEEproof}
Assume $H$ is not recoverable.
$T_j$ reads an uncommitted value written by $T_i$ and $C_j$ occurs before $C_i$.
Only a speculative read can access an uncommitted value, and such a read creates a commit dependency $T_i \rightarrow T_j$.
The commit protocol requires $C_i$ to precede $C_j$, a contradiction.
\end{IEEEproof}

\begin{thm}
Under \texttt{Serializable}, the committed transactions of $H$ ($C(H)$) are conflict-serializable.
\end{thm}

\begin{IEEEproof}
Suppose $CG(H)$ contains a cycle.
By Lemma~1, every edge on the cycle is in $SG(H)$, so the cycle lies among committed transactions in $SG(H)$, contradicting the second invariant of the commit protocol.
$CG(H)$ is therefore acyclic, and $C(H)$ is equivalent to a serial history in a topological order of $CG(H)$~\cite{BHG}.
By Theorem~2, every value read in $C(H)$ was written by a committed transaction, so this serial history is well defined.
\end{IEEEproof}

\noindent Theorems 2 and~3 follow from the commit protocol.
The timing of cycle detection does not matter, so asynchronous detection by \emph{ADD} does not affect correctness.
Under \texttt{Read Committed} and \texttt{Repeatable Read}, plain {\tt SELECT}s create no edges, so write skew~\cite{COWBOOK} remains possible as in the underlying DBMS.
\rcc thus guarantees recoverability at all isolation levels and serializability under \texttt{Serializable}.
\noindent To corroborate these guarantees, all TPC-C runs passed BenchmarkSQL's built-in consistency checker~\cite{benchmarksql}.
The reproducible scripts and logs are available in our repository~\cite{repo}.

\comment{
\begin{thm}
Under \textit{Serializable}, for every edge $T_i \rightarrow T_j$ in $CG(H)$, either $T_i$ commits before $T_j$'s conflicting operation, or $SG(H)$ contains a path from \revised{$T_i$ to $T_j$}.
\end{thm}

\begin{IEEEproof}
Under \textit{Serializable}, every read is speculative, so both conflicting operations pass through the lock queue of the record, $T_i$ ahead of $T_j$.
If $T_i$'s entry has been removed when $T_j$ is enqueued, $T_i$ has already committed.
\textcolor{red}{
If no RW anti-dependency edge from $T_i$ to $T_j$ is retained, $T_i$ commits before $T_j$'s conflicting operation.
Otherwise, the read yields an edge from $T_j$ to $T_i$ by rule ($iii$), so $SG(H)$ contains a path from $T_i$ to $T_j$.
If $T_i$ is active, a WR or WW conflict yields an edge, or a chain of edges through intervening writers, from $T_i$ to $T_j$ by commit dependency, and an RW conflict yiedls the same by RW anti-dependency.
}
Predicate conflicts arise when $T_i$ reads a range of records and $T_j$ later changes which records fall in that range.
If $T_i$ updates a record that $T_i$ has read, this is a record conflict as above.
If $T_j$ inserts a record into the range, the insert waits until $C_i$ on the range lock that $T_i$ acquired for the read, which \rcc does not release speculatively.
\end{IEEEproof}

\begin{thm}
If $SG(H)$ is acyclic, speculative histories $H$ guarantee commit ordering.
\end{thm}

\begin{IEEEproof}
Let $SG(H)$ contain a cycle: $T_i \rightarrow T_j \rightarrow T_k \rightarrow T_i$.
\revised{By the commit protocol, $C_i$ must precede $C_j$, $C_j$ must precede $C_k$, and $C_k$ must precede $C_i$, which is a contradiction.}
Therefore, if $SG(H)$ contains a cycle, $H$ cannot satisfy commit ordering.
\end{IEEEproof}

}




\section{Performance Evaluation}\label{sec:eval}


\subsection{Experimental Setup}
\label{sec:exp-setup}
All experiments are conducted on an Ubuntu 22.04 machine with an AMD EPYC 9754 CPU (128 cores at 2.25\,GHz) and 256 GB DRAM.
To isolate the effects of disk I/O~\cite{Quro}, we set the buffer large enough to keep the entire database in memory. 
Unless otherwise specified, the isolation level is \texttt{Serializable}. 
Each experiment runs for 30 seconds, with an additional 30-second warm-up for MySQL.

\noindent \module{Workloads.} We configure YCSB and TPC-C as follows.
\begin{itemize}[leftmargin=10pt, topsep=0pt, partopsep=0pt]
    \item \textbf{YCSB} We configure YCSB~\cite{ycsb} as a mixed read-update workload on a single table of 10 million records. Each transaction accesses 16 distinct records drawn from a Zipfian distribution ($\theta = 0.9$), with reads and updates in a 1:1 ratio. Each update modifies a single column, and records are accessed in ascending key order to avoid deadlocks. On MySQL, we use SysBench~\cite{sysbench}; on DBx1000, we use its built-in YCSB implementation~\cite{Bamboo-impl, rebirth-retire-code}.
    \item \textbf{TPC-C} TPC-C~\cite{tpcc} is an industry-standard OLTP benchmark. Unless otherwise specified, we run TPC-C with a single warehouse. We run \texttt{New-Order} and \texttt{Payment} transactions in a 1:1 ratio, with 1\% of \texttt{New-Order} transactions aborted at random to model user-initiated aborts. On MySQL, we use tpcc-mysql~\cite{tpccmysql}; on DBx1000, we use its built-in TPC-C implementation.
\end{itemize}

\noindent \module{Concurrency Control Protocols.}
To evaluate \rcc against existing protocols within the same system, we implement Bamboo, Rebirth-Retire, and the \rcc variants in MySQL.
The protocols evaluated on MySQL are as follows:
\begin{itemize}[leftmargin=10pt, topsep=0pt, partopsep=0pt]
    \item \textbf{Vanilla (VAN)~\cite{mysqlcode}}:  MySQL's default concurrency control, holding exclusive write locks until commit~\cite{mysqllockinfo}. 
    \item \textbf{Bamboo (BB)~\cite{BAMBOO}}: \textit{Wound-Wait} with the \textit{Retire} operation after the last write to a record.
    \item \textbf{Rebirth-Retire (RR)~\cite{rebirth_retire}}: \textit{Wound-Retire} with timestamp reassignment to reduce unnecessary aborts.
\end{itemize}
For the DBx1000 experiments (Section~\ref{sec:eval-dbx1000}), we implement \rcc and compare it with existing implementations of Wound-Wait (WW), Bamboo and Silo~\cite{Bamboo-impl}.


\begin{figure*}[t]
  \centering
  \begin{minipage}[t]{0.48\textwidth}
    \centering
    \subfloat[Throughput]{%
      \includegraphics[width=.48\linewidth]{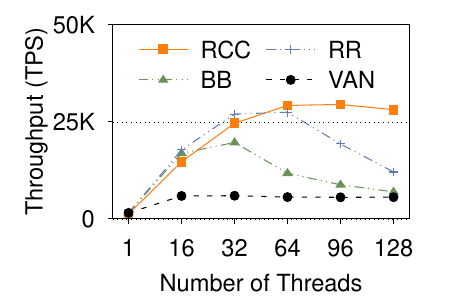}%
    }%
    \hfill
    \subfloat[Time Breakdown]{%
      \includegraphics[width=.48\linewidth]{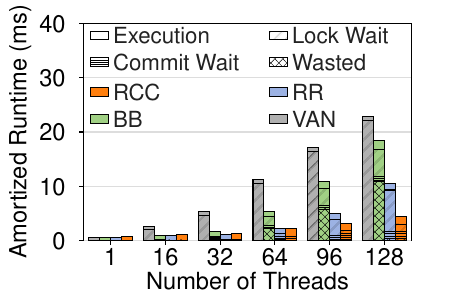}%
    }
    \caption{YCSB Benchmark with Varying Clients (MySQL)}
    \label{fig:ycsb_varying_clients_threads}
  \end{minipage}%
  \hfill
  \begin{minipage}[t]{0.48\textwidth}
    \centering
    \setcounter{subfigure}{0}
    \subfloat[Throughput]{%
      \includegraphics[width=.48\linewidth]{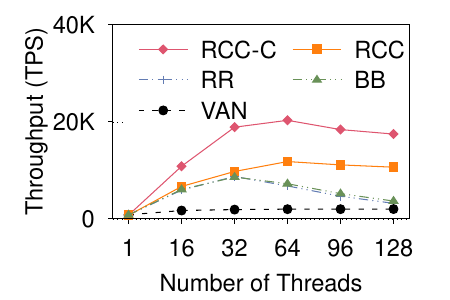}%
    }%
    \hfill
    \subfloat[Time Breakdown]{%
      \includegraphics[width=.48\linewidth]{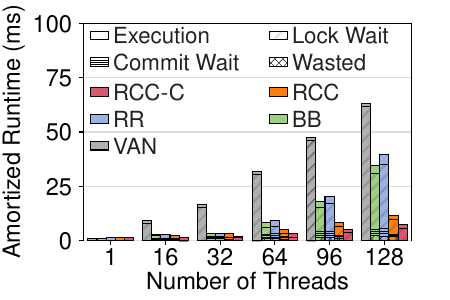}%
    }
    \caption{TPC-C Benchmark with Varying Clients (MySQL)}
    \label{fig:tpcc_varying_clients_threads}
  \end{minipage}
\end{figure*}

\comment{
    \subsection{Performance on Synthetic Benchmark}\label{sec:perf:synthetic}
    We analyze the potential benefits of \rcc in an ideal setting with a single hotspot where all concurrent transactions have the same execution time under \textit{Repeatable Read}.
    
    \module{Varying Transaction Length.} 
    To evaluate the effect of transaction length on \stx's performance, we design a synthetic workload where each transaction performs one hotspot update at the beginning, followed by $Len$ random reads.  We vary the number of threads from 1 to 128 and $Len$ from 64 to 256. \autoref{fig:synthetic_benchmark}a shows that the speedup of \stx over \vanillamy grows as $Len$ increases, reaching up to 65.9$\times$. In Vanilla, since the write lock on the hotspot is held until commit, longer transactions lead to longer WW blocking. In contrast, \stx allows subsequent transactions to access the speculative version immediately after each update, without waiting for the preceding transaction to commit. For this reason, \rcc's speedup over Vanilla scales as the number of threads increases up to 64. 
    \revised{For shorter transactions ($Len$=64, 128), the speedup saturates at higher contention(\eg, 128 threads), due to the limited inherent parallelism of the workload~\cite{BAMBOO}.}
    
    \module{Varying Hotspot Position.}
    To assess the effect of hotspot position on \stx's performance, we design a synthetic workload where each transaction performs one hotspot update and 256 random reads. $Distance$ denotes the fraction of reads placed before the hotspot. We vary $Distance$ from 0 (hotspot at the beginning) to 1 (at the end) with 128 concurrent threads. As shown in \autoref{fig:synthetic_benchmark}b, although increasing distance shortens the lock holding time, Vanilla improves throughput only at $Distance$=1. At all other positions ($Distance$ < 1), lock thrashing~\cite{lock_thrashing} prevents Vanilla from benefiting from the reduced lock duration. In contrast, \stx sustains throughput irrespective of the hotspot position, as speculative unlock releases the lock immediately after updating the hotspot.
}

\subsection{Performance on YCSB Benchmark}
\label{sec:eval-ycsb}

To compare \rcc and existing CC protocols, we measure throughput in transactions per second (TPS) and transaction runtime on YCSB while increasing the number of concurrent client threads from 1 to 128. The throughput and time breakdown are presented in \autoref{fig:ycsb_varying_clients_threads}. Each time component is normalized by the number of committed transactions.

\module{Transaction Throughput.}
\rcc outperforms Vanilla, Bamboo, and Rebirth-Retire by up to 5.3$\times$, 4.1$\times$, and 2.3$\times$, respectively.
\revised{
At low concurrency (up to 32 threads), \rcc does not consistently outperform Bamboo and Rebirth-Retire due to the cost of tracking dependencies on every record access (Section~\ref{sec:design:add}).
However, with more than 64 threads, Bamboo and Rebirth-Retire degrade as the cost of synchronous deadlock handling on the lock acquisition path grows with the number of conflicting threads, whereas \rcc, which keeps it off that path, remains robust.
}

\module{Lock Wait Time.} 
\revised{\rcc substantially reduces lock wait time (denoted by the diagonally hatched \emph{Lock Wait} bars in \autoref{fig:ycsb_varying_clients_threads}(b)) relative to Bamboo and Rebirth-Retire for two reasons.
First, \rcc performs no deadlock checks on the lock acquisition path.
In Bamboo and Rebirth-Retire, the cost of deadlock handling grows as more threads contend for the same record, causing lock wait time to grow with thread count.
For Rebirth-Retire, a longer lock list also increases the cost of the cycle detection that precedes timestamp reassignment~\cite{rebirth-retire-code}. Since this detection runs on the lock acquisition path, its cost propagates to the lock wait time of subsequent transactions.
Second, \rcc supports lightweight rollback via the \emph{TLA}: unlike Bamboo and Rebirth-Retire, which require subsequent transactions to wait for in-place undo before accessing a record~\cite{bamboo-inplace,rebirth-inplace}, \rcc allows aborted transactions to discard their \emph{TLA} entries immediately and in parallel. 
Together, these reduce lock wait time by 74\% and 82\% compared to Bamboo and Rebirth-Retire, respectively, at 128 client threads.
}
\comment{
    \revisedmsj{\rcc achieves lower lock wait time than Bamboo and Rebirth-Retire for two reasons.
    First, \rcc avoids deadlock handling overhead on the lock acquisition path by optimistic deadlock handling.
    In Bamboo and Rebirth-Retire, the cost of deadlock handling grows as more threads contend for the same record, causing lock wait time to grow with thread count.
    Under high contention, longer lock wait queues increase the cost of cycle detection for Rebirth-Retire’s timestamp reassignment~\cite{txsql}. Since this detection occurs on the lock acquisition path, the increased overhead further increases lock wait time for subsequent transactions.
    Second, \rcc supports lightweight rollback via the TLA: unlike Bamboo and Rebirth-Retire, which require subsequent transactions to wait for in-place undo before accessing a record, \rcc allows aborted transactions in parallel to discard their TLA entries immediately.
    Together, these design choices reduce lock wait time by 81\% and 86\% compared to Bamboo and Rebirth-Retire, respectively, at 128 client threads. 
    }
    \swlee{두 module 순서 바꾸었음.}
}

\module{Wasted Time.}
\revised{As shown in \autoref{fig:ycsb_varying_clients_threads}(b), as the thread count increases, Bamboo incurs significant wasted time (denoted by the cross-hatched \emph{Wasted} bars) due to false aborts and cascading aborts~\cite{rebirth_retire}. Rebirth-Retire removes false aborts by reassigning timestamps, but its rebirth operations lengthen lock wait time (\autoref{fig:ycsb_varying_clients_threads}(b)).
In contrast, \rcc avoids false aborts without deadlock checks on the lock acquisition path, since it detects deadlocks asynchronously.
\rcc thus reduces wasted time by 96\% compared to Bamboo at 128 client threads.
Compared to Rebirth-Retire, however, \rcc incurs a higher abort rate (9\% vs.\ 2\%), because deadlocks can form during speculative execution before \emph{ADD} detects them.
}

These results demonstrate that \rcc mitigates the costs of speculative write visibility more effectively than existing schemes by combining lightweight rollback with optimistic deadlock handling.
\comment{
    \revisedmsj{As thread count increases, Bamboo incurs significant wasted time due to false aborts and cascading aborts~\cite{rebirth_retire} and Rebrith-Retire also 
    Whereas Rebirth-Retire eliminates false aborts through timestamp reassignment, \rcc avoids them through optimistic deadlock handling with asynchronous deadlock detection.
    As a result, \rcc reduces wasted time by 88.5\% compared to Bamboo at 128 client threads. Compared to Rebirth-Retire, \rcc incurs a higher abort rate (10.5\% vs.\ 1.6\%) and increased wasted time, as deadlocks can form during speculative execution before \emph{ADD} detects them.
    }
}

\revised{
\subsection{Performance on TPC-C Benchmark} 
\label{sec:eval-tpcc}
To demonstrate the effectiveness of \rcc on another OLTP workload, we measure throughput and transaction runtime on TPC-C while increasing the number of concurrent client threads from 1 to 128. 
We also evaluate \rccc to assess the effect of column-level dependency tracking. The results are presented in \autoref{fig:tpcc_varying_clients_threads}. At 128 client threads, \rcc outperforms Bamboo and Rebirth-Retire by 2.9$\times$ and 3.3$\times$, respectively.

\module{Lock Wait Time.} 
At 128 threads, \rcc reduces lock wait time by 73\% and 76\% compared to Bamboo and Rebirth-Retire, respectively, by avoiding deadlock handling overhead on the lock acquisition path.
With a single warehouse, all \texttt{New-Order} and \texttt{Payment} transactions access the same \texttt{WAREHOUSE} record.
As more transactions request the lock on this record, the cost of deadlock handling increases in Bamboo and Rebirth-Retire, and the delays accumulate as subsequent transactions queue for the record, increasing lock wait time. 

\module{Commit Wait Time.} 
Commit wait time in \rcc is 60\% and 66\% lower than in Bamboo and Rebirth-Retire, respectively.
Bamboo and Rebirth-Retire make short \texttt{Payment} transactions (writers) wait for longer \texttt{New-Order} transactions (readers) to commit, which lengthens commit wait time (denoted by the horizontally hatched \emph{Commit Wait} bars in \autoref{fig:tpcc_varying_clients_threads}(b)). In contrast, \rcc retains RW edges only for cycle detection but does not require writers to wait for readers to commit.
Commit wait nevertheless remains \rcc's main residual cost under high contention: as more transactions join a dependency chain on a hot record, short transactions such as \texttt{Payment} wait longer in the pre-commit state, because commit order follows lock acquisition order. 

\module{\rccc.}
Although \rcc eliminates RW commit wait, false dependencies between transactions updating disjoint columns of the same record still cause commit wait and unnecessary aborts. 
By eliminating these false dependencies, \rccc reduces commit wait time by 88\% and wasted time by 98\% compared to \rcc, achieving 1.6$\times$ higher throughput at 128 client threads. Overall, \rccc outperforms Vanilla, Bamboo, and Rebirth-Retire by 8.6$\times$, 4.7$\times$, and 5.4$\times$, respectively, at 128 client threads.
}

\begin{figure}[t]
  \centering
  \subfloat[Throughput]{%
    \includegraphics[width=.48\linewidth]{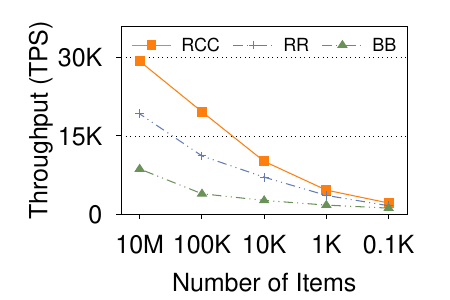}%
  }%
  \hfill
  \subfloat[Abort Rate]{%
    \includegraphics[width=.48\linewidth]{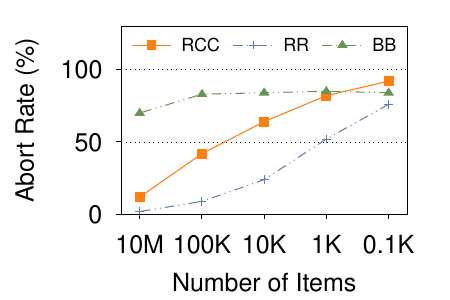}%
  }
  \caption{\revised{Varying Deadlock Probability (MySQL, 64 threads)}}
  \label{fig:deadlock_possiblity}
\end{figure}

\subsection{Performance of \rcc on Other Parameters}\label{sec:eval-other-param}

\module{Isolation Levels.}
To evaluate the effect of a lower isolation level, we repeat the \texttt{Serializable} experiments of \autoref{fig:ycsb_varying_clients_threads} and \autoref{fig:tpcc_varying_clients_threads} under \texttt{Repeatable Read}, and omit the graphs due to space limitations. At this isolation level, \rcc tracks only WW dependencies from speculative writes, while Bamboo and Rebirth-Retire still incur deadlock handling overhead on WW conflicts. At 128 client threads, \rcc outperforms Bamboo and Rebirth-Retire by 2.5$\times$ and 2.1$\times$ on YCSB, and by 1.2$\times$ and 2.3$\times$ on TPC-C, respectively. \rccc further improves TPC-C throughput by 1.3$\times$ over \rcc by eliminating false WW dependencies. These results confirm that \rcc and \rccc maintain their performance advantages at high concurrency even at a lower isolation level.

\comment{
\module{Warehouse Count.} To evaluate whether \rcc scales with database size, we vary the number of warehouses from 1 to 8 while fixing the number of concurrent clients per warehouse at 16. 
}
\comment{
    \revised{\module{Warehouse Scalability.}To evaluate \rcc's scalability under fixed per-warehouse contention, we vary the number of warehouses from 1 to 8 while fixing the number of concurrent clients per warehouse at 16.}
    To isolate WW conflicts from page latch contention, we partition \texttt{Warehouse} and \texttt{District} tables according to \texttt{Warehouse ID} for both \stx and Vanilla. \stx consistently outperforms Vanilla by up to 1.9$\times$ across all scale factors. The speedup remains stable because each warehouse introduces the same contention pattern and \rcc eliminates WW blocking uniformly across warehouses.
}
\comment{ 분량 문제로 삭제
    \module{Buffer Size.} The results presented so far assume that the entire database fits in memory. To evaluate the effect of \stx under the I/O-heavy configuration, we measure the transaction throughput while running the TPC-C benchmark with 1 warehouse and 64 clients. We set the buffer pool size to 30\%, 20\%, and 10\% of the initial data size (\ie, 250MB). \stx outperforms \vanillamy by 4.1$\times$ to 5.9$\times$ and \vanillapg by 3.8$\times$ to 4.1$\times$ across the three buffer configurations. These results indicate that the performance benefit of \stx still holds under I/O-bound conditions, although the gain slightly decreases as the buffer size shrinks and accordingly the I/O time becomes more dominant.
}

\revised{
\module{Deadlock Probability.} 
To evaluate how deadlock probability affects \rcc, we measure throughput and abort rate on YCSB with 64 client threads, fixing the Zipfian parameter ($\theta=0.9$) and varying the number of records from 10M down to 0.1K (\autoref{fig:deadlock_possiblity}).
Fewer records make concurrent transactions conflict on the same records and read each other's speculative versions, so dependency cycles, and thus deadlocks, become more likely.
As a result, throughput falls and abort rates rise for all protocols.
At 0.1K records, \rcc incurs an abort rate of 92\%, higher than Bamboo (84\%) and Rebirth-Retire (72\%), yet it still achieves higher throughput.
This is because aborts are cheap in \rcc: aborted transactions discard their speculative versions in parallel, without in-place undo that would block subsequent lock requests.}

These results show that frequent deadlocks offset the benefit of \rcc's optimistic deadlock handling, but lightweight rollback keeps abort costs low enough for \rcc to retain its throughput advantage over the other schemes.

\comment{
    \subsection{Performance on TPC-C Benchmark}\label{sec:perf:tpcc}
    To measure the performance of \rcc on the TPC-C benchmark, we use tpcc-mysql~\cite{tpccmysql} for MySQL and sysbench-tpcc~\cite{sysbenchtpcc} for PostgreSQL. In both benchmark settings, each read-modify-write sequence is configured to issue \texttt{SELECT FOR UPDATE} prior to \texttt{UPDATE}, as is standard practice for read-modify-write operations~\cite{mysql-locking-read,pg-locking-read}.
    We run \texttt{New-Order} and \texttt{Payment} transactions in a 1:1 ratio and report throughput in transactions per second (TPS). Since both MySQL and PostgreSQL exhibit similar performance trends, we present results from one DBMS per experiment for brevity.

    \module{Varying Client Threads.} 
    To evaluate the scalability of \stx, we measure throughput while increasing the number of concurrent client threads from 1 to 128 with a single warehouse. \autoref{fig:tpcc_throughput}a shows the results for PostgreSQL. As the number of clients increases, throughput in Vanilla quickly plateaus at eight threads due to WW blocking on the \texttt{Warehouse} and \texttt{District} tables. Since each warehouse has a single \texttt{Warehouse} row and ten \texttt{District} rows, concurrent updates on these records are serialized. In contrast, \stx scales up to 64 client threads, achieving up to 7.6$\times$ higher throughput than Vanilla (8.4$\times$ on MySQL). Since \stx allows subsequent transactions to access the speculative version immediately after the preceding transaction completes its update, WW-conflicting transactions are pipelined. 
    Beyond 64 clients, \rcc exhibits a modest throughput decline, as contention concentrates on a single \texttt{Warehouse} record, amplifying the lock thrashing problem in 2PL~\cite{lock_thrashing,BAMBOO}.
    
    \comment{
    In PostgreSQL, \stx exhibits lower throughput than MySQL across all configurations. This is because PostgreSQL's \textit{Repeatable Read} isolation assumes that a tuple cannot be deleted by a transaction with a smaller ID than its creator. Therefore, reading future versions can violate this assumption, causing false-positive aborts. As the number of clients increases, such aborts become more frequent, diminishing the performance gains of \stx.
    }
    
    \module{Varying Warehouse Count.} 
    To evaluate whether \rcc scales with database size, we vary the number of warehouses from 1 to 8 while fixing the number of concurrent clients per warehouse at 16. To isolate WW conflicts from page latch contention, we partition \texttt{Warehouse} and \texttt{District} tables according to \texttt{Warehouse ID} for both \stx and Vanilla. \autoref{fig:tpcc_throughput}b shows that \stx consistently outperforms Vanilla by up to 3.1$\times$ (3.3$\times$ on PostgreSQL) across all scale factors. The speedup remains stable because each warehouse introduces the same contention pattern, and thus \rcc eliminates WW blocking at each warehouse independently.
    
    \comment{
        \module{Transaction Latency} Compared to Vanilla, \stx can significantly reduce transaction latency by alleviating the wait time for WW conflicts. To confirm this, we measure transaction latency while running the TPC-C benchmark with 1 warehouse and 64 clients using MySQL. 
        \stx reduces the average and 99th percentile ($P99$) latencies of \texttt{Payment} transactions by 15$\times$ (from 60.1 to 4.0 ms) and 11$\times$ (from 78.2 to 7.0 ms), respectively. The $P99$ latency improvement is lower than the average because transactions at the tail of long dependency chains incur an extended commit wait.
        Meanwhile, for \texttt{New-Order} transactions, \stx slightly increases the latency (from 2.6 to 3.6 ms on average). This unexpected result is because increased throughput leads to more concurrent \texttt{INSERT} operations, which intensifies page latch contention.
    }
    
    \comment{
        \label{sec:perf:other}
        \begin{table}[t]
        \small
        \centering
        \caption{Transaction latency on TPC-C (MySQL, 1 warehouse, 64 clients)}
        \label{tab:tail_latency_tpcc}
        \begin{tabular}{rrrrr}
        \toprule
        Latency & \multicolumn{2}{c}{\tt New Order} & \multicolumn{2}{c}{\tt Payment} \\
        \cmidrule(lr){2-3}\cmidrule(lr){4-5}
        (ms) & Vanilla & \stx & Vanilla & \stx \\
        \midrule
        P90 & 3.5 & 4.6 & 67.7 & 5.3 \\
        P95  & 4.4 & 5.8 & 70.3 & 5.8 \\
        P99   & 5.9 & 7.9  & 78.2 & 7.0 \\
        Average & 2.6 & 3.6 & 60.1 & 4.0 \\
        \bottomrule
        \end{tabular}
        \end{table}
    }
    
    \comment{
        \begin{table}[t]
        \small
        \centering
        \caption{IOPS and TPS with varying buffer size (MySQL, 1 warehouse, 64 clients)}
        \label{tab:diff-buffer-miss}
        \begin{tabular}{rrrrr}
        \toprule
        & {\tt In-Memory } & {\tt 30\%$^\dagger$} & {\tt 20\%$^\sharp$} & {\tt 10\%$^\ast$} \\
        \midrule
        Vanilla's IOPS(4KB) & - & 2K & 3K & 5K \\
        \stx's IOPS(4KB) & - & 32K & 35K & 37K \\
        Vanilla's TPS  & 2027 & 1824 & 1528 & 1370 \\
        \stx's TPS   & 17132 & 9408  & 8146  & 5955  \\
        TPS Improvement & 8.4$\times$ & 5.2$\times$ & 5.3$\times$ & 4.3$\times$ \\
        \bottomrule
        \end{tabular}
        \small{
        \begin{tablenotes}
        \item[*] Page Size: 4KB, $^\dagger$ 75MB, $^\sharp$ 51MB, $^\ast$ 25MB
        \end{tablenotes}
        }
        \end{table}
    }
    
    \module{Buffer Size.} The results presented so far assume that the entire database fits in memory. To evaluate the effect of \stx under the I/O-heavy configuration, we measure the transaction throughput while running the TPC-C benchmark with 1 warehouse and 64 clients. We set the buffer pool size to 30\%, 20\%, and 10\% of the initial data size (\ie, 250MB). \stx outperforms \vanillamy by 4.1$\times$ to 5.9$\times$ and \vanillapg by 3.8$\times$ to 4.1$\times$ across the three buffer configurations. These results indicate that the performance benefit of \stx still holds under I/O-bound conditions, although the gain slightly decreases as the buffer size shrinks and accordingly the I/O time becomes more dominant. 
    \comment{ 추후 Revisit 필요
    \msj{위에 언급된 내용이라면, In-memory -> 30\%로 줄어도 lock holding time이 증가하여 TPS 격차가 비슷하거나 상승해야함. 하지만 Vanilla는 10\%감소, RCC는 50\% 감소되므로 우리가 주장하는 바에 대해서 말하기 어려우므로 삭제?}
    }
}


\comment{
    
    \subsection{Effect of Asynchronous Commit}
    As described in Section~\ref{sec:async-commit}, we introduce an asynchronous commit mechanism to realize the full potential of \stx. This mechanism allows a client to issue its next transaction as soon as the current transaction enters the pre-commit state, without waiting for blocking long transactions to complete.
    
    We evaluate this mechanism using a synthetic workload consisting of 15 short transactions and one long transaction, comparing vanilla MySQL with two RCC variants: one without asynchronous commit and one with asynchronous commit enabled. All transactions begin with a read-modify-write operation on the same hotspot record. Each short transaction then performs 16 random reads on non-hotspot records, while the long transaction performs a larger number of random reads, ranging from 32 to 512. As shown in Figure~\autoref{fig:async_commit_hetero}, \stx without asynchronous commit suffers a throughput decline as the length of the long transaction increases, since short transactions are forced to wait for the blocking long transaction to commit. In contrast, \stx with asynchronous commit sustains high throughput regardless of the long transaction length. This is because a client issuing short transactions can submit its next transaction immediately after the current one pre-commits, without waiting for the long transaction to complete.
    
    Figure~\autoref{fig:async_commit_tpcc} presents the effect of asynchronous commit on TPC-C. The performance gain is less pronounced than in the synthetic workload, which is expected given that the execution lengths of \texttt{NewOrder} and \texttt{Payment} transactions are similar.
    
    \comment{
    \stx replaces lock-wait with a commit-wait phase to ensure serialization order. \hjin{3.4.2의 async commit 내용 참고해서 추후 수정 필요} To mitigate overhead, we implement an asynchronous commit optimization allows that transactions to complete without commit-wait. To isolate the impact of asynchronous commit, we present results under conditions where deadlock-induced aborts are negligible.\hjin{without abort로 수정} \hjin{그럼 롤백은? -> hetero는 애초에 abort 없는 설계고 그리고 이 옵션은 durability나 consistency 보다 performance가 극도로 중요한 상황을 위한 최적화다} 
    
    We first evaluate the impact of this optimization using a heterogeneous workload \hjin{앞에서 쓴 데이터랑 똑같다 언급 필요?} consisting of 15 short transactions (1 RMW + 16 reads) and one long transactions whose read count varies from 32 to 512. As shown in \autoref{async_commit_tpcc}(a), the baseline \stx without asynchronous commit collapses as the long transaction extends: the long transaction increasingly dominates because short transactions spend more time waiting for it to commit, and this commit-wait amplification negates the benefit of early lock retirement. In contrast, \stx with asynchronous commit maintains stable throughput across all long transaction's lengths. By decoupling the client response from the background commit processing, it effectively isolates short transactions from the latency of long-running ones. 
    In the TPC-C, the performance gain is less pronounced compared to the heterogeneous scenario. Since NewOrder and Payment transactions exhibit similar latencies, the commit-wait time is relatively short\hjin{commit-wait skew?}. Nevertheless, it still contributes to sustained performance by relaxing the strict commit-wait phase. 
    }
}

\comment{
Pre-commit시 ack하고 connection에서 다음 transaction issue하게 함. (like Amazon Aurora asynch commit, eventual durability (Amazon도 결국 ED임) 처럼.) \stx의 concurrency 효과를 극대화하기 위해... 

\begin{itemize}
    \item 기대효과: t1 -- tn으로 dependency 관계에 있을 때, t1이 길고 depending tx들 (Ti) 이 짧은 경우, Ti precommit 후에 다음 Tx를 issue하면 concurrency 및 tps 향상 기대
    \item 예상 결과: PM only에서는 거의 동일하기 떄문에 효과 적음; As shown in Figure~\ref{fig:async_commit_tpcc}, NO+PM에서 긴 PM에 1/3정도 짧은 PM이 depneding할 때, TPS가 30\%(?) 추가 향상됨. 단, concurrent connection수가 적당히 작아서 CPU util이 여유가 있어야 함(?)
\end{itemize}
}


\comment{
    \revised{
    \subsection{Performance with Columnar RCC}\label{sec:perf:columnarCC}
    This section evaluates \stx with columnar CC in \textit{Repeatable Read} and \textit{Serializable} isolation using the TPC-C benchmark.
    RCC-ww은 Wound-wait을 선택한 이유를 설명하고, Like Bamboo, Vanilla대비 성능 격차에 대한 성능 결과가 Bamboo 논문과 비슷하다. Like Bamboo라고 설명하고 Bamboo라고는 하지 말기.
    \module{Repeatable Read}
    \fix{R2D3}{CC protocol 효과와 speculative versioning 효과를 분리하여 보여줄 수 있는 추가 지표 언급}, 
    \msj{20230329: Bamboo(RCC with Wound-Wait)의 경우, Older가 Younger를 abort를 하는 경우, Lock을 대기하던 모든 Younger TX abort -> active TX를 intervention하므로 대기하여 lock wait 시간 증가 (Abort하면 더 악화)-> 연쇄적으로 Amortized time 증가. 하지만 RCC의 경우, Younger abort에 대한 대기 없이 pipeline되고 pre-commit 후 Async deadlock detection을 진행 -> 대부분의 TX가 slepp 상태이므로 Fast retry 가능}
    
    \module{Serializable}
    Enforcing serializability introduces additional overhead compared to \textit{Repeatable Read}, as transactions must track RW dependencies in addition to WW ones~\cite{PWV}. We compare \vanillamy, \stx, \stx with Wound-Wait and \stx with Columnar CC. \revised{+ Thread가 늘어날 때 Wound-wait의 False abort(+Cascading abort), RCC with Columnar CC의 Deadlock(+Cascading abort) 비교}
    }
}

\comment{
    \subsection{Effects of Deadlock Detection and \rccc}\label{sec:perf:rcc_variants}
    To measure the effects of optimistic commit-time deadlock detection in \stx and Columnar \rcc, we run the TPC-C benchmark with 1 warehouse under \textit{Repeatable Read} and \textit{Serializable} isolation levels, using five different configurations for deadlock handling and \rccc:
    \begin{itemize}[leftmargin=10pt, topsep=5pt, partopsep=0pt]
     \item \textbf{Vanilla}: MySQL's exclusive write locking until commit.
     \item \textbf{\rcc}: Commit-time deadlock detection.
     \item \textbf{\rccw}: Wound-wait for deadlock prevention~\cite{BAMBOO}.
     \item \textbf{\rccd}: Deadlock detection on every lock acquisition. 
     \item \textbf{\rccc}: \rcc with column-level conflict detection (\S\ref{sec-columnar}).
    \end{itemize}
    
    To isolate the effect of \rcc's commit-time deadlock detection, we additionally implement \rcc's two variants in terms of deadlock prevention: \rccw, which avoids deadlocks via wound-wait~\cite{BAMBOO}, and \rccd, which detects deadlock synchronously upon every lock acquisition.
    
    \module{Deadlock Detection.}
    To evaluate the impact of deadlock handling on speculative write versioning, we compare \rccw, \rccd, and \stx under both isolation levels. 
    Under \textit{Repeatable Read} (\autoref{fig:sr_tps}a), \rccw outperforms \vanillamy by up to $7\times$ through speculative write versioning, comparable to the gains reported by Bamboo~\cite{BAMBOO}. 
    However, its wound-wait preemptively aborts younger transactions on every lock conflict, causing false aborts even in the absence of actual deadlock cycles. In MySQL, each rollback synchronously blocks subsequent transactions accessing the same record until it completes.
    \rccd achieves up to $4.2\times$ speedup but degrades beyond 32 threads due to latch contention on the centralized lock table, as eager deadlock detection requires frequent cycle checks. 
    \stx outperforms \rccw by $1.6\times$ by avoiding these false aborts, and \rccd by $2.1\times$ by eliminating per-conflict detection overhead. 
    Under \textit{Serializable} (\autoref{fig:sr_tps}b), since \rcc must track RW and WR conflicts as well as WW conflicts, increasing the overhead of wound-wait and eager detection relative to commit-time detection, \rcc's gain over \rccw and \rccd increases to $2.7\times$ and $2.3\times$, respectively.
    By avoiding preemptive aborts and eliminating per-conflict cycle checks, \stx reduces wasted time by $8.8\times$ over \rccw and wait time by $7.3\times$ over \rccd (\autoref{fig:sr_breakdown}a, 64 clients).

    \module{Columnar \stx (RCC-C).}
    To assess columnar \stx, we compare \stx and \rccc under \textit{Serializable}. \autoref{fig:sr_tps}b shows that \rccc improves throughput by up to $2.3\times$ over \stx. \autoref{fig:sr_breakdown}b presents the average ($avg$) and 99th percentile ($p99$) latencies at 64 clients against a \textit{Baseline} of a single transaction.
    While \rcc improves the average \texttt{Payment} latency (denoted as $avg$) by 8$\times$ over Vanilla MySQL, the latencies of \texttt{New-Order} and \texttt{Payment} converge.
    This is due to record-level dependency tracking, which creates false dependencies between transactions modifying disjoint columns.
    For example, a \texttt{Payment} updating \texttt{d\_ytd} must wait for the commit of a long-running \texttt{New-Order} updating \texttt{d\_next\_o\_id} on the same \texttt{District} record.
    Under \textit{Serializable}, RW and WR conflicts across other tables such as \texttt{Warehouse} and \texttt{Customer} further exacerbate these false dependencies, forming cycles that trigger unnecessary aborts. Thus, when moving from \textit{Repeatable Read} to \textit{Serializable}, all schemes except for \rccc in \autoref{fig:sr_tps} experience considerable throughput degradation. In contrast, \rccc exhibits significantly less throughput degradation than the other schemes, as column-level dependency tracking avoids the false dependencies that would otherwise form cycles and trigger unnecessary aborts. In addition, as shown in the last row of \autoref{fig:sr_breakdown}b, \rccc can significantly reduce the transaction latency. 
}



\comment{
To verify whether column-level dependency further improves concurrency, we compare \stx and \rccc. As shown in \autoref{fig:sr_bb_tps}, \rccc outperforms \stx by up to $2.2\times$ under \textit{Serializable}. 
Specifically, it reduces the average latency of \texttt{Payment} from 9ms to 5ms and \texttt{New-Order} from 6ms to 1ms. 
This improvement stems from resolving the execution time disparity between the two transactions; \texttt{New-Order} is inherently longer than \texttt{Payment}. 
In \stx, since both transactions update the same \texttt{District} and \texttt{Warehouse} record, the shorter \texttt{Payment} is frequently forced to wait for the completion of the longer \texttt{New-Order} or is subject to cascading aborts due to record-level conflicts. 
By leveraging the physical structure of redo logs to identify column-level offsets, \rccc enables speculative versioning that distinguishes between the different columns accessed by these transactions (e.g., \texttt{d\_next\_o\_id} and \texttt{d\_ytd} in \texttt{District}). This eliminates false dependencies on the \texttt{District} record, effectively removing the pre-commit waiting and abort overheads. \hjin{execution time 차이를 설명해줘야 할까..}
}
\comment{
To understand performance improvement in detail, we analyze the runtime breakdown in \autoref{fig:sr_breakdown}. \texttt{Payment} in \stx-col reduces wait time by x.xx\% and wasted time by x.xx\% compared to \stx, because \texttt{Payment} no longer forms false dependencies on \texttt{New-Order} transactions on the \texttt{District} record, reducing pre-commit waiting and cascading aborts.
In contrast, as shown in \autoref{fig:sr_no_time_breakdown}, wasted time of \texttt{New-Order} in \stx-col increases. Although columnar \stx eliminates false WW conflicts, the higher concurrency intensifies true WW conflicts among concurrent \texttt{New-Order} transactions on the \texttt{Stock} table, forming deadlock cycles that increase wasted time.
}

\comment{
    \module{Effect of deadlock detection} 
    Under \textit{Repeatable Read}, \stx-ww outperforms \vanillamy by up to 7$\times$ by eliminating WW blocking through speculative write versioning, which is consistent with results reported in Bamboo~\cite{BAMBOO}. \stx outperforms \stx-ww by up to 1.6$\times$, because wound-wait forces the older transaction to abort the younger transaction and wait for its synchronous rollback to complete, whereas \stx relies on deadlock detection to resolve only actual dependency cycles. Under Serializable isolation, tracking RW dependencies in addition to WW ones increases the lock conflict points where wound-wait triggers preemptive aborts, widening the throughput gap between \stx and \stx-ww from x.xx$\times$ to 2.7$\times$.
    
    \textbf{Effect of columnar \stx}
    As shown in \autoref{fig:sr_bb_tps} and \autoref{fig:sr_bb}, \stx-col outperforms \stx by up to x.xx$\times$ under \texttt{Serializable} isolation level, and reduces the average latency of both \texttt{New-Order} and \texttt{Payment} by x.xx$\times$ and x.xx$\times$, respectively.
    We attribute this to column-level dependency detection, which eliminates the false WW conflicts on the \texttt{District} record. To understand this improvement, we breakdown the runtime in \autoref{fig:sr_pm_time_breakdown}. \texttt{Payment} in \stx-col reduces commit wait time by x.xx\% and wasted time by x.xx\% compared to \stx, because the false dependency on the \texttt{New-Order} transactions no longer causes pre-commit waiting and cascading aborts. In contrast, as shown in \autoref{fig:sr_no_time_breakdown}, wasted time of \texttt{New-Order} in \stx-col increases. Although columnar \stx reduces commit wait time and cascading aborts by eliminating false conflicts, the higher concurrency intensifies true WW conflicts among concurrent \texttt{New-Order} transactions on the \texttt{Stock} table, forming deadlock cycles that increase wasted time and cascading aborts.
}

\comment{
\subsection{Performance on Sysbench Benchmark}\label{sec:perf:ycsb}
\revised{
    실험결과에 따라서 Sysbench -> TPC-C 순서 혹은 그 반대로 결정
    \module{Varying Client Threads}
    
    \module{Varying Skew Factor}

    \module{Runtime Overhead.} To quantify the runtime overhead of \rcc, we measure the time spent on speculative version construction, dependency management, cycle detection, and applying redo logs from the TLA to data pages at commit with 32 clients. These operations account for approximately 0.43\%, 0.11\%, 0.59\%, and 1.64\% of average transaction latency, respectively, totaling about 2.8\%.

}
}


\begin{figure}[!t]
  \centering
  \begin{minipage}{\linewidth}
    \centering
    \subfloat[YCSB]{%
      \includegraphics[width=.48\linewidth]{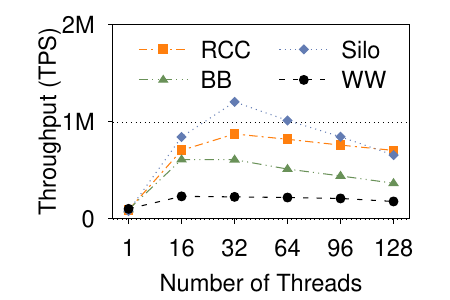}%
    }%
    \hfill
    \subfloat[TPC-C]{%
      \includegraphics[width=.48\linewidth]{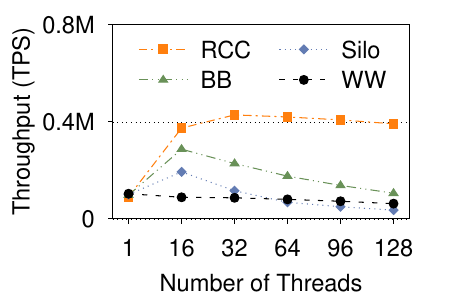}%
    }
    \caption{\revised{Throughput with Varying Clients (DBx1000)}}
    \label{fig:dbx1000_throughput}
  \end{minipage}

  \begin{minipage}{\linewidth}
    \centering
    \subfloat[Throughput]{%
      \includegraphics[width=.48\linewidth]{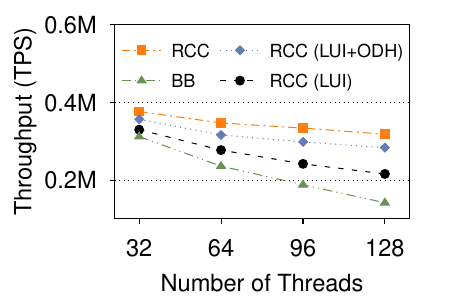}%
    }%
    \hfill
    \subfloat[Time Breakdown]{%
      \includegraphics[width=.48\linewidth]{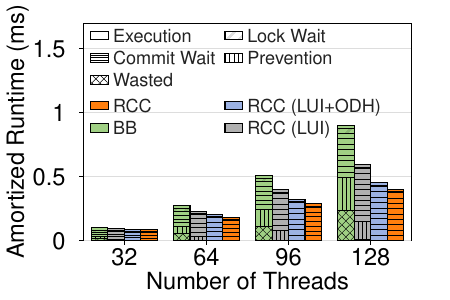}%
    }
    \caption{\revised{Ablation Study (DBx1000, YCSB, 1M records, $write\ ratio = 0.9$)}}
    \label{fig:rcc_ablation}
  \end{minipage}
\end{figure}

\subsection{Performance on DBx1000}
\label{sec:eval-dbx1000}
\module{YCSB and TPC-C.}
\revised{
To verify that \rcc's performance benefits hold in other DBMSs, we run YCSB and TPC-C on DBx1000 with 1 to 128 client threads\footnote{We do not include Rebirth-Retire in the DBx1000 comparison, as we could not reproduce its public implementation in our environment.
}.
The results are presented in \autoref{fig:dbx1000_throughput}. \rcc outperforms Wound-Wait and Bamboo on both workloads at high concurrency. At 128 client threads, \rcc achieves 3.9$\times$ and 1.9$\times$ higher throughput than Wound-Wait and Bamboo on YCSB, and 6.2$\times$ and 3.7$\times$ on TPC-C, respectively. The gap is larger on TPC-C because all transactions access the single WAREHOUSE record.  As more transactions request the lock on this record, Bamboo's deadlock handling cost on the lock acquisition path grows. Together with RW commit wait, this delays every transaction queued behind it, so the delays accumulate with contention. In contrast, \rcc keeps deadlock handling off the lock acquisition path and does not make writers wait for readers to commit, so neither cost accumulates as contention on the record grows. 

Silo~\cite{silo} benefits from optimistic execution: on YCSB, it achieves the highest throughput among all schemes up to 32 threads, even exceeding \rcc. As contention grows, however, frequent validation failures lead to high abort rates and wasted work, causing its throughput to drop beyond 64 threads and become comparable to \rcc at 128 threads. On TPC-C, Silo underperforms both Bamboo and \rcc at lower thread counts than on YCSB, since every transaction touches the single \texttt{WAREHOUSE} record: short \texttt{Payment} transactions repeatedly commit updates to it, invalidating the reads of the longer \texttt{New-Order} transactions.

These results confirm that \rcc maintains its performance advantage in an in-memory DBMS by allowing conflicting transactions to proceed speculatively while detecting deadlocks asynchronously off the lock acquisition path.
}
\comment{
    To verify that \rcc's performance benefits hold in other DBMSs, we run YCSB and TPC-C on DBx1000 with 1 to 128 client threads\footnote{We omit Rebirth-Retire from the DBx1000 comparison, as we were unable to run its publicly available implementation in our environment.}.
    The results are presented in \autoref{fig:dbx1000_throughput}.
    \rcc outperforms Wound-Wait and Bamboo on both workloads at high concurrency.
    At 128 client threads, \rcc achieves 3.93$\times$ and 1.92$\times$ higher throughput than Wound-Wait and Bamboo on YCSB, and 6.23$\times$ and 3.70$\times$ on TPC-C, respectively.
    A key difference between the two workloads is that all transactions access the single WAREHOUSE record in TPC-C.
    As more transactions request the lock on this record, deadlock handling overhead increases in Bamboo, and delays from both deadlock handling and RW commit wait can affect more subsequent transactions.
    \rcc avoids deadlock handling overhead on the lock acquisition path and eliminates RW commit wait.

    Silo~\cite{silo} benefits from optimistic execution, achieving a higher peak throughput than \rcc on YCSB at 32 threads.
    As contention grows, however, frequent validation failures lead to high abort rates and wasted time, degrading its throughput.
    At 128 threads, \rcc also outperforms Silo on YCSB.
    These results confirm that \rcc maintains its performance advantage in an in-memory DBMS by allowing conflicting transactions to proceed speculatively while detecting deadlocks asynchronously off the lock acquisition path.
}

\module{Ablation Study.}
\comment{
    To evaluate the contribution of each technique, we measure throughput and transaction runtime on DBx1000 using YCSB with 1M records, $\theta=0.9$, and a write ratio of 0.9.
    We compare Bamboo with three \rcc configurations:
    \textbf{(i) \rcc{} (LUI)} uses \emph{lazy update installation} (LUI) with Bamboo's deadlock prevention;
    \textbf{(ii) \rcc{} (LUI+ODH)} replaces deadlock prevention with optimistic deadlock handling, but writers still wait for readers to commit; and
    \textbf{(iii) \rcc} further eliminates RW commit wait.
    The throughput and time breakdown are presented in \autoref{fig:rcc_ablation}.
    Each time component is normalized by the number of committed transactions.
    Prevention denotes the time transactions spend on deadlock prevention or waiting behind it.
    This time is excluded from the Lock Wait and Commit Wait components and presented separately as Prevention.

    At 128 client threads, \rcc{} (LUI) reduces wasted time by 97.9\% and prevention overhead by 43.5\% compared to Bamboo, achieving 1.52$\times$ higher throughput.
    Through lazy update installation, \rcc{} (LUI) supports lightweight rollback, reducing abort costs and the time subsequent transactions wait for in-place undo during wound operations.
    \rcc{} (LUI) also blocks new accesses to affected records until the cascading abort completes, preventing the abort from spreading to additional transactions.
    However, prevention overhead still accounts for 25.0\% of \rcc{} (LUI)'s amortized runtime.
    By eliminating this synchronous prevention overhead, \rcc{} (LUI+ODH) achieves 1.31$\times$ higher throughput than \rcc{} (LUI).
    \rcc further eliminates RW commit wait, reducing commit wait time by 11.5\% and achieving 1.12$\times$ higher throughput than \rcc{} (LUI+ODH).
    Overall, \rcc achieves 2.24$\times$ higher throughput than Bamboo. 
}
\revised{
To analyze how \rcc's three key techniques contribute to concurrency, we conduct an ablation study on DBx1000 using YCSB with 1M records, $\theta=0.9$, and a write ratio of 0.9. 
The throughput and time breakdown are presented in \autoref{fig:rcc_ablation}. 
We compare Bamboo with three \rcc configurations that enable the techniques incrementally:
\textbf{(i) \rcc{} (LUI)} uses \emph{lazy update installation} (LUI) with Bamboo's deadlock prevention;
\textbf{(ii) \rcc{} (LUI+ODH)} replaces deadlock prevention with optimistic deadlock handling (ODH), but writers still wait for readers to commit; and
\textbf{(iii) \rcc} further eliminates RW commit wait.
Prevention overhead (denoted by the vertically hatched \emph{Prevention} in \autoref{fig:rcc_ablation}(b)) is the time committed transactions spend on deadlock prevention or waiting behind it, including delays propagated through dependency chains. 
We separate this time from \emph{Lock Wait} and \emph{Commit Wait} to isolate the cost of synchronous deadlock prevention.
In \autoref{fig:ycsb_varying_clients_threads}(b), this time is included in those components.

At 128 client threads, \rcc{} (LUI) reduces wasted time by 97.9\% and prevention overhead by 43.5\% compared to Bamboo, achieving 1.5$\times$ higher throughput.
Through lazy update installation, \rcc{} (LUI) supports lightweight rollback, reducing abort costs and the time subsequent transactions wait for in-place undo during wound operations.
As described in Section~\ref{sec:abort}, \rcc{} (LUI) blocks new accesses to a record affected by an ongoing cascading abort until the abort completes, so that no new transaction depends on versions created by aborting transactions. 
However, the synchronous timestamp-based prevention overhead still accounts for 25.0\% of \rcc{} (LUI)'s amortized runtime.
By eliminating this overhead, \rcc{} (LUI+ODH) achieves 1.3$\times$ higher throughput than \rcc{} (LUI).
\rcc further eliminates RW commit wait, reducing commit wait time by 11.5\% and achieving 1.1$\times$ higher throughput than \rcc{} (LUI+ODH).
Overall, \rcc achieves 2.2$\times$ higher throughput than Bamboo.
}

\comment{
    This performance advantage stems from reduced commit wait time.
    At 128 client threads, \rcc reduces commit wait time compared to Bamboo by 50.6\% on YCSB and 72.5\% on TPC-C.\swlee{commit wait time을 더 줄이는 이유?}\reminder{YCSB에서는 skew가 있긴 하지만, TPC-C에서는 모든 트랜잭션이 동일한 warehouse record를 접근하기 때문에 대기해야할 predecessor가 많기 때문에 더 많이 줄일 수 있게 됩니다.}
    In addition to eliminating RW commit wait, \swlee{\rcc eliminates deadlock prevention overhead from scanning the retired list in Bamboo. 어떤 의미?}\reminder{Bamboo는 deadlock prevention 하기 위해서 latch를 잡고 retired list를 순회하면서 wound할 대상이 있는지 확인해야 하는데, \rcc는 deadlock detection이라 이런 오버헤드가 없다는 의미입니다.}
    Specifically, Bamboo scans the retired list for transactions to wound on the lock acquisition path, holding the latch until the scan completes~\cite{BAMBOO}. As the number of client threads increases, longer retired lists extend the scan time and increase latch contention. This delays predecessors’ commits, increasing commit wait time for their dependents. In contrast, \rcc detects deadlocks asynchronously without scanning the retired list for deadlock prevention during lock acquisition.
}

\comment{
    To evaluate \rcc's performance on workloads beyond the synthetic and TPC-C experiments, and in particular to demonstrate its scalability advantage over competing concurrency control schemes under high contention, we run YCSB~\cite{ycsb} on MySQL and DBx1000~\cite{dbx1000} with a single table of 10 million records. Each record has a primary key and 10 columns of 100 bytes each. Each transaction accesses 16 records drawn from a Zipfian distribution ($\theta = 0.99$, $read\_ratio = 0.5$) under \textit{Serializable} isolation.
    \module{DBx-1000.} To investigate how existing concurrency control protocols behave under high contention, we compare four protocols on DBx1000~\cite{dbx1000}: Wound-Wait~\cite{wound-wait}, Silo~\cite{silo}, TicToc~\cite{tictoc}, Bamboo~\cite{BAMBOO}.
    \footnote{Rebirth-Retire~\cite{rebirth_retire} is excluded because its release implementation~\cite{rebirth-retire-code} did not produce stable results in our environment.}
    As \autoref{fig:ycsb_dbx_99} shows, all four protocols degrade beyond 32 threads, each limited by a different bottleneck.
    Wound-Wait suffers from prolonged lock holding that serializes concurrent writers.
    Silo and TicToc avoid lock waits entirely, but validation fails for nearly every transaction, with abort rates close to 100\%.
    Bamboo retires locks early but inherits wound-wait's preemptive aborts, which cascade frequently under high contention~\cite{rebirth_retire}.
    These results confirm that per-conflict runtime overhead, preemptive aborts, or synchronous deadlock handling limits the scalability of existing concurrency protocols, as the number of concurrent threads goes beyond 32. 
    
    \module{MySQL.} We evaluate \stx and its variants on MySQL using the same YCSB workload. \rccw and \rccd exhibit performance degradation beyond 16 threads, which is consistent with the TPC-C results in Section~\ref{sec:perf:rcc_variants}.
    In contrast, \stx aborts transactions only upon true deadlocks at commit time,
    keeping the abort rate low and scaling to 128 threads. This demonstrates that commit-time deadlock detection works well across different high contention workloads. Unlike TPC-C where all threads contend on a single \texttt{Warehouse} row, YCSB distributes contention across multiple hot records of the Zipfian distribution ($\theta=0.99$). This mitigates the lock thrashing observed in TPC-C at high thread counts, allowing \rcc to scale to 128 threads.

}

\comment{
    교수님 초안: While direct comparison between RCC and Bamboo-family protocols (Bamboo, Rebirth-Retire) in MySQL is challenging due to their implementation complexity in production systems, an indirect comparison can be drawn. On DBx1000 with 128 cores, both Bamboo and Rebirth-Retire exhibit throughput degradation beyond 32 threads — Bamboo due to cascading false aborts from wound-wait, and Rebirth-Retire due to the O(V+E) topological sort overhead on every lock conflict. In contrast, RCC on MySQL consistently scales to 128 threads by eliminating blocking entirely through speculative unlock and confining deadlock detection to the commit phase, achieving up to 3× higher throughput than Bamboo at 128 threads.
    }

\comment{ 
    \revised{
    
    We evaluate \stx using the Yahoo! Cloud Serving Benchmark (YCSB)~\cite{ycsb}. We use a single table with 1 million records. Each YCSB tuple has a single primary key column and 10 additional columns, each with 100 bytes of randomly generated string data. The DBMS creates a B-tree index for the primary key. Each transaction performs 8 \emph{RMW} operations on records chosen from a Zipfian distribution, where each \emph{RMW} issues a \texttt{SELECT FOR UPDATE} followed by an \texttt{UPDATE} on a single column chosen at random. Unless otherwise specified, we set $\theta = 0.99$, use 128 concurrent client threads, and set the isolation level to \emph{Serializable}.
    
    \module{Varying Clients} To evaluate the effect of \stx on other workloads than TPC-C, we compare all five configurations on YCSB while varying the number of concurrent client threads from 1 to 128. As shown in \autoref{fig:ycsb_varying_clients_threads}, \stx-col and \stx consistently outperform the other configurations, achieving up to $13.8\times$ and $5.8\times$ higher throughput than \vanillamy, respectively. The gap between \stx-col and \stx widens as concurrency increases, because higher concurrency raises the probability that concurrent transactions access disjoint columns of the same record, exposing the overhead of record-level false WW dependencies that \stx-col eliminates.
    
    \module{Varying Updatable Columns} To evaluate the effectiveness of columnar dependency detection, we restrict the candidate columns for update to $N$ out of the 10 total columns, and vary $N$ from 1 to 10. As shown in \autoref{fig:ycsb_varying_colrange}, at $N$=1, \stx-col and \stx achieve the same throughput because all WW conflicts target the same column and are true WW conflicts. As the number of candidate columns increases, concurrent transactions are more likely to modify disjoint columns, and \stx creates record-level dependency edges even for these non-conflicting accesses. \stx-col eliminates such false WW dependencies through columnar dependency detection, and the gap widens accordingly: at $N$=10, \stx-col outperforms \stx by up to $2.4\times$.
    }
}


\comment{
\subsection{Comparison with Bamboo and Rebirth}\label{sec:perf:bamboo}

    \swlee{DBX 100 2PL bamboo rebirth 128 core TPC-C TPS graph}
    
    While direct comparison between RCC and Bamboo-family
    protocols (Bamboo, Rebirth-Retire) in MySQL is
    challenging due to their implementation complexity
    in production systems, an indirect comparison
    can be drawn. On DBx1000 with 128 cores, both
    Bamboo and Rebirth-Retire exhibit throughput
    degradation beyond 32 threads — Bamboo due to
    cascading false aborts from wound-wait, and
    Rebirth-Retire due to the O(V+E) topological sort
    overhead on every lock conflict. In contrast,
    RCC on MySQL consistently scales to 128 threads
    by eliminating blocking entirely through
    speculative unlock and confining deadlock
    detection to the commit phase, achieving up to
    3× higher throughput than Bamboo at 128 threads.
}

\comment{
\subsection{Performance under Serializable Isolation}\label{sec:perf:serializable}
This section evaluates \stx under \textit{Serializable} isolation. Enforcing serializability introduces additional overhead compared to \textit{Repeatable Read}, as transactions must track RW dependencies in addition to WW ones~\cite{PWV}. 
We compare \stx against \vanillamy and Bamboo, a 2PL-based protocol that reduces contention of hotspots by violating two-phase locking. Bamboo allows transactions to retire locks after their final update, and preserves serializability through dependency tracking and cascade aborts.

To compare Bamboo and RCC under a unified codebase, we implement Bamboo’s protocol in MySQL. Our implementation includes Bamboo’s Wound–Wait policy, the lock-retire operation, and in-place updates. In addition, we implement Bamboo’s optimization that avoids aborting lower-priority writers on read-after-write conflicts by reading from local copies. We intentionally omit other optimizations and focus the evaluation strictly on the concurrency control logic. Note that Bamboo was originally evaluated on DBx1000~\cite{BAMBOO}, an in-memory DBMS. When Bamboo is applied to a disk-based DBMS such as MySQL, overheads from buffer management and latch contention can limit its performance gains~\cite{looking-glass}.


\module{Varying Client Threads}
As shown in \autoref{fig:sr_bb}a, \stx outperforms Bamboo and \vanillamy by up to 2.8$\times$ and 4.5$\times$, respectively. 
To understand why \stx still achieves higher throughput, we analyze the execution time breakdown in \autoref{fig:sr_bb}b. Both \stx and Bamboo reduce wait time caused by WW conflicts, which contributes to their performance gains over \vanillamy. Bamboo incurs significant wasted time because Bamboo's wound-wait strategy leads to false aborts (\ie~aborts that could have been avoided)~\cite{rebirth_retire}, and its in-place updates force blocked transactions to wait for the rollback process. In contrast, \stx allows immediate access to speculative versions and relies on deadlock detection to avoid false aborts. Although deadlock detection generally incurs high rollback costs~\cite{deadlock-ov}, \stx achieves lightweight abort by simply discarding speculative versions in \emph{TLA}. This combination of deadlock detection and lightweight abort significantly reduces wasted time, which is the primary factor behind the performance gap between Bamboo and \stx.

\comment{
Since Bamboo relies on explicit read dependency tracking to ensure correctness, we conduct this experiment under Serializable mode for all systems. \stx also natively supports a serializable mode by acquiring shared-locks, enabling a fair comparison against Bamboo under identical consistency guarantees.

\autoref{fig:sr_bb} illustrates the throughput in Serializable mode, where tracking shared lock dependencies inherently degrades performance across all systems compared to MVCC-based \textit{Repeatable Read}. Specially, \stx experiences increased wasted time due to cascading aborts, as every read operation establishes dependency chain. While Bamboo employs a wound-wait strategy to preempt deadlocks, this approach incurs significant overheads: it triggers frequent false aborts \cite{rebirth_retire} and blocked transactions to wait for the UNDO process of aborted transactions to preserve consistency, thereby inflating wait time. In contrast, \stx demonstrates efficiency by relying on deadlock detection to eliminate false aborts and leveraging lightweight TLA discards for rapid rollbacks, which significantly minimizes the blocking duration for concurrent transactions.
}

\module{Varying Abort Rate}
Although \stx achieves lightweight abort, cascading aborts can still amplify wasted work. To evaluate the effect of cascading aborts, 
we measure the throughput by varying the default abort rate of 1\% in the TPC-C benchmark from 0\% to 10\% for both \texttt{New-Order} and \texttt{Payment} transactions. 
As shown in \autoref{fig:abort}a, the throughput of all systems decreases as the abort rate increases. At a 10\% abort rate, throughput drops by 11\% for \vanillamy, 29\% for Bamboo, and 47\% for \stx. Nevertheless, \stx outperforms both Bamboo and \vanillamy across all abort rates. To understand this trend, we analyze the breakdown of execution time in \autoref{fig:abort}b. \vanillamy spends most of its time waiting for locks, so aborts incur little wasted work. In Bamboo, though cascading aborts still reduce throughput due to retries, Wound-Wait limits the length of dependency chains and keeps wasted time stable. In contrast, 
Although \stx experiences more cascading aborts, its lightweight abort enables fast retry and allows \stx to maintain 2$\times$ higher throughput than Bamboo even at a 10\% abort rate. Since most OLTP transactions commit successfully (\comment{as discussed in }Section~\ref{sec:oppor}), the impact of cascading aborts is expected to be limited in practice.

\fix{R2D3}{CC protocol 효과와 speculative versioning 효과를 분리하여 보여줄 수 있는 추가 지표 언급}
\hjin{Wound-Wait vs Detection}
\comment{
In \autoref{fig:abort}, we inject probabilistic commit or abort decisions into TPC-C, instead of standard 1\% rate, to evaluate the robustness of \stx. Bamboo employs a Wound-Wait and aborts lower-priority transactions immediately upon conflicts, so aborts occur early and incur relatively little wasted work. In contrast, \stx continues optimistically while tracking dependencies, which can amplify wasted time due to cascading aborts, especially when aborts occur late. Consequently, at a 10\% abort rate, throughput drops by 29\% for Bamboo, whereas \stx degrades by up to 47\%. Nevertheless, \stx maintains 2$\times$ higher throughput than Bamboo even at a 10\% abort rate, suggesting that \stx's light-weight rollback remains effective at high abort rates and helps mitigate the impact of cascading aborts.
}
}

\comment{
\subsection{Comparison with OCC}\label{sec:perf:occ}
Optimistic concurrency control (OCC)~\cite{OCC} allows transactions to execute without acquiring locks and defers conflict detection to the validation phase at commit time. If validation detects that another transaction has modified the same data, the transaction aborts.\comment{and restarts.}
We compare OCC with \vanillamy and \stx, both under \textit{Serializable} isolation.
However, implementing OCC in a disk-based DBMS requires tracking read and write sets for each transaction and performing set intersection checks during validation~\cite{OCC}.
For a fair comparison under a unified codebase,
we implement a \comment{client-side} synthetic benchmark on MySQL that emulates OCC's core phases (\textit{read}, \textit{validation}, \textit{write}).
The workload consists of 12.8M records with a timestamp column for validation, 128 of which serve as hotspot. 
Each transaction access 16 records (8 from hotspot and 8 elsewhere) without blocking, validate by checking whether timestamps have changed. We vary the number of read-modify-write ($RMW$) operation to control the contention.
On success, the transaction updates timestamps with its transaction ID and installs modifications.


\comment{Each transaction performs an atomic read-modify-write operation on 16 records. We partition the data into a small hot region (100 tuples) and a large cold region, and control contention through $P_{hot}$, where $P_{hot}$=0 means all accesses target the cold region and $P_{hot}$=1 means all accesses target the hot region.}

\comment{
To evaluate the performance of OCC without modifying the underlying database kernel, we implemented the OCC baseline \hjin{OCC-like benchmark? or Simulated OCC baseline?} as a client-side benchmark. This benchmark emulates the cored phases of mordern OCC protocols (e.g. Silo)\hjin{뭘 따라했는지 설명 필요? ex. including the read phase, validation phase, and write phase, running on top of MySQL’s Read Committed isolation.}

Each transaction executes an atomic RMW operation on a fixed set of 16 records on synthetic workloads. We explicitly regulate the workload contention level through $P_{hot}$, a parameter representing the Hostpot Access Probability. We utilize a bimodal distribution by partitioning the data into a small fixed-size hot region(consisting of 100 tuples) and a large cold region. $P_{hot}$ determines the probability that a transaction access the hot region. For example, when $P_{hot}$=0.8, 80\%  of accesses are directed exclusively hot region. This setup allows for precise control over the conflict rate to evaluate how each protocol handles contention.
}

\comment{
a transaction first reads the versions and values of records (Read Phase) and then attempts to increment the values (Commit Phase). This configuration simulates a high-contention environment where every data access results in a state modification, putting significant pressure on the locking and validation logic compared to typical read-heavy workloads.
}
\module{Varying Contention} 
OCC exhibits the largest degradation due to increasing validation failures, resulting in an 88\% throughput drop, whereas the 2PL-based protocols decline by 39\% due to waiting for record lock. \stx avoids such aborts through access-time conflict ordering while achieving higher concurrency than \vanillamy through speculative write versioning.

\module{Varying Client Threads}
In \autoref{fig:occ}b, we fix the $RMW$ ratio at 0.7 and vary the number of threads. OCC saturates at 16 threads due to validation conflicts, while \stx scales up to 64 threads. As contention increases with higher $RMW$ ratios or more threads, OCC suffers from excessive aborts due to validation failures. \stx avoids such aborts through access-time conflict ordering and achieves higher throughput by reading speculative versions.
\swlee{While OCC and its variants such as SILO eliminate lock contention, its abort ratio in the validation phase grows super-linearly with thread count, causing throughput collapse beyond 32 cores. \stx addresses this fundamental limitation by combining speculative write versioning with commit-time deadlock detection, ensuring transactions are never blocked or unnecessarily aborted, achieving consistent scalability up to 128 cores.
}
\comment{
\module{Varying Hotspot Access Probability}
\autoref{fig:occ}(a) presents the throughput as we vary $P_{hot}$ from 0.2 to 0.8 with 128 threads. In low-contention scenarios ($P_{hot}$ = 0.2), OCC outperforms both Vanilla and \stx because it incurs minimal overhead when conflicts are rare. However, as contention intensifies, OCC throughput collapses due to high validation failure rates~\cite{Brook-2PL, IC3}. \stx maintains stable throughput because it resolves conflicts at access time.
}

\comment{
\autoref{fig:occ}(a) reports the throughput as we increase the contention level by varying $P_{hot}$ from 0.2 to 0.8 with 128 threads. In low-contention scenarios ($P_{hot}$ = 0.2), OCC outperforms both Vanilla and \stx because it incur minimal overhead when conflicts are rare, as they execute optimistically without locks. However, as contention intensifies, the performance of OCC collapses drastically due to high rate of validation failures. \stx demonstrates robustness.
}

\comment{\autoref{fig:occ}(b) presents the TPC-C throughput with 1 warehouse. \msj{hj, 아래 내용 한글로 추가부탁해. 1. 실험 설정 ($P_{hot}$), 2. RCC도 성능이 떨어지는 이유 설명 필요 + OCC 대비 성능이 높은 이유, 3. clients threads Module도 앞 module과 동일하게 high contention scenario에서 OCC보다 \stx가 더 낫다는 것을 보여주기 위함인지?}}
}


\comment{

Below are descriptions of the three workloads used in the experiments:
\module{Orderline table partitioning} connection 수가 많아졌을 때, OL 테이블의 activie page에 과도한 tuple insertion 때문에, bottleneck이 되어서, 이를 완화하기 위해, OL table을 Wh-id 기준으로 partition했음. 그렇지 않은 경우, Vanilla보다 성능이 낮아져서 1달 간 고생했음 :(
뮌헨서 베를린 가는 기차 안에서 해결... -> 현재 MySQL 8.4는 OL partitioning하지 않은 버전으로 진행

\module{Warehouse, District table Partitioning} MySQL 8.0부터 Lock Mutex(Sharding, 5.7에서는 global mutex 변수 하나로 사용) 및 Page Latch 등을 page_id 단위로 설정하고 있음. 이 경우 다른 레코드임에도 작은 테이블일 경우 한 페이지에 있게 되어 latch 경합이 심해짐. 그 중 페이지 접근이 많고 레코드가 적은 WH(w_id), DT(w_id, d_id)로 파티셔닝 적용

\subsection{Benchmark Issue}
\module{Common}
\begin{itemize}
    \item New-order TRX: Stock update의 경우, item ID가 중복되어서 (비복원 + Non-uniform) 동일한 TRX 내에서 동일한 Stock Record가 두 번 업데이트 될 수 있음. {\bf 따라서 New-order TRX에서는 모든 Write가 Final Write가 아님.}
    \item Delievery TRX: SELECT min(no o id) FROM new orders 에서 과거 값을 읽어 다음 DELETE FROM new orders에서 X LOCK 대기로 인한 병목, 256conn 이상에서 성능이 하락하는 주된 요인 -> 실험은 NO, PM만 진행
\end{itemize}
\module{Postgres}
\begin{itemize}
    \item New-order TR
        \begin{itemize}
            \item Batch로 X-Lock을 잡아버림으로써 ShadowTX에서는 Deadlock이 발생할 확률이 높아짐. 따라서 ShadowTX에 맞게 Batch S4U + Update를 S4U + Update Loop로 변경
            \item S4U Stock -> Insert Order Line -> Update Stock에서 S4U Stock -> Update Stock -> Insert Order Line로 변경
            \item Deadlock 상황을 줄이기 위해서 ItemID를 오름차순으로 정렬 (percona tpc-c도 동일하게 구현되어있음)
        \end{itemize}
    \comment{
        Batch S4U + Batch Update: Benchbase TPC-C 벤치마크의 New-order TRX에서 Stock 테이블을 업데이트할 때, Batch로 Lock을 잡고, Batch로 Update를 수행함. (Percona TPC-C MySQL의 경우 Lock -> Update를 반복적으로 수행함.) 예를 들어, 10개의 Record를 업데이트해야한다면, Benchbase TPC-C는 10개의 Record에 대해서 S4U(X-Lock)을 전부 잡아두고 Batch로 10개의 Record를 업데이트하고, Percona TPC-C는 1. Record에 대해서 S4U, 2. Record Update를 반복적/1으로 10번 수행함. Benchbase TPC-C의 경우, Batch로 X-Lock을 잡아버림으로써 ShadowTX에서는 Deadlock이 발생할 확률이 높아짐. 따라서 ShadowTX의 성능을 높이기 위해서는 Batch S4U + Update를 S4U + Update Loop로 바꾸는게 중요함.
    }
    \item Payment TRX
        \begin{itemize}
            \item Percona 에서는 바닐라의 성능 개선을 위해 Update WH, DT -> Select 순으로 쿼리 진행(다른 칼럼을 접근하기 때문에 과거 값을 읽어도 괜찮기 때문). STX에서는 이렇게 하면 Select에서 과거 버전 또는 미래 버전을 읽기 위한 오버헤드가 크기 때문에 성능이 감소하여 S4U->Update로 수정하였음. 
        \end{itemize}
    
\end{itemize}

\module{in-memory vs. disk-based version에서 \stx 효과}

\subsection{Basic Performance Analysis with TPC-C}
To compare the cost-effectiveness of three systems – DRAM-only {\tt Vanilla} (5GB DRAM), {\tt NV-SQL(2:1)}, and {\tt NV-PPL(2:1)} (2GB DRAM and 1GB NVDIMM) – we measure TPS and relevant metrics while running the TPC-C benchmark on each system.

\module{Transaction throughput and latency}

\module{Runtime Analysis} 
like Bamboo 논문, \stx lock wait 시간, commit wait 시간 (commit dependecy 처리 시간..), abort 시간? stakced graph...

\subsection{Effect of \stx on Other Parameters}\label{sec:other-param-eval}

\module{Effect of the number of Warehouses}  Question: 8 conn. per 1 warehouse 동시 접속 트랜잭션 수행 vs.  64 conn per 8 warehouses: 동시성 관점에서 각 warehouse별로 비슷하게 충돌할까? 아니면, 64개가 서로 충돌을 일으켜서 동시성이 훨씬 저하될까?  Answer: TPC-C에서 Warehouse \#와 Client \# 관계는 concurrency conflict 패턴에 큰 영향을 줍니다.  TPC-C에서 Warehouse는 일종의 데이터 파티션입니다 대부분 트랜잭션은 특정 warehouse ID (W\_ID)에 국한되어 수행됩니다. {\bf 예외적으로, New-Order 트랜잭션은} 일부 확률(약 10\%)로 다른 warehouse의 stock을 참조하면서 cross-warehouse 트랜잭션을 발생시킵니다.

{\bf NewOrder - TX cross-warehouse stock order: TPC-C new order에서 10개 제품중 하나는(즉, 10\%) 다른 warehouse 제고를 주문하도록 설계되어 있다네.} 따라서, wh를 늘리면, 서로 다른 wh간 tx들이 서로 충돌을 유발하게 되어서, \stx 입장에서 dependency length가 늘어남. 

https://chatgpt.com/share/683a8fa5-ddf0-8005-82fb-a6e741097c89

이런 cross-warehouse item 주문 일반적? : 아주 핵심을 찌르는 질문입니다. 결론부터 말하자면: 일반적인 OLTP 애플리케이션에서는 TPC-C처럼 cross-warehouse item 주문이 흔하지 않습니다. 하지만 예외나 유사한 구조가 있는 특정 도메인에서는 유사한 패턴이 일어날 수는 있습니다. (위 내용 참고)

따라서, cross-wh new order 제거하고 실험; 또는 every 10th new order tx에 대해 cross-wh 추가할 경우 성능 평가 포함. 

\module{Effect of SSD Performance} I/O 수반시 \stx trade-off - In-memory의 경우, cpu와 blocking이 tps에 영향을 미치는 주요 factor임. 그러면, buf size가 작아서 i/o 발생시 (즉, Tx별 buf mis에 따른 io latency가 도입될 때, \stx 효과 확인 필요 
\reminder{design 장에 장점 간단 언급?}
\begin{itemize}
\item  (장점) vanilla는 blocking T1이 miss에 따른 IO 유발시 T2는 더 긴 시간 대기 (CC와 cpu util 더 낮아짐.)  blocking time 증대 Vs. Stx는 no blocking 및 그에 따라 tx latency 감소 및 io 병렬성 증대
\item  SSD slow write 때문에 read stall 증대에 따른 read I/O 개선 limit될수 있음; 이를 숨기기 위해  (SSD write 성능이 bottleneck이 안되도록) high end ssd 사용? SSD 성능에 따른 \stx 개선 정도 
\item  동일 ssd에서도 buffer size별 개선 정도 및 분석 제시
\item \stx: t2 계속 진행해서 동시성 높아짐. 하지만, 한 conn서 t2 pre-commit후 다음 tx를 issue하지 않으면 \stx 효과 별로 없을 수도... 
\end{itemize}

{\bf 주요 예상 결과: cpu/lock bound 대신 io bound가 더 dominant해질 때}
\begin{itemize}
    \item buffer size가 줄어서 상대적으로 io가 더 많이 유발될 때 (즉, 모든 tx가 read miss 시) 격차 경향? buf size 늘어서 tx별 buf miss가 줄어들때 경향? 
    \item SSD iops가 빨라질수로 \stx 성능 향상 증가 
\end{itemize}

To verify the effect of \stx on different SSDs, we measure the transaction throughput while running the TPC-C benchmark using two MySQL versions on three other commercial SSDs.
The results are summarized in \autoref{tab:diff-ssd}.
The first two rows of the table illustrate varying random read and write IOPS capabilities across the three SSDs tested. Across all SSDs, \stx consistently outperforms Vanilla, though the degree of improvement varies depending on the SSD. Specifically, for {\tt SSD-A}, which shares a vendor and similar internal architecture with the SSD, \stx achieves a performance improvement over Vanilla of 6.9$\times$.

Meanwhile, on {\tt SSD-B} and {\tt SSD-C}, \stx outperforms Vanilla by more than 40 times.
Both SSDs exhibit lower random write IOPS than {\tt SSD-A}, resulting in Vanilla achieving only 22 TPS or fewer (approximately one-tenth of {\tt SSD-A}'s TPS).
Under the \stx scheme, which transforms the write-heavy I/O pattern into a read-heavy one, the random read performance of SSDs becomes more critical than the random write performance, provided that a certain level of random write performance is maintained.

\begin{table}[t]
\small
\centering
\caption{Effect of \stx on Different SSDs (TPC-C)}
\label{tab:diff-ssd}
\begin{tabular}{rrrr}
\toprule
& {\tt SSD-A$^\dagger$} & {\tt SSD-B$^\sharp$} & {\tt SSD-C$^\ast$} \\
\midrule
Read IOPS(4KB) & 370K & 300K & 95K \\
Write IOPS(4KB) & 500K& 100K & 90K \\
Vanilla (TPS)  & 174 & 22  & 17 \\
\stx (TPS)   & 1205  & 896  & 776  \\
TPS Improvement  & 6.9$\times$ & 40.7$\times$ & 45.6$\times$ \\
Write Reduction Ratio &77\% & 86\% & 89\% \\
\bottomrule
\end{tabular}
\small{
\begin{tablenotes}
\item[*] $^\dagger$ Samsung 970 Pro NVMe SSD 512GB, $^\sharp$ Samsung SM951 256GB, $^\ast$ Micron Crucial MX500 250GB
\end{tablenotes}
}
\end{table}

\module{Effect of Log Device Performance}  One benefit of \stx is that with faster log device, tx logic 수행이 TPS bottleneck이 되므로, \stx 장점 극대화(?); how to simulate the log device performance, NVDRAM? ( Log devices: HDD, SSD, NVRAM )
\begin{itemize}
    \item 특히, NVRAM의 경우, T1가 commit record를 쓰는 순간, 모든 depending TXs cascade abort 필요성 사라짐. 
    \item Do we need to limit the number of depending TXs before T1 commits (i.e., before its commit record is made durable in log device), so as to limit the penalty of cascaded abort? If T1 commits, remove the limitations? cascaded group commit?
\end{itemize}

\module{CPU Core 수에 따른 \stx 효과} 1WH 128 connectsions에서 core수 증가에 따른 vanilla vs. \stx TPS graph

\module{1024 core에서 WH 수 및 connection 수에 따른 성능 차이}
\comment{

\textcolor{red}{\module{Transaction throughput}} Figure~\ref{fig:throughput-128}: 
(Figure~\ref{fig:throughput-128} 128 core 기준 1 wh 32 connection에서 성능 폭이 커서, wh 증가 실험에서도 32conn/wh로 실험 (1~4); 4wh 이상에서는 CPU utilization으로 tps saturated!

Figure~\ref{fig:throughput-32}: 32 core 기준 - 128 core보다 성능 향상폭이 consistently 큼 because the transaction execution time gets longer, thus the effect of pipelining by \stx increases and concurrency increases. 
\reminder{혜진 - 128 vs. 32 core machines에서 단일 NO/PM 수행시간 측정 제시.} 
\hjin{128core에서 NO(2ms), PM(0.5ms) 소요}

In the remainder, all the experiments are conducted using 128 cores.  

\textcolor{red}{{\bf Effect of Concurrency}} \reminder{아래 commented 부분 포함? 성재 작성?} 
\module{Runtime Overhead of \stx} connection수가 extreme하게 늘어날 수록, stx에서 도입한 자료구조와 dependency tracking/deadlock detection algorithm에 의해 부담이 생김. 이러한 overhead가 얼마나 되고 언제쯤 들어나는 지 investigate하기 위해, 1warehouse 128 connection에서 어떤 어떤 값들을 measure/present했음. overhead가 marginal??

\textcolor{red}{\module{Effect of Transaction Length}} The main reason why \stx outperforms vanilla is due to the pipelining. To check whether the effect of pipilining of conflicting transactions scales with longer transactions, we run the TPC-C only using PM. extend PM's length by calling the {\em sleep()} function with 0.0005 sec and 0.001prior to commit. 0.0005 is the double of single PM execution at 128 cores.  

\begin{itemize}
    \item Figure~\ref{fig:tx-length}: PM only로 PM의 길이를 {\em sleep()} 함수를 끝에 추가해서 TX 길이를 길게했을 때 pipeline에 따른 concurrency의 증가 정도
    \item As expected, with longer PM, \stx outperform vanilla. 12x and  18x, respectively. 
\end{itemize}

\textcolor{red}{\module{Effect of Hotspot position}} \reminder{혜진: Bamboo 논문 관련 section/subsection 참고해서.. 관련 문구 copy해두고, 가상의 결과를 Table or Figure로 일단 추가 부탁} 

트랜잭션 간 충돌이 발생하는 지점을 hotspot으로 정의함. transaction 길이 수정한 것과 유사하게 PM을 이용해서, WH/DISTRICT hotspot conflict 위치를 수정해서 .... following condition에 대해 \stx의 효과를 검증하였음.
1) Hotspot position이 트랜잭션 앞부분에 존재할 때 vs. 끝에 존재할 때
2) Hotspot 개수 - 단일 트랜잭션 내에서 다른 트랜잭션과 충돌이 발생하는 지점의 개수를 1개, 2개, 3개 ... 증가시킬 때 성능 (아니면 2개로 제한하고?)

\textcolor{red}{\module{Effect of Logging Performance}}
\reminder{MySQL 8.0의 경우, logging 기능을 OFF 할 수 있음! NO-logging 버전에서 성능평가 필요!}
\hjin{1wh/64conn에서 실험해본 결과 TPS가 van: 2033->2121, stx: 16411->17727, 성능: 8.1x->8.4x 로 증가하였습니다. (바닐라의 경우 mysql이 제공하는 logging off 기능을 이용, stx에서는 이 config로 off하면 shadowing이 적용되지 않아서 하드코딩하여 redo log를 쓰는 부분을 주석처리 하여 간단히 확인해보았습니다. 제대로 구현하면 stx의 성능이 좀 더 올라갈 것 같습니다)}
One benefit of \stx is that with faster log device, tx logic 수행이 TPS bottleneck이 되므로, \stx 장점 극대화(?); how to simulate the log device performance, NVDRAM? ( Log devices: HDD, SSD, NVRAM )

\textcolor{red}{\module{Effect of Hybrid Deadlock Detection or Adaptive \stx}} Cascade abort에 따른 폐해를 줄이기 위해 dependency chain이 길어지면 adaptive하게 동작.. 및 그 효과. 
}

}
\section{Related Work}\label{sec:related}

\comment{\stx improves concurrency by maintaining multiple speculative write versions per object and enabling early write visibility  under the optimistic assumption that most transactions commit. In addition, it guarantees correctness using lock-based conflict ordering at record access time. In this regard, \rcc is closely related to four categories of concurrency schemes: optimistic concurrency control~\cite{OCC,hekaton,hyper}, relaxed pessimistic concurrency control~\cite{ELR,CLV,Altruistic,BAMBOO}, early write visibility through analysis and chopping~\cite{PWV,chopping,IC3,MCC}, and multi-version concurrency control~\cite{mvcc-eval,mv2pl-Carey:92,mv2pl-chen:82}.
\rcc is closely related to four categories of concurrency schemes: optimistic concurrency control, relaxed pessimistic concurrency control, early write visibility through analysis and chopping, and multi-version concurrency control. 
}
\comment{
vs. Optimistic concurrency control (OCC)~\cite{OCC} - speculative write version but no access-time conflict ordering, leading to excessive aborts and repeated retries

vs. Bamboo~\cite{BAMBOO} - early write visibility via program analysis but immediate update, thus incurring cascade abort. repeatable read 시까지 lock hold (local read copy로 해결?)vs. \stx: repeatable read 문제 없음. skip in-place updates at normal time by group commit. and no undo overhead at cascade abort.
}

\module{Optimistic Concurrency Control.}
Assuming that conflicts are rare, optimistic concurrency control (OCC)~\cite{OCC} allows transactions to execute without blocking, validating against committed versions only at commit time. Silo~\cite{silo} optimizes OCC for multi-core systems by decentralizing the validation phase. TicToc~\cite{tictoc} improves scalability through data-driven timestamps. 
Like OCC, \rcc keeps transaction-private versions during execution and wastes work on abort. Unlike OCC, \rcc avoids excessive abort rates under high contention~\cite{Brook-2PL,IC3} through lock-based access-time conflict ordering.

\comment{
\module{Relaxed Pessimistic Concurrency Control.} Pessimistic concurrency control (PCC) schemes, such as strict two-phase locking (S2PL) protocol~\cite{COWBOOK}, guarantee serializability but limit the concurrency due to the strict locking protocol. Numerous techniques have been proposed to improve concurrency in PCC by relaxing the conservative locking protocol duration, while preserving the correctness guarantee. Based on the observation that pre-committed transactions can safely expose their updates while waiting for log flush~\cite{Partial-Strictness}, early lock release~\cite{ELR} and controlled lock violation~\cite{CLV} allow transactions to release locks before commit. 
However, since ongoing updates remain hidden, conflicting transactions are still serialized.
In contrast, in \rcc, each transaction speculatively releases its lock immediately after updating a record, pipelining subsequent transactions on the record.
}

\module{Early Write Visibility.} 
To further improve concurrency by advancing write visibility prior to transaction commit, several techniques decompose transactions into smaller pieces to avoid conflicts~\cite{chopping,IC3,MCC,PWV}. Transaction Chopping~\cite{chopping} decomposes transactions into sub-transactions to enable early update visibility, but requires static analysis of transaction access patterns.
Runtime Pipelining~\cite{MCC} and IC3~\cite{IC3} relax chopping's SC-cycle constraints through runtime enforcement, allowing transactions to read uncommitted data at the risk of cascading aborts.
PWV~\cite{PWV} avoids cascading aborts via early commit points and deterministic execution, while requiring 
transaction decomposition. In contrast, \stx does not require static analysis or transaction decomposition.
\comment{
To further improve concurrency by advancing the write visibility prior to transaction commit, several techniques decompose transactions into smaller pieces to avoid the conflicts~\cite{chopping,IC3,MCC}. Transaction Chopping~\cite{chopping} decompose transactions into sub-transactions to enable early visibility of updates. A valid chopping must satisfy two conditions: rollback-safe, where any rollback statement must lie in the first subtransaction, and absence of SC-cycles in the chopping graph. Runtime Pipelining~\cite{MCC} allows concurrency beyond SC-cycle constraints by enforcing a fixed table access order and permits dirty reads. IC3~\cite{IC3} achieves fine-grained chopping through column-level static analysis and enforces piece execution order at runtime. IC3 also allows dirty reads and merges deadlock-prone SC-cycles into larger atomic pieces to prevent deadlocks. These approaches require full knowledge of workload access pattern and impose restrictions to avoid deadlocks. Piece-wise Visibility (PWV)~\cite{PWV} extends pipelining to handle transactions with unknown access sets but still requires access information. Among these techniques, Runtime Pipelining and IC3 incur cascading aborts when a predecessor aborts because they allow dirty reads, while PWV avoids cascading aborts via early commit points that make writes visible only after all abortable operations have finished. In contrast, based on the optimistic assumption that transactions will commit and update data objects just once, \stx does not require either static analysis or transaction decomposition. In case the assumption does not hold, it can efficiently abort the depending transactions because of its out-of-place update. 
}

\comment{ 
Bamboo~\cite{BAMBOO} allows transactions to retire locks after their final write to a record, enabling subsequent transactions to read uncommitted data during execution. In this regard, Bamboo is closest to \stx in that it provides speculative write visibility and supports cascade aborts through dependency tracking among conflicting transactions. However, to achieve high performance, Bamboo performs control and data flow analysis to determine the last write position. In addition, Bamboo operates on single-version storage, where transactions update the current version in place, incurring high undo overhead upon cascade aborts. To mitigate this overhead, Bamboo adopts wound-wait, which can cause false aborts~\cite{rebirth_retire}. In contrast, \stx confines speculative updates to redo logs, enabling lightweight rollback by simply discarding speculative versions.\comment{without touching data pages.} This low abort overhead allows \stx to employ deadlock detection instead of prevention, reducing false aborts (Section~\ref{sec:perf:serializable}).
}

Bamboo~\cite{BAMBOO} extends Wound-Wait to enable subsequent transactions to read uncommitted data by retiring locks before commit, but its deadlock prevention can cause unnecessary aborts. Rebirth-Retire~\cite{rebirth_retire} mitigates such aborts by reassigning transaction timestamps and detecting deadlocks upon conflicting lock acquisitions. 
However, \rcc differs from those speculative write visibility schemes in four ways: lightweight rollback via lazy update installation, \revised{optimistic deadlock handling}, elimination of RW commit wait, and columnar concurrency control through redo log deltas. 

\revised{Recently, TXSQL~\cite{txsql} and Brook-2PL~\cite{Brook-2PL} have been proposed to reduce hotspot contention under locking, and each must also decide how to handle the deadlocks its approach creates or exposes. TXSQL lets followers update a hotspot row serially without locking and resolves deadlocks with timeouts and proactive rollback, but it supports only one hotspot row per transaction. Brook-2PL avoids deadlocks by statically analyzing transaction types and releasing locks early, which requires the transaction set to be known in advance. \rcc requires neither.}

\module{Multi-Version Concurrency Control.}
Existing MVCC systems maintain multiple committed versions of the same logical object using undo logs to address the read-write conflicts among concurrent transactions, increasing concurrency~\cite{mvcc-eval,mv2pl-Carey:92,mv2pl-chen:82}. In contrast, \stx leverages redo logs as speculative version storage, enabling concurrent updates while guaranteeing that serialization order follows lock acquisition order.


\comment{PSSI maintains a dependency serialization graph and performs precise cycle detection at commit time. RCC likewise detects actual cycles, but maintains its own access-time-ordered dependency graph and performs detection asynchronously.}

\module{Serializable Snapshot Isolation.} Snapshot isolation can admit non-serializable execution such as write skew due to cycles involving RW anti-dependencies. SSI~\cite{SI:TODS09} guarantees serializability by detecting dangerous structures formed by RW anti-dependencies, but its conservative mechanism may cause unnecessary aborts. PSSI~\cite{PSSI} instead maintains a dependency serialization graph and performs precise cycle detection at commit time, reducing unnecessary aborts. 
Like PSSI, \rcc uses cycle detection to guarantee serializability at commit time, but performs it asynchronously off the execution path, reducing the commit latency caused by cycle detection.
\comment{
    Redo log as future version, Lazy Update, Select-For-Update 문, (Final Write) 사용 WW 해결하는 점에서 차별점!!!
    
    SSI, PG SSI 등에서 RW-antidependency 개념, SIREAD Lock 등 개념 기반; 관련해서, PG SSI 논문 포함 비교(?)
    
    \begin{itemize}
        \item 2PL, SI/MV, OCC와 비교: 기존 기법 대비 \stx에서 허용하는 장점들; false deadlock  측면 등. 
        \item Partial strictness in 2PL~\cite{PS-in-2PL}: 'commit request' op를 도입해서, tx logic 다 수행한 경우, 다른 TX가 update해도  무방함을 보임; final R/W hint의 도움을 통해 logic 수행 안으로 시점을 더 앞당김; 또한, \stx는 자신의 모든 update를 redo log로 남기기 때문에, final write 끝난 객체를 읽어도 OK이고, 그 update도 redo log로 남기기 때문에, 다른 tx의 repeatable read도 방해하지 않음.  cascade rollback도 쉽게(어떤 의미? need to refine) 달성가능함. 
        \item O2PL서 차용한 commit ordering 논문 
        \item TicToc~\cite{tictoc}: time travel storage engine 과의 차별성 언급; 
        기존 occ에서 사용하는 global ts가 아니라 데이터 기반 ts (튜플 레벨에서 read/write ts 별도 관리), lazy timestamp allocation 으로 확장성 개선
        \stx는 database system에서 사용하고 있는 undo/redo log를 활용하여 트랜잭션 레벨에서 consistent한 데이터 접근 가능;
    \end{itemize}
    
    \module{Concurrency Control scheme} 
    CC scheme 비교; Locking / OCC Validation / MV 카테고라이즈화해서 
    
    (vs. Field call and OCC) FC or OCC는 current version 값을 동시에 읽고 다 동시 진행 후, 나중 commit시에 validation 진행; 실패 시 abort (hot object에 대해 abort 과도한 것으로 알려져 있음.); VS. STX는 Lock 기반 최신 값에 대해 serialized하게 읽어 진행해서 validation 불필요; Only when depending TXs abort, T is cascade-rollbacked. 
    (유사점) FC처럼, precommit시 redo log LSN 할당, lazy update, lazy commit 등 유사한 점 많음. (포함?) FC는 repeatable read 미지원 vs. \stx는 기존 read view construction으로 repeatable read 지원
    
    \module{vs. OCC:} A crucial distinction from OCC lies in how RLV manages concurrent updates. OCC allows transactions to update the same object without any coordination and relies on commit-time validation to detect conflicts. This leads to wasted work and frequent aborts under contention. By contrast, Redo Logs as Versions (RLV) requires each updating transaction to briefly acquire a write lock on the latest version of the object. However, unlike traditional locking, the lock is held only for the duration of the update and is immediately released thereafter. This short-duration locking prevents destructive write–write conflicts upfront while avoiding long blocking times. As a result, RLV achieves two goals simultaneously: 1) no commit-time validation is necessary, since conflicts are already prevented at update time, and 2) high concurrency is preserved, as locks are held only briefly and never until commit.
    
    \begin{itemize}
        \item  OCC와 \stx가 udpate를 current version 대상으로 하지 않고 local copy나 redo log에 temporary하게 저장하는 측면에서 유사점
        \item  OCC에서는 여러 update를 동시 허용하고 나중에 validation함. vs. \stx는 SFU lock 획득 transaction들이 순차적으로 controlled된 방식으로 수행하기 때문에, 특히 hot spot의 경우 retry가 훨씬 줄게 됨(?)
        \item  논리적으로 current version W lock을 획득한 하나의 TX만 update 허용하는 controlled 방식 --> OCC와 차이
        \item memory overhead 입장에서 Redo log를 활용하기 때문에, OCC private record memory 불필요(?)(또는 유사?)
        \item page local copy vs. local redo log; 무조건적 parallel execution에 의해 복잡한 validation 및 과도한 rollback vs. controlled WW 허용과 상대적으로 단순한 group commit 
        \item  page local copy vs. local redo log; 무조건적 parallel execution에 의해 복잡한 validation 및 과도한 rollback vs. controlled WW 허용과 상대적으로 단순한 group commit 
        \item OCC는 특정시점 t0의 current version 바탕으로 N개 update 동시 진행. VS. \stx에서는 future version 을 controlled way로 읽어 진행
        \item OCC는 충돌이 안 날걸로 보고 validation을 하는데, WW/WR/RW 관점에서 충돌하지만, final read/write 여부 확인하고, commit 할 것으로 보고, no cascade abort 부담없이 진행하는 기법. OCC는 timestamp로 하기 때문에, false conflict 존재; \stx는 no false!!
    \end{itemize}
    
    (GPT: While our approach delays in-place updates, it fundamentally differs from optimistic concurrency control (OCC). OCC treats updates as speculative and resolves conflicts through validation at commit time, aborting transactions that violate serializability. In contrast, we treat redo logs as explicit version storage, where redo records represent future versions that are managed for visibility and conflict resolution throughout transaction execution.
    
    Rather than validating conflicts retrospectively, our mechanism proactively resolves update conflicts by allowing concurrent writers to operate on separate future versions without forcing serialization on the latest state. This decoupling of version creation from physical updates enables conflict resolution without relying on speculative execution or late-stage aborts.)
    
    \module{TRX reordering} Trx 내 쿼리 순서를 reordering하더라도 WW blocking을 완전히 해결하는 것은 불가능하지만, \stx 구조를 활용한다면 WW Blocking 문제를 해결가능하다는 장점이 있다.
    따라서 Final Write가 확인되면 후행트랜잭션이 계속 업데이트가 가능해서 병렬성이 더 높아짐.
    
    \module{MVOCC} Hekaton, HYRISE, MemSQL, Hyper MVOCC 차용;
    read, write phase 2단계로 구분됨; 
    두개 이상의 트랜잭션이 동일한 데이터 쓰기 시도 후, 한 트랜잭션이라도 먼저 commit 하면 다른 트랜잭션들은 validation phase에서 자신이 읽은 데이터 (read set)
    write-write conflit가 발생함;
    \stx는 유사하게 optimistic한 transaction processing을 수행하지만, 
    redo log를 별도로 남기기 때문에 WW conflict를 피할 수 있다는 점;
    While MVOCC abort \reminder{conservatively? --> blindly?} and thus waste CPU cycles (useful work), \stx can resolve most WW conflict in a controlled way and thus save CPU cycles.

    \module{Bamboo~\cite{BAMBOO}} 
    \begin{itemize}
        \item hotspot에 집중(warehouse, district); 이외에도 \stx는 해결하는 문제?
        \item 가장 큰 차이: future log 기반 
        \item single version vs. multiversion:  최신 single version에 모두 update함. 따라서, T1, ..., Tn의 update overwrite함.; \stx서는 최신 tuple 대신 future version store에 update 유지한 후, 모두 commit후에 redo log를 tuple에 최종 반영; 따라서, abort, +? 등이 light-weight해짐. 
        \item 충돌 난 레코드에 대해 no undo 할 필요가 없다 (future log만 삭제)
        \item With Bamboo, {\bf cascading Abort overhead 때문에}, 업데이트가 끝나더라도 후행 트랜잭션이 sequential 하게 업데이트를 진행하지 못함. VS \stx의 경우, MySQL의 경우 Redo Log를 메모리에 별도로 유지하고, PostgreSQL의 경우 Tuple레벨로 MVCC를 진행하기 때문에 Cascading Abort로 인한 overhead 최소화.
        \item pre-commit / group commit 개념?; abort/recovery 차이점? 
        \item columnar CC: 칼럼 레벨로 dependency가 없어, cascaded abort 필요하지 않는 경우도 있음. 
        \item Bamboo는 deadlock-free (deadlock prevention <- old tx 우선 순위, tx timestamp 기반으로 commit ordering) VS> 우리는 deadlock 허용/resolution (자세히 확인 필요) 
        \item vs. Lightweight cascaded abort; Bamboo서는 관련 tuple에 대해 lock을 지고 있어야 하는데, 우리는 abort가 빨라서 lock hold time이 작다(혜진). 
    \end{itemize}
    
    \subsection{TicToc} Andy Pavlo 비교; storage 관점 version storage를 time travel 형태로 관리; vs. \stx는 transaction lavel의 time travel로서  ww conflict 해결 
    
    \subsection{Controlled Lock Violation} 
    
    Early Lock Release:
    \begin{itemize}
        \item `` In essence, one transaction must not publish or persist another transaction’s update until the update is durable''(from controlled lock violation paper)
        \item [4] `` The principal idea is to remove the locks from the  lock manager early but to retain “tags” to convey commit dependencies among transactions.'' -> \stx와 유사? 
        \item Altruistic Locking~\cite{altruistic} 유사점 설명 필요
    \end{itemize}

    Controlled Lock Violation~\cite{CLV}:
    \begin{itemize}
        \item Single Version 가정-yes. CLV paper related work section 마지막; read only tx 문제 등 다루지만 별시리... 
        \item It eschews the early lock release and the commit dependency management. 
        \item  With respect to concurrency, controlled lock  violation matches early lock release as it permits subsequent  transactions to violate or ignore lock conflicts, but only in very  controlled situations. More specifically, a subsequent transaction  may acquire a lock that violates an existing lock if the transaction  holding the lock is already in its commit process, i.e., \underline{if it has  allocated a commit record in the log buffer}. Thus, controlled lock  violation enables concurrency in all situations in which early lock  release (corrected with tags) enables concurrency.
        \item Commit 시간 해결에 주안점 (flash, nvram log device로 이 부분 문제 완화됨) vs. application logic시 WW lock에 의한 waiting 방지 
        \item T0's commit log 가 log buffer에 쓰여졌다는 것은, 결국 commit하는 것이기 때문에, crash 만 아니면 성공적으로 수행했다는 보증. \stx에서도 이점 활용 가능!! T0 commit log 쓸 때, depending but already precommitted Ti들에 대해서도 Commit log를 같이 쓰도 됨 (group commit))
    \end{itemize}
    
    Controlled Lock Violation:
    \begin{itemize}
    \item t0가 lock을 끝까지 holding;  
    \item for long t0 vs. short t1, t1 should wait long? 
    \item 두 시스템에서 commit dependency 개념/구현 차이?
    \item group commit, precommit 개념 차별화?
    \item group commit commit dependency vs. precommit한 tx
    \item Single version으로 current copy에 WW overwrite함. 
    \item read-only tx 이슈
    \item what about WW conflict? For ``canned transactions'' coded as stored procedures, static code analysis can enable controlled lock violation even {\bf before a transaction’s commit request.} - {\bf The same is true with \stx resorting to final write!!!} -- 이 경우, 둘 다 cascaded rollback 필요 -- 어쩌면 여기에 더 가까움. 이때 두 기법의 차이가 극명; redo log만 discard하면 됨. 
    \item CLV - eager apply vs. \stx - lazy apply at (pre)commit --> less cascade abort overhead!! 
    \item In CLV, {\bf long T0 \& short T1 경우}, T1 cannot commit until T0 LSN high water mark is made durable. 하지만, \stx의 경우, T1는 precommit하고, 해당 session에서 다음 TX 시작 가능하고, T0 commit시 T1도 group commit 가능; 이 측면에서 `Rethink the sync' (external synchrony)에 더 충실함. 
    \end{itemize}
    
    Field Call/Escrow TX 유사성/차이점 
    
    \module{Piece-Wise Visibility (PWV)~\cite{BHG}:} 
    \begin{itemize}
    \item high transaction throughput from delayed to early write visibility~\cite{PWV}: To enable early write visibility, PWV decomposes transactions into a set of sub-transactions or pieces, such that each piece consists of one or more transaction statements. PWV then schedules pieces such that their corresponding transactions execute in a serializable order. PWV makes a piece’s writes visible as soon as its transaction commits, even if one or more pieces of the same transaction have not yet executed. 
    \item WW/WR conflict 관점 write visibility가 CC의 significant impediment임을 보임. 
    \item 아주 제한된 상황에서 abort하는 경우, 해당 logic만 수행하고 난 후에 write visibility를 early 허용함. (그럼에도 불구하고 abort할 수 도 있지 않나??) VSVS. final write에 대해 early write visility 측면에서 유사함. 하지만, \stx는 아무런 제한 없이 (또는 final write 이후) final write visibility를 허용함. 
    \item PWV는 deterministic CC 가정, deadlock 등에 의한 abort 없는 것으로 가정. 
    \item TX간 commit dependency는 어떻게 유지?? 차이점 기술 필요.
    \item Single Versioning 가정? serializable isolation protocols such as 2PL and OCC 가정 
    \end{itemize}
    
    \subsection{Serializable Snapshot Isolation} 
    \begin{itemize}
        \item S2PL 방식
        \item PostgreSQL 방식
        \item PSSI
    \end{itemize}

    \comment{
    \module{CTR} MS Azure SQL Database has recently introduced a novel database recovery mechanism, called Constant Time Recovery (CTR), which combines AIRIES recovery with multi-version concurrency control to achieve  \comment{leverages the row versions generated for MVCC to support} database recovery in constant time~\cite{CTR}. 이는 \textbf{CC 기능을 활용해서 Recovery acceleration하는 최초 방법임}. {\bf 반면, \stx}는 \textbf{recovery용 Redo buffer를 (multi-)version store로 활용함으로써, concurrency control acceleration에 활용}하는 최초의 방법임. 
    \module{Aurora} Log is database. --> \stx는 Log can play the role of (future, not past) Version Store!
    }
    
    \comment{
    In that the durable log in \nvipl is leveraged/helpful to 1) reduce writes to the storage,  2) reduce reduce stall and enable parallel I/Os, 3) increase concurrency by treating log as 가상의 쓰기 and thus reduce WW conflicts, 4) enable instant recovery, there are several closely related works. 
    Recent advances in flash memory SSD technology have increased the I/O bandwidth and throughput of storage devices based on flash memory significantly, and it has been demonstrated that this new storage medium can be used to improve the throughput of large scale OLTP systems by orders of magnitude with much less energy consumption~\cite{Intel09,LM09}. This impressive development has been realized by various research and development endeavors for better hardware and software such as multi-channel architecture, over-provisioned capacity, large DRAM buffer, flash translation layers, wear leveling algorithms, and so forth. In particular, flash translation layers play a significant role in overcoming the limitations of flash memory such as erase-before-write and the limited number of erase cycles~\cite{DFTL,intel:ftl,FAST07}.
    Nonetheless, if flash memory SSDs are used as a simple replacement of disk drives, a database system may not be able to fully exploit the advantages of flash memory nor fully overcome its disadvantages. Rather, as is demonstrated in previous work~\cite{LM07,Tsirogiannis09}, some of the key components of database server design may well be revisited so that they become {\em aware} of the distinct characteristics of flash memory. The \tipl scheme presented in this paper has been inspired by the same spirit and developed as a flash-aware design for transaction
    management.
    }
    
    \comment{
    Not much work has been reported in the literature about transactional support for flash memory based storage and database systems so far. In this section, we briefly review research results related to this work in the area of transactional file system, multiversion-based recovery, and more traditional recovery techniques.
    \parag{Transactional IPL} No wal, undo-less recovery 등은 \tipl~\cite{icde11:tipl}에서 빌려옴. 하지만, flash-based \tipl은 -memory log sector의 durability를 위해 IO가 필요하고, 따라서, logging 및 recovery {\em three phase commit} and {\em page-check}와 같은 new log type을 도입할 필요가 있음. 또한, redo-only logging을 필요함. 반면, in-NVDIMM log는 log의 instant durability로 인해 이런 단점이 전형 없음. 
    \parag{IPL-P~\cite{vldb11:iplp}, SQLite/PPL~\cite{PPL:2015} 등 - 1) tranditional wal log기반 rec x
    2) multi-versioning 지원} 
    }
    
    \comment{ 
    \parag{Transactional Flash File Systems}
    There have been a few recent attempts to provide transactional support
    for file systems, for example, {\em transactional flash}~\cite{Prabhakaran08}
    and {\em Light-weight Time shift FTL (LTFTL)}~\cite{LTFTL08}. 
    In order to support various useful features such as atomic writes, 
    rollback, undelete, and time travel operations in a file system,
    they exploit the existence of multiple versions of a page, which are produced
    by out-of-place updates in flash memory, much like what we do in \tipl. 
    In addition, unlike pure FTL-based approach, they assume more functional
    interfaces other than the usual read and write in order to provide
    tighter interaction between the file system and flash memory. 
    Unlike \tipl, however, these approaches write a whole page as a unit of
    update propagation instead of writing change logs,
    and thus a force commit policy is mandatory for them to ensure
    atomicity and consistency.
    The main drawback of the approaches is the storage overhead for maintaining
    old version pages and the overhead of garbage collection
    (performed in the background) to clean up old version pages. 
    }
    
    \comment{
    \parag{Multiversion Based Database Recovery}
    The Constant Time Recovery Technique in MS's Azure SQL Database~\cite{CTR}: {\bf Similarity} It demonstrates how to leverage multiple versions in reducing recovery time. (undo-less recovery는 tIPL insight borrow?) MV (in-row + off-row Versions) 활용 redo와 undo 빠르게 해서 constant time recovery 달성 - 이 점에서 \nvipl 유사;  
    {\bf Difference:} write reduction, \stx로써 version store 활용 못함.  
    Several multiversion based database recovery techniques had been proposed
    in the early 1980s such as the {\em atomic actions} by Reed~\cite{Reed83} 
    and its variants~\cite{Bober92,Chan82}.
    Some of the recent multiversion based snapshot isolation
    approaches\cite{Berenson95,Oracle95} can be traced back to these pieces of work.
    The existence of multiple versions makes it easier to
    implement multiversion concurrency control and
    guarantee the atomicity of transactions.
    Like the transactional flash file systems, the multiversion based recovery approaches use a whole page as a unit of version and adopt a force policy for commit protocol. In addition, these approaches are often criticized for its coarse lock granularity, space overhead from  multiple versions, and write bandwidth~\cite{Gray78}.  This criticism is also shared by the transactional flash file systems mentioned above.  It is interesting to note that the transactional flash file systems and the multiversion based database recovery techniques attempt to achieve a similar goal at different layers of the system hierarchy.
    }
    
    \comment{
    \parag{Other Recovery Techniques}
    Besides the popular ARIES-style database recovery, several important recovery techniques have been proposed. We review two of them that are closely related to \tipl.
    The shadow page mechanism has been proposed as a straightforward way of achieving database recovery~\cite{Lorie77}. The main advantage is that neither undo nor redo recovery is necessary under this approach. To the best of our knowledge, there is at least one commercial database product marketed for embedded systems with flash memory, called Polyhedra FlashLite~\cite{FlashLite09}, that uses this approach. However, the shadow paging mechanism has not been widely adopted by major database systems, because of its limited scalability due to garbage collection overhead, storage space overhead, heavy fragmentation of a table object, and inflexible concurrency control~\cite{Bernstein87}. The same criticism will be still valid even when the shadow paging is used for flash memory. 
    There is another recovery strategy that has been proposed based on {\em undo at read time} approach~\cite{Rappaport75}. Under this recovery method, undo is not done in the recovery phase, but is done at read time after restart. On fetching a data item, it is checked whether the data item is created by a successfully committed transaction. If not, it is ignored and then its previous version is retrieved. The {\em no undo} recovery by \tipl is quite similar to this approach. To some extent, this recovery method similar to the deferred rollback approaches~\cite{Mohan93,Oracle01}.
    }
}

\section{Conclusion}

In this paper, we observe that write-write conflicts limit concurrency severely, and argue that the long-standing practice of in-place updates to records is the primary culprit.
We then propose \rcc, to our knowledge, the first approach to leverage redo logs as speculative versions. \rcc aims to improve transaction latency and throughput by pipelining conflicting transactions, not serializing them, while achieving serializability through lock-based access-time ordering and dependency management. 
To further exploit redo log-based speculative versioning, \rcc introduces two techniques: optimistic deadlock handling with asynchronous detection and columnar concurrency control. 
We show that \stx integrates seamlessly with existing DBMSs, delivering significant concurrency improvements and maintaining high throughput at 128 concurrent threads --- well beyond the degradation points of existing speculative write versioning systems.

\section*{Acknowledgment}
The authors used GPT-6 Astra and Claude Opus 5.5 to assist with language editing, writing, spelling checks, figure refinement, implementation details, and debugging during code development. All AI-assisted content was carefully reviewed and revised by the authors, who take full responsibility for the paper. The core research ideas, methodology, experimental design, analysis, and conclusions are the authors’ own.


\bibliographystyle{IEEEtran}
\bibliography{ref}
\newpage
\appendix
\section{Appendix}\label{appendix}
This appendix presents the implementation details of \rcc on \revised{MySQL,} PostgreSQL and DBx1000, followed by experimental results on PostgreSQL.

\subsection{Implementation}
\revised{\rcc leverages redo logs as speculative version storage in N2O and materializes speculative versions as tuples in the version chain in O2N. Both the MySQL and PostgreSQL implementations modify the lock manager and transaction management modules, with approximately 3K and 2K lines of code changes, respectively.}
Our PostgreSQL implementation supports \textit{Repeatable Read} and tracks only commit dependencies; extending \rcc to PostgreSQL's SSI~\cite{SSIpostgres} is left as future work.
For DBx1000, we support \textit{Serializable} and use the same configuration without redo logging as Wound-Wait and Bamboo for a fair comparison.

\subsubsection{Speculative Version Construction}

\leavevmode\par\module{MySQL.}
\stx maintains speculative versions as redo log entries to store uncommitted updates in the \emph{TLA}. When a transaction accesses its predecessor’s speculative version, it first copies the version to its own TLA. For an update, it then modifies the copy and stores the redo log in its TLA. \stx initially allocates 4KB of heap memory for each transaction's \emph{TLA} and dynamically allocates additional memory as needed. In our TPC-C experiments, each speculative version consumes about 300 bytes on average, comparable to the corresponding record size.  

\module{PostgreSQL.}
\rcc creates speculative tuples following the same procedure as \vanillapg. 
However, when updating a record, \rcc follows the version chain to the latest speculative version without waiting for its writer to commit.
To read the speculative version and create a new speculative tuple version, \stx utilizes three existing fields in the tuple header: \texttt{xmin}, \texttt{xmax}, and \texttt{t\_ctid}, which stores the pointer to the next version in the chain. 
It checks the status of the transactions identified by \texttt{xmin} and \texttt{xmax} via \texttt{ProcArray} and \texttt{pg\_xact} (transaction status log)~\cite{postgres-xact}.

\module{DBx1000.}
\rcc stores speculative versions as full tuple copies in each transaction's local memory.
To access these versions, \rcc uses each record's \texttt{retired list}, which retains lock entries after speculative unlock~\cite{BAMBOO}.
When updating a record, a transaction uses this list to find the preceding writer's speculative version.
If a preceding writer exists, the transaction copies that writer's speculative version; otherwise, it copies the committed record.
The transaction then updates its local copy, leaving the original record unchanged until commit.

\subsubsection{Dependency Management}\label{appendix:dep-man}

\leavevmode\par\module{MySQL.}
\rcc needs to track dependencies among transactions to guarantee serializability.
In our implementation, a transaction enqueues a lock request in the lock table (\texttt{lock\_sys->rec\_hash}) and creates the dependencies derived from the corresponding record access in \emph{ADD}'s graph. 
A transaction is removed from the graph when it aborts, or when it has committed and none of its predecessors remains.
While a committed transaction remains, its record accesses stay visible to later conflicting accesses, which record the corresponding edges.
This retention suffices because a committed transaction receives no new incoming edge, so a transaction with no remaining predecessor can no longer lie on a cycle.
For columnar dependency tracking, \rcc adds a {\tt field\_bitmap} field to {\tt lock\_t}. A dependency edge is created only when the bitmaps of two conflicting locks overlap~\cite{PCLL}.

\module{PostgreSQL.}
Unlike MySQL, where lock acquisition order in the central lock table serves as a dependency graph, PostgreSQL stores row locks in tuple headers~\cite{postgres-tuplock}, providing no such structure. To track commit dependencies, \stx extends PostgreSQL's transaction manager with a bidirectional dependency graph. 
\stx maintains transaction IDs in the \texttt{pred\_list} in each transaction's local memory and the \texttt{dep\_list} in shared memory using a partitioned hash table to reduce mutex overhead. 
When a transaction creates a speculative version or acquires a row lock, \stx extracts the predecessor ID from the tuple's \texttt{xmin}, adds it to its \texttt{pred\_list}, and inserts itself into the predecessor's \texttt{dep\_list}. 

\module{DBx1000.}
During record access, \rcc identifies predecessors through the record's \texttt{retired list} and stores their IDs and dependency types (WW, WR, or RW) in the accessing transaction's local array.
When new dependencies are added to these arrays, \emph{ADD} reads the entries from all active transactions and adds the corresponding edges to its graph.
\emph{ADD} maintains graph nodes separately from transaction objects so that dependencies can be retained after commit.
As in MySQL, a node is deleted if the transaction has aborted, or if it has committed and has no retained predecessors.

\subsubsection{Commit}

\leavevmode\par\module{MySQL.}
A transaction can commit once all its commit dependencies are resolved and \emph{ADD} has cleared it. If \emph{ADD} finds no cycle in the edges involving the transaction, it sets the transaction's {\tt cleared} flag. Since \emph{ADD} performs cycle detection asynchronously off the critical path, the transaction itself need not traverse the dependency graph. Once the transaction is cleared to commit, its redo logs accumulated in the \emph{TLA} are parsed using \vanillamy's recovery function. Following the WAL protocol, \rcc then first writes undo logs, applies the redo logs in the \emph{TLA} to data pages, and finally appends them to the WAL buffer. \rcc then releases its locks and resolves the commit dependencies of its dependents, waking up any pre-committed dependents.
Note, however, that its node and edges must still be retained in the dependency graph until all predecessor nodes have been removed, ensuring that any cycle involving its edges can still be detected by \emph{ADD}. 

\module{PostgreSQL.}
\stx checks the status of all predecessors in the \texttt{pred\_list} and waits for them to commit.
\emph{ADD} performs a depth-first search on the \texttt{dep\_list} asynchronously off the execution path, so the transaction does not traverse the dependency graph itself.
Since WAL records are already written when speculative versions are stored in heap pages, no additional logging beyond the normal commit procedure is required.

\module{DBx1000.}
\rcc maintains a counter for each transaction to track its unresolved commit dependencies.
At commit, the transaction waits until this counter reaches zero and \emph{ADD} has approved it.
For each updated record, \rcc acquires the record's latch and copies the speculative version from transaction-local memory to the original record.
Before releasing the latch, \rcc removes its lock entry and decrements successors' dependency counters.
Since redo logging is disabled, no log flush is required at commit.

\subsubsection{Abort}

\leavevmode\par\module{MySQL.}
An abort affects only dependent transactions that have read the aborting transaction's speculative versions. For cascading aborts, \rcc traverses the lock table starting from the aborting transaction and marks each dependent for abort. In doing so, any transaction attempting to acquire a lock must wait until the cascading abort completes.

\module{PostgreSQL.}
When a transaction aborts, \vanillapg marks it as aborted in \texttt{pg\_xact}, making its tuples invisible until vacuum reclaims them~\cite{Stonebraker87}. 
\rcc leverages this mechanism. 
A transaction aborts itself if it has been marked as a victim by the background deadlock detector or detects a predecessor abort by checking \texttt{pg\_xact}.
Aborted speculative versions remain as dead tuples and can increase storage usage, but our TPC-C experiment shows only 1.7\% additional storage overhead, as the vacuum process continuously reclaims them in the background.

\module{DBx1000.}
When a transaction aborts, \rcc discards its speculative versions, leaving the original records unchanged.
For each record updated by the aborted transaction, \rcc acquires the record's latch and traverses the \texttt{retired list} from that transaction's lock entry to identify dependents.
\rcc marks these dependents for abort, including any dependent lock owner.
It then removes the aborted transaction's entry and its dependents' entries from this record before releasing the latch.
Each marked transaction discards its speculative versions and marks dependents on other records it updated for abort.
Subsequent transactions must wait until the aborted transactions' lock entries are removed before accessing the record.
Separately, the aborting transaction sends an abort event to \emph{ADD}, which removes its node and dependency edges from the graph.

\begin{figure}[t]
  \centering
  \begin{minipage}[b]{.48\linewidth}
    \centering
    \includegraphics[width=\linewidth]{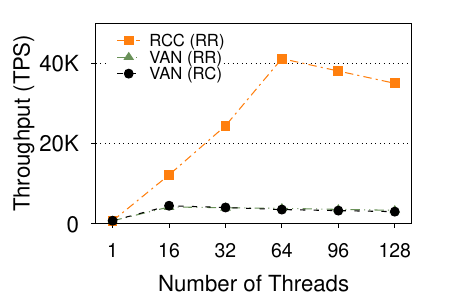}\\
    {\footnotesize (a) YCSB}
  \end{minipage}%
  \hfill
  \begin{minipage}[b]{.48\linewidth}
    \centering
    \includegraphics[width=\linewidth]{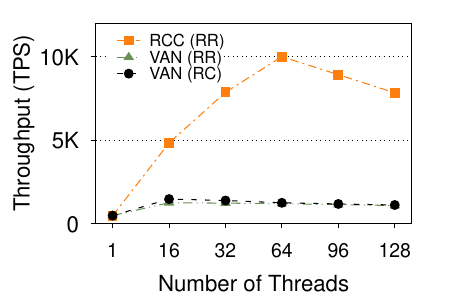}\\
    {\footnotesize (b) TPC-C}
  \end{minipage}
  \caption{PostgreSQL: Throughput with Varying Clients}
  \label{fig:pg_throughput}
\end{figure}

\subsection{Experimental Results on PostgreSQL}
To verify that \rcc's performance benefits extend to PostgreSQL, we measure throughput on YCSB and TPC-C while increasing the number of concurrent client threads from 1 to 128.
We compare \rcc under \textit{Repeatable Read} with Vanilla PostgreSQL under \textit{Read Committed} and \textit{Repeatable Read}.
The results are presented in \autoref{fig:pg_throughput}.

At 64 client threads, \rcc outperforms Vanilla PostgreSQL under \textit{Repeatable Read} by 10.80$\times$ on YCSB and 8.36$\times$ on TPC-C, while also outperforming Vanilla PostgreSQL under \textit{Read Committed} on both workloads.
By allowing subsequent transactions to access speculative versions without waiting for preceding writers to commit, \rcc reduces WW blocking.
These results confirm that \rcc also improves throughput in PostgreSQL's O2N MVCC architecture.


\end{document}